\documentclass[reqno,12pt]{article}
\usepackage{varioref} 
\usepackage{xr} 
\usepackage[colorlinks]{hyperref}
\usepackage{booktabs}
\usepackage{subcaption}
\usepackage{multirow}
\RequirePackage[OT1]{fontenc}
\RequirePackage{amsthm,amsmath}

\usepackage{epsfig,amssymb,euscript,amsfonts,hyperref}
\usepackage{natbib}
\usepackage{array}
\usepackage{rotating}
\usepackage{nccbbb}
\usepackage{verbatim}
\usepackage[lined,boxed,commentsnumbered, ruled,linesnumbered]{algorithm2e} 
\usepackage{graphicx}
\usepackage{float}
\usepackage{bm}
\usepackage{bbm}
\usepackage{dsfont}
\usepackage{MnSymbol}
\usepackage{enumitem}
\usepackage{color}
\usepackage{caption,setspace}
\usepackage{lipsum}
\usepackage{xcolor}
\usepackage{mathbbol}
\usepackage{marginnote}
\usepackage{forest}
\usepackage{authblk}

\numberwithin{equation}{section}

\newcolumntype{L}[1]{>{\raggedright\let\newline\\\arraybackslash\hspace{0pt}}m{#1}}
\newcolumntype{C}[1]{>{\centering\let\newline\\\arraybackslash\hspace{0pt}}m{#1}}
\newcolumntype{R}[1]{>{\raggedleft\let\newline\\\arraybackslash\hspace{0pt}}m{#1}}
\newcolumntype{H}{>{\setbox0=\hbox\bgroup}c<{\egroup}@{}}

\def\spacingset#1{\renewcommand{\baselinestretch}
{#1}\small\normalsize} \spacingset{1}

\usepackage{marginnote}

\newtheorem{theorem}{Theorem}[section]
\newtheorem{lemma}[theorem]{Lemma}
\newtheorem{corollary}[theorem]{Corollary}
\newtheorem{proposition}[theorem]{Proposition}

\theoremstyle{definition}
\newtheorem{example}{Example}[section]
\newtheorem{assumption}{Assumption}[section]
\newtheorem{remark}{Remark}[section]
\newtheorem{definition}{Definition}[section]

\hypersetup{citecolor=blue}
\hypersetup{linkcolor=red}

\definecolor{BUred}{rgb}{0.8, 0.0, 0.0}
\definecolor{BLUE}{rgb}{0.0, 0.28, 0.67}
\definecolor{magenta}{rgb}{1,0,.6}
\definecolor{blue0}{rgb}{0.95, 0.97, 1.0}
\definecolor{brown}{rgb}{0.75, 0.25, 0.0}

\usepackage{color}
	\definecolor{brown}{rgb}{0.75, 0.25, 0.0}
	\definecolor{mygreen}{rgb}{0.1, 0.5, 0.1}	
	\definecolor{myred}{rgb}{0.7, 0.1, 0.1}
	\definecolor{myblue}{rgb}{0.1, 0.1, 0.7}
	\definecolor{mygold}{rgb}{0.95, 0.85, 0.37}
	\definecolor{mygray}{rgb}{0.5,0.5,0.5}
	\definecolor{myorange}{rgb}{1,0.5,0}

\usepackage{mathtools}
\makeatletter
\DeclareRobustCommand\widecheck[1]{{\mathpalette\@widecheck{#1}}}
\def\@widecheck#1#2{
    \setbox\z@\hbox{\m@th$#1#2$}
    \setbox\tw@\hbox{\m@th$#1
       \widehat{
          \vrule\@width\z@\@height\ht\z@
          \vrule\@height\z@\@width\wd\z@}$}
    \dp\tw@-\ht\z@
    \@tempdima\ht\z@ \advance\@tempdima2\ht\tw@ \divide\@tempdima\thr@@
    \setbox\tw@\hbox{
       \raise\@tempdima\hbox{\scalebox{1}[-1]{\lower\@tempdima\box
\tw@}}}
    {\ooalign{\box\tw@ \cr \box\z@}}}
\makeatother

\makeatletter
\def\widebreve{\mathpalette\wide@breve}
\def\wide@breve#1#2{\sbox\z@{$#1#2$}
     \mathop{\vbox{\m@th\ialign{##\crcr
\kern0.08em\brevefill#1{0.8\wd\z@}\crcr\noalign{\nointerlineskip}
                    $\hss#1#2\hss$\crcr}}}\limits}
\def\brevefill#1#2{$\m@th\sbox\tw@{$#1($}
  \hss\resizebox{#2}{\wd\tw@}{\rotatebox[origin=c]{90}{\upshape(}}\hss$}
\makeatletter

\newcommand{\diag}{\mathop{\mathrm{diag}}}
\newcommand{\vect}{\mathop{\mathrm{vec}}}

\DeclareMathOperator{\sgn}{sgn}

\newcommand{\T}{\textsc{T}}

\newcommand{\dd}{\textnormal{d}}
\newcommand{\mathdag}{\normalfont{\dag}}

\newcommand{\E}{\mathsf{E}} 

\newcommand{\Cov}{{\mathsf{Cov}}} 
\newcommand{\Var}{{\mathsf{Var}}}

\newcommand{\Bias}{{\mathsf{Bias}}}
\newcommand{\bias}{{\mathsf{Bias}}}

\newcommand{\MSE}{{\mathsf{MSE}}}

\newcommand{\WMSE}{{\mathsf{WMSE}}}

\newcommand{\iid}{\textsc{iid}} 
\newcommand{\simIID}{   \overset{ \iid   }{\sim} }

\newcommand{\AR}{\textsc{ar}}
\newcommand{\MA}{\textsc{ma}}
\newcommand{\ARMA}{\textsc{arma}}

\newcommand{\Unif}{\textnormal{Unif}} 
 
\newcommand{\Normal}{\textnormal{N}}

\newcommand{\Exp}{\textnormal{Exp}}

\newcommand{\Bart}{\mathop{\mathrm{Bart}}}

\newcommand{\showName}{1}
\newcommand{\isMain}{1}
\newcommand{\paperTitle}{
Tail postcoloring in long-run variance estimation of time series
}
\def\isXr{1}
\begin{document}

\if1\isMain{
\title{\bf \paperTitle}
}\else{
\title{\bf Supplementary note to ``\paperTitle''}
}\fi

\if1\showName{
	\author[1]{Xu Liu 
				\footnote{Liu is currently affiliated with the Department of Statistics at the University of Washington. The contact email at the new affiliated institution is \url{xuliu8@uw.edu}}
				\thanks{Email: 
				\href{mailto:lexilxu@link.cuhk.edu.hk}
				{\nolinkurl{lexilxu@link.cuhk.edu.hk}}}}
	\author[1]{Kin Wai Chan \thanks{Email: 
				\href{mailto:kinwaichan@cuhk.edu.hk}
				{\nolinkurl{kinwaichan@cuhk.edu.hk}}}}
	\affil[1]{Department of Statistics and Data Science, The Chinese University of Hong Kong.}
}\else{
	\author[]{}
}\fi	
\maketitle

\bigskip
\begin{abstract}
Prewhitening is a common approach to deal with strong autocorrelation. 
In this article, we propose a new approach called tail postcoloring,
motivated by it. 
It uses parametric models to project, or color back, the neglected tail autocovariances 
in nonparametric estimators onto the final estimator.
This approach bridges the nonparametric variance estimator and the parametric coloring model through a scaling factor. 
It automatically switches between these two arms using a bandwidth parameter, 
without the need to transform the entire dataset into residuals, as in the standard prewhitening approach. 
When the coloring model is well-specified, a parametric rate can be achieved.
In finite samples, 
it is also more robust to misspecification of the coloring model compared to the whitening model in the standard approach.
Besides, it avoids severe potential variance inflation or power reduction caused by the recoloring factor 
in the standard approach.
We show that multiple parametric models can be used to construct a multiply robust tail postcolored estimator. 
It also naturally works for multivariate time series. 
A real-data example in Markov chain Monte Carlo output analysis is provided.    
\end{abstract}

\spacingset{1.7}
\newpage

\section{Introduction}\label{sec:intro}
Estimating variance in the presence of serial dependence plays an important role in many inference procedures.
In time series studies, 
research has been dedicated to the estimation of long-run variance 
\citep{carlstein86,kunsch89,chanyau2015_hoc,chan2022mean,chan2022optimal,ChanYau2022}
or spectral density \citep{parzen1957,priestley1982,politis1995}.
In econometrics, 
variance estimators are often discussed in the context of 
heteroskedasticity and autocorrelation consistent (HAC) estimation
{\citep{newey_west_1987,andrews1991}},
heteroskedasticity-autocorrelation robust (HAR) testing
{\citep{Kiefer2000,KieferVogelsang2002Econometrica,Kiefer2005Vogelsang}}
and generalized method of moments {\citep{hansen1982}}.
In Markov chain Monte Carlo (MCMC),
estimating the standard error of MCMC samples 
is important for computing the effective sample size and
conducting fixed-width convergence tests {\citep{Galin2006,james2010,wu2009,chanyau2013,chanyau2014_rTACM,LiuFlegal2018,LeungChan2022}}.

Let $\{ X_i \}_{i\in\mathbb{Z}}$ be a stationary time series
with mean $\mu = \E(X_0)$ and autocovariance  
$\gamma_k = \text{cov}(X_k, X_0)\ (k\in\mathbb{Z})$.
Under some regularity conditions \citep{brockwellDavis1991}, 
the central limit theorem for $\bar{X}_n = \sum_{i=1}^n X_i/n$ holds, i.e., 
  $
    n^{1/2}(\bar{X}_n - \mu)\Rightarrow\Normal(0,v),
  $
where $v = \sum_{k=-\infty}^{\infty}\gamma_k$
and ``$\Rightarrow$'' denotes weak convergence.
The target is the long-run variance (LRV) $v$.
Many commonly used estimators of $v$ \citep{bartlett1950periodogram,jowett1955comparison,hannan1957variance,parzen1957,newey_west_1987} 
can be written or asymptotically represented as
  \begin{equation}\label{eqt:kernelEst}
    \tilde{v} = \tilde{v}(\ell; K) = \sum_{k=-\infty}^{\infty} K\left(\frac{k}{\ell}\right) \tilde{\gamma}_k,
    \quad \text{where} \quad 
    \tilde{\gamma}_k = \frac{1}{n}\sum_{i=|k| +1}^{n} (X_i - \bar{X}_n)  (X_{i-|k| } - \bar{X}_n) .
  \end{equation}
  for $|k| < n$ and $\tilde{\gamma}_{k} = 0$ for $|k|\geq n$.
We may denote $\tilde{v}$ as $\tilde{v}(\ell; K; X_{1:n})$ 
to emphasize the input data $X_{1:n} = \{X_1, \ldots, X_n\}$. 
In this article, $\tilde{v}$ is referred to as an unadjusted estimator,
and is highlighted as $\tilde{v}_{\textsc{un}}$. 
There are two user-specified inputs in (\ref{eqt:kernelEst}).
First, the kernel function $K$ is typically symmetric and satisfies $K(0)=1$.
For example, the Bartlett kernel has a triangular shape defined as
$K_{\Bart}(t) = \max(1-|t|,0)$.
Second, the bandwidth $\ell$ controls the weight $K(k/\ell)$ applied on $\tilde{\gamma}_k$.
Major efforts for improving $\tilde{v}(\ell; K)$
focus on the selection of $\ell$ {\citep[e.g.,][]{OBM1984,andrews1991,song1995,politis2011}}
and $K$ {\citep[e.g.,][]{newey_west_1987,gallant1987,andrews1991,asyTheoryEcon2000,LazarusLewisStockWatson2018,vats2022lugsail,LiuChan2022}}.

Besides parameter selection,
\cite{andrews1992} noticed that regardless of the bandwidth and kernel choices,
strong temporal structure harms the performance of $\tilde{v}$,
especially for small samples.
Specifically, the effects include
size distortions in hypothesis testing {\citep{kurozumi2010reducing}},
undercoverage of confidence intervals and
early termination of MCMC sampling \citep{VatsDootika2019Moaf, vats2022lugsail}.
{ 
To deal with a strong temporal structure,
prewhitening has been proposed long ago in the literature 
\citep{blackman1958measurement,press1959power}.
The idea is 
to
transform the time series
into a less correlated structure closer to the white noise
via
a parametric model.
In the context of HAC estimation,
\cite{andrews1992}
uses prewhitening to flatten the spectral density at point zero
so that the estimate is less biased.
Recently, \cite{casini2024prewhitened} proposed a locally prewhitened LRV estimator
that is robust to nonstationarity.
They also provided new insights and corrections to the results in this context.}
However, there are two main issues with the prewhitening.
First, as a parametric approach,
when the model is misspecified, 
the performance deteriorates sharply in terms of bias distortion and variance inflation;
see Example~\ref{exp:simu_ar1}.
Second, because a recoloring coefficient is needed in the prewhitened estimator,
when the fitted process is close to having a unit root,
the estimator becomes very unstable with inflated variance.
\cite{andrews1992} uses a $0.97$ bound on the magnitude of the $\AR(1)$ parameter as a fix of the problem.
It is not a perfect fix though, as the bound is rather arbitrary and 
the problem still occurs when the parameter is close to but slightly below $1$.
In addition, prewhitening leads to power loss in multiple types of testing problems {\citep{Kiefer2005Vogelsang,muller2014hac}}.

{
Prewhitening is often used to de-bias the estimator on the time series.
We also remark that there are also other bias correction techniques in the literature.
For kernel estimators,
bias correction can be achieved through selecting different kernel functions. 
For example,
\cite{politis1995} discussed bias correction for Bartlett's spectral density estimator using the trapezoidal kernel function; 
\cite{vats2022lugsail} proposed the lugsail lag window to adjust for the finite-sample downward bias;
\cite{LiuChan2022} developed a class of converging kernel functions that correct for the bias.
There are also other methods
that utilize parametric models.
For example,
\cite{astfalck2024debiasing} introduced de-biasing techniques for the Welch's method for spectral density estimation.
\cite{astfalck2025bias} discussed bias correction methods for quadratic estimators of spectral density.

}

We first present the results under stationarity and then extend them to the nonstationary case in Section~\ref{sec:supp_het}.
This is reasonable as we focus more on HAC estimation than on
heteroskedasticity robust testing.
Relevant discussions on HAR testing
include the fixed-$b$ approach \citep{Kiefer2000,KieferVogelsang2002,Kiefer2005Vogelsang}, 
where the normalizer replacing the LRV estimator converges weakly to a non-degenerate distribution,  
and the test statistic has been shown to have non-pivotal distribution under nonstationarity \citep{casini2024fixed}.
\cite{casini2024prewhitened} proposed a nonlinear vector autoregressive prewhitening estimator that deals with nonstationarity 
using a double kernel approach; see also \cite{casini2023theory}.
Such methods typically involve 
a component of detrending or differencing to deal with nonstationarity.
Our method, on the other hand, implements solely the coloring step,
and can be combined with additional differencing approaches.
\cite{preinerstorferPW} compared the finite-sample size and power properties of prewhitened estimators.

In this article, 
instead of a complete transformation of the dependence structure,
we introduce the tail postcoloring method that approaches the problem of strong autocorrelation 
from a different perspective.
This construction bridges the parametric and nonparametric approaches,
while the bandwidth $\ell$ tunes the relative effect of each component.

The rest of the paper is organized as follows. 
Section \ref{sec:notation} presents the notation
and the setting of the model. 
The standard prewhitening technique will also be reviewed. 
Section \ref{sec:Unadj} motivates the key proposed principle and methodology. 
The asymptotic properties of the proposed estimator are derived. 
In particular, a closed-form formula for the optimal bandwidth and 
the new method for bridging parametric and nonparametric approaches are presented. 
Theoretical and empirical comparisons are demonstrated. 
Section \ref{sec:general} generalizes the proposed principle to different situations, 
including 
the use of multiple models for achieving multi-robustness, 
the use of a general kernel for performing tail postcoloring, 
extension to multivariate time series, and 
robustness to heteroskedasticity. 
Section \ref{sec:discussion} discusses the idea of tail postcoloring as a hybrid method that combines parametric and nonparametric tools, 
as well as the concept of using multiple models to assist the main estimator. 
We compare these ideas with similar approaches in the literature.
Section \ref{sec:simulation} demonstrates the practical usefulness via two applications: 
HAC estimation and convergence diagnosis in
MCMC simulation problems. 
All proofs and additional simulation results are deferred to the supplement. 
The R package \texttt{"postcoloring"} is available online for implementation. 

\section{Notation, setting and the prewhitening technique}\label{sec:notation}

\begin{assumption}\label{ass:K} 
Let $p\in\bar{\mathbb{N}}=\mathbb{N}\cup\{\infty\}$ and $\bar{\mathbb{R}}=[-\infty, \infty]$.
The class of $p$th order kernels is
  \begin{align}\label{eqt:ordinaryKernel}
    \mathcal{K}_{p} = \left\{ K \in \mathcal{C}_p : K(0)=1,
    K(t)=K(-t)\text{ for all }t\in\bar{\mathbb{R}}
    \text{ and }\int_{0}^{\infty}K^2(u) \, \dd u<\infty
    \right\},
  \end{align}
where
$\mathcal{C}_p$ is the set of functions $K:\bar{\mathbb{R}}\rightarrow [-1,1]$ that are
continuous at $0$ and at all but a finite number of other points,
and $p$ is the largest value such that 
$\lim_{t\downarrow0} \{ K(0)-K(t) \}/t^p \in\mathbb{R}\setminus\{0\}$.
\end{assumption}
We use the framework of dependence measure in \citet{wu2005};
see also \citet{wu2010}, \citet{wu2011} and \citet{wu2012covariance}. 
Let $\{\varepsilon_i,\varepsilon_i'\}_{i\in\mathbb Z}$
be independent and identically distributed innovations. 
Suppose that $X_i=g(\mathcal{F}_i)$ for some measurable function $g$,
where $\mathcal F_i=(\dots,\varepsilon_{i-2},\varepsilon_{i-1},\varepsilon_{i})$. 
Define the coupled version of $X_i$ as $X'_{i}=g(\mathcal F'_i)$,
where $\mathcal F'_i=(\dots,\varepsilon_{-1},\varepsilon'_{0},\varepsilon_{1},\dots,\varepsilon_{i-1},\varepsilon_{i})$.
For $p\geq1$, the physical dependence measure and its aggregated value across time are defined as
$\delta_{p,i}={\|X_i-X'_{i}\|}_p$ and $\Delta_p=\sum_{i=0}^{\infty}\delta_{p,i}$, respectively,
where $\|\cdot\|_{p}=\{\E({|\cdot|}^p)\}^{1/p}$ denotes the $\mathcal{L}^{p}$ norm. 
Write $\|\cdot\| =\|\cdot\|_{2}$. 
Define that for $p\in\mathbb{N}_0$, 
\begin{align}\label{eqt:v_u_kappa}
    v_p=\sum_{k\in\mathbb Z}|k|^p\gamma_k,\qquad
    u_p=\sum_{k\in\mathbb Z}|k|^p|\gamma_k| \qquad \text{and} \qquad
    \kappa_p = v_p/v,
\end{align}
where $\mathbb{N}_0 = \{0,1,2,\ldots\}$ and $0^0 = 1$.
In addition, define two constants
\[
  A = \int_{0}^{\infty} K^2(t) \,\dd t \qquad \text{and} \qquad  B=\lim_{t\downarrow 0} \{ K(t) - K(0) \}/ t^{p}.
\]

  \begin{assumption}\label{assp:AspSeries}
    Let $\{X_i\}$ be strictly stationary and ergodic 
    with mean $\E(X_1)=\mu$.
    Assume 
    (i) $\E\left({|X_1|}^{\nu}\right)<\infty$ for some $\nu>4$;
    (ii) $\{X_i\}$ is $q$-stable with $q=4$, i.e., $\Delta_4<\infty$.
  \end{assumption}
Assumption~\ref{assp:AspSeries}
ensures short-range dependence. 
It is
satisfied by many commonly used models;
see \citet{wu2005}.
Write $x_n \asymp y_n$ if there exist constants $C_1,C_2>0$ such that $C_1 y_n < x_n < C_2 y_n$.
Denote
the mean-squared error (MSE) of $\tilde{v}$ by $\MSE(\tilde{v}) = \E\{(\tilde{v}-v)^2\}$
and the standardized MSE by $\MSE(\tilde{v})/v^2$.
The bias of $\tilde{v}$ is $\bias(\tilde{v}) = \E(\tilde{v}) - v$.

Now we review \cite{andrews1992}'s prewhitening technique for estimating $v$.
For illustration, a first-order autoregressive (\textsc{ar}$(1)$) whitening model is considered: 
  \begin{definition}[\textsc{ar}$(1)$ prewhitened estimator]\label{def:SPW} 
    Define the prewhitened time series $\{Z_i\}_{i=2}^{n}$ as
     $Z_i = X_i - \bar{\phi} X_{i-1}\ (i = 2,\ldots,n)$,
     where $\bar{\phi} = \tilde{\gamma}_1/\tilde{\gamma}_0$ is the sample autocorrelation at lag $1$.
    The $\textsc{ar}(1)$ prewhitened estimator 
    proposed in \cite{andrews1992} is
    \begin{equation}\label{eqt:AndrewsAR1}
      \hat{v}_{\textsc{pw}} = {\tilde{v}(\ell; K; Z_{2:n})}/{(1- \bar{\phi})^2}.
    \end{equation}
  \end{definition}
The recoloring factor $1/(1-\bar{\phi})^2$ in (\ref{eqt:AndrewsAR1}) 
leads to a notable inflation of variance {\citep{Sul2005}} when $\bar{\phi}\approx 1$.
This motivates us to develop a more robust technique.

\section{The tail postcoloring technique}\label{sec:Unadj}
\subsection{Motivations}

A direct impact of strong autocorrelation on $\tilde{v}$ is 
that the bandwidth $\ell$ is often too short 
to assign enough weights $K(k/\ell)$ 
to significant $\tilde{\gamma}_k$ at large lags $k$.
We propose parametrically projecting the neglected tail autocovariances 
onto the nonparametric estimator $\tilde{v}$.

For example, in Figure~\ref{fig:coef} (a), 
if $\ell=15$ and $K=K_{\Bart}$, 
then only $\tilde{\gamma}_k$ where $|k|\leq \ell$
are used in the computation of $\tilde{v}$.
However, the tail autocovariances $\{\gamma_k\}_{|k|>\ell}$ are not used but should 
also determine the value of $v$.
Due to this fact, 
rather than estimating $v = \sum_{k=-\infty}^{\infty}\gamma_k$,
the estimator $\tilde{v} = \sum_{k=-\infty}^{\infty} K(k/\ell) \tilde{\gamma}_k$ should be better at estimating
another quantity, namely 
\[
  M_{\ell, K} = \sum_{k=-\infty}^{\infty} K(k/\ell) \gamma_k,
\]
which essentially only includes the autocovariances in the main part of the correlogram.
So, we propose rescaling $\tilde{v}$ by a parametric estimator of the ratio $\eta_{\ell,K} = v/M_{\ell, K}$
to adjust for the neglected tail autocovariances. 
The proposed estimator is formulated as follows: 

  \begin{figure}[t]
      \centering
      \includegraphics[width=.95\linewidth]{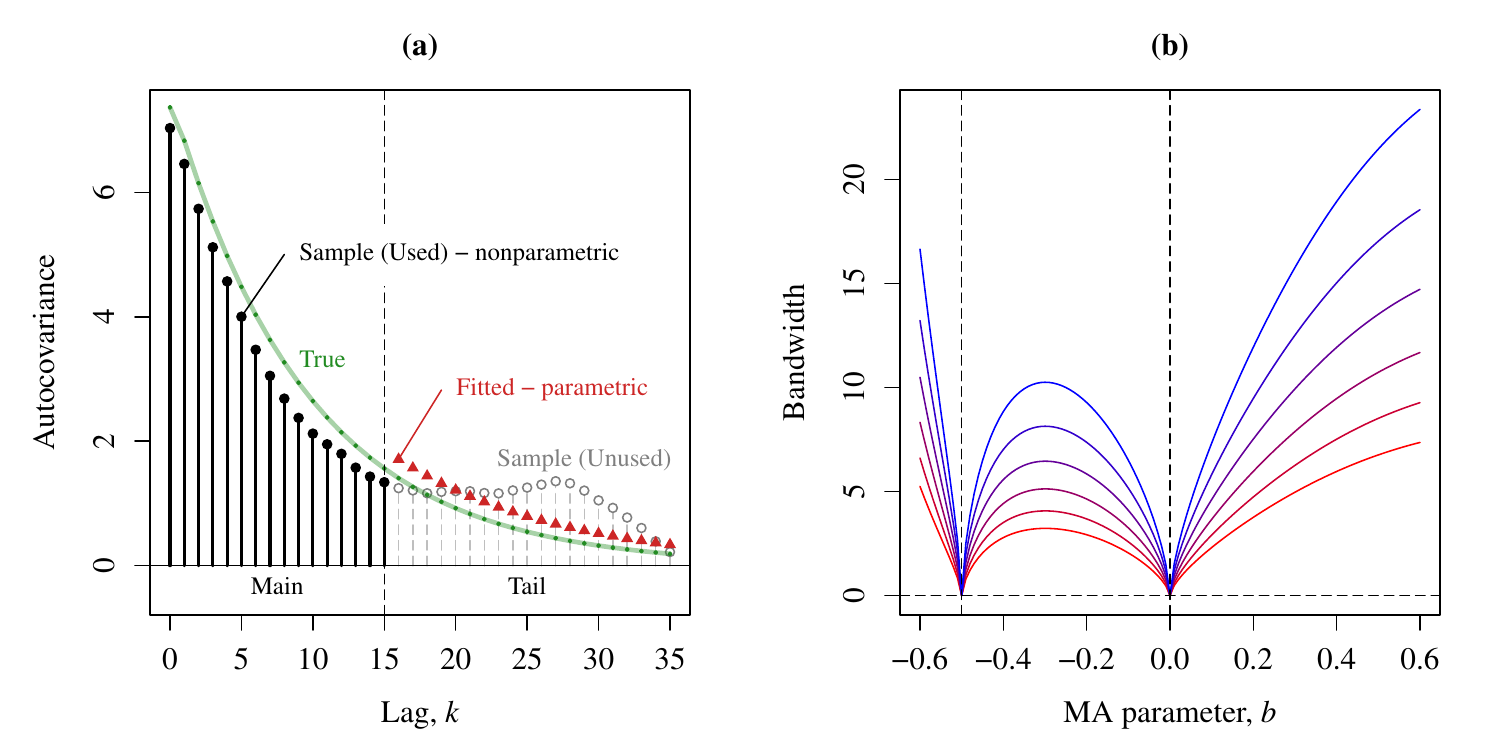}
    \vspace{-0.4cm}
    \caption{ 
      (a) A schematic illustration of tail autocovariances.
      If $\ell=15$, only $\{\tilde{\gamma}_k\}_{|k|\leq 15}$ are used, while $\{\tilde{\gamma}_k\}_{|k|>15}$ are unused
      in the nonparametric estimator $\tilde{v}$. 
      Tail postcoloring projects the data to the nearest fitted model, 
      and estimates the tail autocovariances $\{\gamma_k\}_{|k|>15}$ parametrically.
      (b) The optimal bandwidth $\ell$ defined in (\ref{eqt:OptBW}).
      Bartlett kernel and $\textsc{ar}(1)$ tail postcoloring model are used.
      The data are generated from the $\textsc{arma}(1,1)$ model with 
      $\textsc{ar}$ parameter $a=0.5$ and 
      $\textsc{ma}$ parameter $b$.
      The lines (from bottom to top) represent sample size $n=100,200,400,800,1600,3200$, respectively.  
      When $b=-0.5,0$, the data generating mechanism reduces to an $\textsc{ar}(1)$ model. 
      }
      \label{fig:coef}
  \end{figure}

\begin{definition}[Tail postcoloring]\label{def:SimpleTPW}
  Let $X_i({\theta}) = h_{\theta}(\mathcal{F}_i)$ be a parametric tail postcoloring model, or a coloring model for short,
  where $h_{\theta}$ is a measurable function, ${\theta} \in \mathcal{H} \subseteq \mathbb{R}^q$
  for $q\in\mathbb{N}$.
  The autocovariance under this model is $\gamma_k({\theta}) = \Cov\{X_i({\theta}), X_{i-k}({\theta})\}$ ($k\in\mathbb{Z}$).
  {Assume $v(\theta) = \sum_{k\in\mathbb{Z}} \gamma_k(\theta) \in\mathbb{R}^+$ for $\theta\in\mathcal{H}$.} 
  Let $\bar{\theta}\in\mathcal{H}$ be a $n^{1/2}$-consistent estimator for some $\theta_{\star}\in\mathcal{H}$.
  Let $\ell \in [1,n) \cap \mathbb{Z}$. 
  Denote the parametric estimators of $v$ and $M_{\ell, K}$ by 
  \[
    v(\bar{\theta}) = \sum_{k\in\mathbb{Z}} \gamma_k(\bar{\theta})
  \quad \text{and} \quad
  M_{\ell, K} (\bar{\theta}) =  \sum_{k\in\mathbb{Z}} K(k/\ell) \gamma_k(\bar{\theta}),
  \]
  respectively. 
  Then the tail postcolored estimator of $v$ is
    \begin{equation}\label{eqt:proposalsimple}
      \hat{v}_{\textsc{tail}} \equiv \hat{v}(\ell ; K; \bar{\theta}) 
            = {\eta}_{\ell, K}(\bar{\theta}) {\tilde{v}(\ell; K)},
      \quad \text{where} \quad
      \eta_{\ell, K}(\bar{\theta}) = v(\bar{\theta}) / M_{\ell, K} (\bar{\theta})  
    \end{equation}
    is a parametric postcoloring coefficient and   
  $\tilde{v}$ is a nonparametric estimator of $v$ defined in (\ref{eqt:kernelEst}).
  {In (\ref{eqt:proposalsimple}), we define $\cdot/0 \equiv 1$. 
	We may abbreviate $\hat{v}_{\textsc{tail}}$ as $\hat{v}$ if no confusion is possible. 
	}
\end{definition}

Tail postcoloring does not change the structure of the data or the user's choice of $\tilde{v}$, 
so it is less severely impacted by coloring model misspecification in finite samples,
as we will demonstrate in Example \ref{exp:simu_ar1}.
We conclude this subsection with an explicate example. 

\begin{example}\label{eg:TPW_AR1}
Consider an \textsc{ar}$(1)$ tail postcoloring model:
$X_i(\phi) = \phi X_{i-1}(\phi) + \varepsilon_i$, where $\phi\in(-1,1)$ and 
$\varepsilon_i\sim\Normal(0,\sigma^2)$ indepedently for some $0<\sigma<\infty$.
Let $\bar{\phi}\in(-1,1)$ be $n^{1/2}$-consistent for any $\phi\in(-1,1)$.
If $K=K_{\Bart}$, then (\ref{eqt:proposalsimple}) is computed as
  \begin{equation}\label{eqt:propDebias}
    \hat{v}_{\textsc{tail}} \equiv \hat{v}(\ell ; K; \bar{\phi}) 
     = \eta_{\ell, K}(\bar{\phi}) \tilde{v}(\ell; K),
     \qquad \text{where} \qquad
     \eta_{\ell, K}(\bar{\phi}) = \frac{\ell-\ell\bar{\phi}^2}{2\bar{\phi}^{\ell+1}-2\bar{\phi}-\ell\bar{\phi}^2+\ell} 
  \end{equation}
parametrically rescales $\tilde{v} = \sum_{k=-\ell}^{\ell}(1-|k|/\ell)\tilde{\gamma}_k$ 
from a biased target $M_{\ell,K} = \sum_{k=-\ell}^{\ell}(1-|k|/\ell)\gamma_k$ to the correct estimand 
$v=\sum_{k=-\infty}^{\infty}\gamma_k$. 
In practice, we set $\bar{\phi} = \tilde{\gamma}_1/ \tilde{\gamma}_0$ as an estimator of $\phi_{\star}=\gamma_1/\gamma_0$
to project the tail autocovariances to the closest \textsc{ar}$(1)$ model. 
If $\ell=1$, then $\hat{v}_{\textsc{tail}}=\tilde{\gamma}_0(1+\bar{\phi})/(1-\bar{\phi})$
is an $\textsc{ar}(1)$-based parametric estimator.
It is $n^{1/2}$-consistent for $v$ if the model is well-specified. 
If $\ell\rightarrow\infty$, then $\hat{v}_{\textsc{tail}}=\tilde{v} +o_p(1/\ell)$ becomes  a nonparametric estimator.
In Section~\ref{sec:OptBW}, we show that the optimal $\ell$ tunes $\hat{v}$ so that 
it switches between two arms accordingly; see Figure~\ref{fig:coef} (b) for a graphical illustration.
\end{example}

\begin{remark}[Positivity of $M_{\ell, K}( {\theta})$]\label{rem:positivity}
{If $K$ is a positive semidefinite (PSD) kernel, such as Bartlett kernel and quadratic spectral kernel, 
then $M_{\ell, K}( {\theta})\geq 0$ for all $\ell$ and $ {\theta}$.
The term $M_{\ell, K}( {\theta})$ can be represented as the LRV of 
$\bar{X}_{n,W} = \sum_{i=1}^n X_i W_i/n$,
where $X_{i} = g_{ {\theta}}(\mathcal{F}_i)$
and  $W_1, \ldots W_n$ are independent on $X_1, \ldots, X_n$ 
and are jointly normal with mean $0$ and autocovariance 
$\E(W_iW_{i-k}) = K(k/\ell)$, i.e.,
$M_{\ell,K}( {\theta}) = \lim_{n\to\infty}\Var(\bar{X}_{n,W})$.
Hence $M_{\ell,K}( {\theta}) = 0$ only in the degenerate case of zero
LRV of $\bar{X}_{n,W}$.
The degenerate case occurs, for example,
if $X_{1} = \cdots = X_{n}=0$.
Hence, PSD kernels are recommended.}
If $K$ is not PSD,  
it is still guaranteed to have $\lim_{n\rightarrow\infty}M_{\ell, K}({\theta}) > 0$. 
So, the quantities in (\ref{eqt:proposalsimple}) are always well-defined and positive asymptotically. 
If $K$ is not PSD and $n$ is small,
one may ensure $M_{\ell, K}( {\theta})>0$ by the correction methods suggested in, e.g., 
\cite{politis2011} and \cite{LiuChan2022}; see Remark \ref{rem:positiveDef} for the details.
{ Such correction do not affect the asymptotic properties of the estimator 
and the finite-sample effect on the performance is minor.
}

\end{remark}

{

\begin{remark}[Targeting the finite-sample variance]\label{rem:FSvar}
As pointed out by one of the referees,
estimating the finite-sample variance 
$v^{\textsc{fs}} = \sum_{k=1-n}^{n-1}(1-|k|/n)\gamma_{k}$ can be a better target than the long-run variance $v$; 
see, e.g., \cite{perron2011irrelevance}.
Our tail postcolored estimator for $v^{\textsc{fs}}$ is 
$\hat{v}_{\textsc{tail}}^{\textsc{fs}} = \left\{ v^{\textsc{fs}}(\bar{\theta})/ M_{\ell, K}(\bar{\theta}) \right\} \tilde{v}(\ell; K)$,
where $v^{\textsc{fs}}(\bar{\theta}) = \sum_{k=1-n}^{n-1}(1-|k|/n)\gamma_{k}(\bar{\theta})$.
Tail postcoloring remains useful for estimating $v^{\textsc{fs}}$
because $\tilde{v} = \sum_{|k|\leq \ell} K(k/\ell) \tilde{\gamma}_k$ also neglects 
part of the tail autocovariances $\{\gamma_k\}_{\ell<|k|<n}$ that appear in $v^{\textsc{fs}} = \sum_{|k|<n} (1-|k|/n)\gamma_k$.
Thus, postcoloring the tail autocovariances may improve $\tilde{v}$.
Besides, it is remarked that the estimators of spectral density 
would have infinite minimax risk and thus its estimation is an ill-posed problem \citep{Potscher2002}. 
Alternatively, instead of estimating $v$ consistently, 
fixed-$b$ asymptotics \citep{Kiefer2000,KieferVogelsang2002Econometrica,KieferVogelsang2002,Kiefer2005Vogelsang} can be used. 
Setting $\ell=n$ in $\tilde{v}(\ell; K)$ is inconsistent for $v$, but it will converge weakly a non-degenerate distribution proportional to $v$. 
This approach leads to less size distortion in finite samples and is useful for constructing the HAR test statistics.
Since the current article studies estimation of $v$, we leave this fixed-$b$ asymptotics for further study. 
\end{remark}}

\subsection{Optimal bandwidth selection and measure of distance}\label{sec:OptBW}

Under the tail postcoloring model, 
we define the dependence ratio as 
$\kappa_p( \theta ) = v_p( \theta )/v_0(\theta)$,
where
$v_p(\theta) = \sum_{k\in\mathbb{Z}} {|k|}^p \gamma_k(\theta)$.
We show that the performance of $\hat{v}_{\textsc{tail}}$ depends on 
the difference between the parametrically modeled ratio $\kappa_p( \theta )$ and the true ratio $\kappa_p$:

  \begin{theorem}[Bias and variance]\label{thm:BiasVar}
    Let $K\in\mathcal{K}_{p}$ for some $p \in \mathbb{N}$. 
    Suppose that the estimator $\bar{\theta}$ satisfies $ n^{1/2}(\bar{\theta} - \theta_{\star}) = O_p(1)$ for some $\theta_{\star} \in \mathcal{H}$
    as $n\to\infty$.
    Assume
    $u_{p}(\theta_{\star}) \equiv \sum_{k\in\mathbb{Z}} {|k| }^{p} \left\mid  \gamma_k(\theta_{\star}) \right\mid  < \infty$, $\sup_{\theta\in\mathcal{H}}\left|{\partial v(\theta)}/{\partial \theta}\right|<\infty$ and 
    $\sup_{\theta\in\mathcal{H}} \left| {\partial M_{\ell,K}(\theta)}/{\partial \theta} \right|<\infty$ for all $\ell>0$.
    If Assumption~\ref{assp:AspSeries} holds, $u_{p}< \infty$ and $\ell \asymp n^{1/(2p+1)}$,  
    then we have (i) 
    \[
      \lim_{h\to\infty} \lim_{n\to\infty}  \MSE_h \{\hat{v}(\ell ; K; \bar{\theta})\}
        = \lim_{n\to\infty} n^{{2p}/({2p+1})} \MSE\{\hat{v}(\ell ; K; \theta_{\star})\},
    \]
    where $\MSE_{h}(\hat{v}) = \E\left[\min\left\{n^{2p/(2p+1)} (\hat v-v)^2,h\right\}\right]$ is the truncated mean-squared error; 
    (ii) $\Var\left\{ \hat{v}(\ell ; K; \theta_{\star}) \right\} \sim \Var\left\{ \tilde{v}(\ell;K) \right\} \sim 4A v^2 \ell/n$; 
    and (iii)
    \begin{equation}\label{eqt:NonnegligibleBias}
      \bias\left\{ \hat{v}(\ell ; K; \theta_{\star}) \right\}
      =  \frac{B \xi_p  }{\ell^p}v + o\left(\frac{1}{\ell^p}\right) + O\left(\frac{\ell}{n}\right), 
      \quad\text{where}\quad \xi_p = \kappa_p - \kappa_{p}(\theta_{\star}). 
    \end{equation}
  \end{theorem}

{ It can be observed that tail postcoloring leads to a change in the leading bias term, 
while the asymptotic variance remains the same.
A comparison with the prewhitened estimator of \cite{andrews1992} is provided in 
{\ifnum\isXr=1{Remark \ref{rem:uni}}\else{Remark A.1}\fi} of the supplement. 
By balancing the squared bias and variance, we can derive the MSE-optimal $\ell$; see Corollary \ref{cor:optimalBW} below.   
We will determine when tail postcoloring leads to an improvement in MSE in Section \ref{sec:comparison_main}.  
We also refer readers to Theorem 3.1 and the discussions in \cite{casini2024prewhitened} for 
a locally prewhitened estimator for nonstationary time series.}

  \begin{corollary}[Optimality]\label{cor:optimalBW}
    Under the conditions of Theorem~\ref{thm:BiasVar}
    and suppose that $\xi_p\neq 0$,
    the asymptotic truncated mean-squared error of $\hat{v}$ is uniquely minimized at 
    \begin{equation}\label{eqt:OptBW}
      \ell_{\textsc{tail}} \sim \left\{ {pB^2 \xi_p^2 n }/{ (2A)} \right\}^{1/(2p+1)},
    \end{equation}
    and satisfies that 
	$
      \lim_{h\to\infty} \lim_{n\to\infty} \MSE_h\{\hat{v}(\ell_{\textsc{tail}} ; K; \bar{{\theta}})\}
       = (2p+1)\left\{  \left( {2A}/{p}\right)^p |B  \xi_p|  \right\}^{2/(2p+1)} v^2.
      $
  \end{corollary}

When we want to emphasize that the optimal $\ell_{\textsc{tail}}$ is used, 
we may write $\hat{v}_{\textsc{tail}}$ as $\hat{v}_{\textsc{tail}\star} = \hat{v}(\ell_{\textsc{tail}} ; K; \bar{{\theta}})$. 
Corollary~\ref{cor:optimalBW} is a straightforward application of Theorem~\ref{thm:BiasVar}.
The optimal bandwidth in (\ref{eqt:OptBW}) 
implies that the proposed estimator $\hat{v}_{\textsc{tail}\star}$ automatically switches between parametric
and non-parametric forms through the tuning of the bandwidth.
If $| \xi_p |$ is small,
i.e., the discrepancy of the data generating process and the parametric model measured by the difference of the dependence ratios
is small, $\ell_{\textsc{tail}}$ would be small according to the definition in (\ref{eqt:OptBW}). 
Using a short bandwidth enhances the parametric component while the nonparametric component is weakened.

 \begin{remark}
    The truncated MSE criterion in Theorem \ref{thm:BiasVar}
    is used in \cite{andrews1991} and \cite{andrews1992} 
    to circumvent the case that the MSE blows up due to the estimator $\bar{\theta}$.
    For example, when the tail postcoloring model is the \textsc{ar}$(1)$ model, 
    this problem can occur when the data is close to having a unit root, i.e., when 
    $\bar{\theta} = \bar{\phi} \approx 1$. 
  \end{remark}

  \begin{remark}[Estimation of $\ell_{\textsc{tail}}$]\label{rem:est_opt_bw}
    In practice, $\xi_p = \kappa_p - \kappa_p(\theta_{\star})$ is unknown. 
    We suggest two methods.
    First, motivated by \cite{andrews1991},
    we estimate $\kappa_p$ by a parametric model that includes the tail postcoloring model. 
    For example, if an $\AR(1)$ tail postcoloring model is used, 
    we may fit an $\ARMA(1,1)$ model,
    i.e., $X_i = a X_{i-1} + \varepsilon_i + b \varepsilon_{i-1}$, to estimate $\kappa_p$. 
    In particular, when $p=1$, a parametric plug-in estimator of $\xi_1$ is 
    \[
      \bar{\xi}_1 = 2(\bar{a}+\bar{b})(1+\bar{a}\bar{b})(1+\bar{b})^{-2}(1-\bar{a}^{-2}) - 2\bar{\phi}(1-\bar{\phi}^2)^{-1},
    \]
    where $\bar{\phi}$ and $(\bar{a},\bar{b})$ are $n^{1/2}$-consistent estimators of 
    $\phi$ in the $\AR(1)$ model and 
    $(a,b)$ in the $\ARMA(1,1)$ model, respectively.
  Alternatively, we can estimate $\kappa_p = v_p/v_0$ with kernel estimators.
  In particular, when $p=1$, a nonparametric plug-in estimator of $\xi_1$ is 
  \[
    \hat{\xi}_1 
      = \tilde{v}_{1,\#} /\tilde{v}_{0,\#}  - 2\bar{\phi}(1-\bar{\phi}^2)^{-1}, 
  \]
  where $\tilde{v}_{r,\#} \equiv \sum_{|k|< n} K_r(k/\ell_r^{\#}) {|k|}^r \tilde{\gamma}_k$
  and $\ell_r^{\#} = O\{n^{1/(2r+3)}\}$ for $r\in\{0,p\}$.  
    We are aware that estimating $v_p$ is difficult especially when $v_p\approx 0$,
    so it may slightly inflate the risk in such cases; 
    see Section~\ref{sec:simulation}.
    Unless otherwise stated, we estimate $\xi_1$ by $\hat{\xi}_1$ with $K_1=K_0=K_{\Bart}$. 
  \end{remark}

\subsection{Theoretical and empirical comparison}\label{sec:comparison_main}
{
We compare the tail postcolored estimator $\hat{v}_{\textsc{tail}}$, 
prewhitened estimator $\hat{v}_{\textsc{pw}}$, and 
the unadjusted estimator $\tilde{v}_{\textsc{un}}$ defined in (\ref{eqt:kernelEst}).
All estimators use the same kernel $K$ and their own MSE-optimal bandwidths.
We prove that 
(i) $\lim_{n\rightarrow\infty} \MSE(\hat{v}_{\textsc{tail}}) / \MSE(\tilde{v}_{\textsc{un}}) <1$ if and only if $|\xi_p| < |\kappa_p|$; and 
(ii) $\lim_{n\rightarrow\infty} \MSE(\hat{v}_{\textsc{pw}}) / \MSE(\tilde{v}_{\textsc{un}}) <1$ if and only if $|\kappa^{Z}_p| < |\kappa_p|$,
  where $\kappa^{Z}_p$
  is the dependence ratio of the prewhitened series $\{Z_{i}\}$; 
  see {\ifnum\isXr=1{Section \ref{sec:comparison_pw}}\else{Section A.1}\fi}  for details.}

We first theoretically compare the asymptotic behavior of 
$\hat{v}_{\textsc{tail}}$, $\hat{v}_{\textsc{pw}}$, and $\tilde{v}_{\textsc{un}}$
in Example~\ref{rem:CompTPW}.
Then a simulation experiment is performed for studying the finite-sample performance 
in Example~\ref{exp:simu_ar1}.
We also compare them with other commonly used estimators. 

{
\begin{example}[Theoretical comparison]\label{rem:CompTPW} 
  Suppose the \textsc{ar}$(1)$ tail postcoloring model in Example \ref{eg:TPW_AR1} is used. 
  Then the modeled dependence ratio is $\kappa_1(\phi)= 2\phi/(1-\phi^2)$. 
  In this case, tail postcoloring reduces MSE asymptotically if and only if
  $\phi\in(\phi_-, \phi_+)$, where 
  \begin{align}\label{eqt:sol1}
     \phi_- = \varphi\mathbb{1}(\kappa_1<0), \quad
     \phi_+ = \varphi\mathbb{1}(\kappa_1>0), \quad
     \varphi = \frac{-1+(4\kappa_1^2 + 1)^{1/2}}{2\kappa_1}.
  \end{align}
  Suppose that
  $X_{1:n}$ are indeed generated from the $\ARMA(1,1)$ model: 
  $X_i = a X_{i-1} + \varepsilon_i + b \varepsilon_{i-1}$,  
  where $a,b\in(-1,1)$ and $\varepsilon_i\sim\Normal(0,1)$ independently. 
  Then
  $\kappa_1 =  2\{ a + b + a(a+b)^2 /(1-a^2)\}/(1+b)^2$.
  In particular, if $b=0$, then $\kappa_1 = 2a/(1-a^2)$ and $\varphi = (a^2-1+|a^4 + 14a^2+1|^{1/2})/(4a)$.
  Hence, tail postcoloring may improve MSE not only when the coloring model is well-specified (i.e., $\phi=a$)
  but also when the coloring model is misspecified (i.e., $\phi\in(\phi_-, a) \cup( a,\phi_+)$). 
  For example, when $a=0.9$, the improvement condition is $\phi\in(0, 0.9486)$ approximately. 

    We compare the range with the one provided in \citet{andrews1992}
    when $\phi \neq a$.  
    If $\phi>0$, then 
    $\MSE(\hat{v}_{\textsc{tail}}) < \MSE(\hat{v}_{\textsc{pw}})$ asymptotically.  
    When $\phi<0$, we have the opposite conclusion.
    The case $\phi>0$ is used and studied more frequently than $\phi<0$ in prewhitening; see, e.g., 
      \cite{muller2014hac}. 
  In particular, if $\phi = \phi_{\star} \equiv \gamma_1/\gamma_0$, 
  then Figure~\ref{fig:ar1pw1}(a)--(b) visualizes the ranges of $(a,b)$ such that $\hat{v}_{\textsc{pw}}$
  and $\hat{v}_{\textsc{tail}}$ improve the unadjusted counterpart $\tilde{v}_{\textsc{un}}$, respectively. 
  It shows that tail postcoloring has a larger improvement region of $(a,b)$ than standard prewhitening.
  Moreover, Figure~\ref{fig:ar1pw1}(c) shows that tail postcoloring improves prewhitening 
  when $a+b>0$, which arguably covers most commonly seen positively correlated time series. 
  Similar discussions on the bias and variance are provided in {\ifnum\isXr=1{Example \ref{exp:ar11}}\else{Example A.1}\fi}  of the supplement.
\end{example}
}

\begin{figure}[t]
    \begin{center}
          \includegraphics[width=\linewidth]{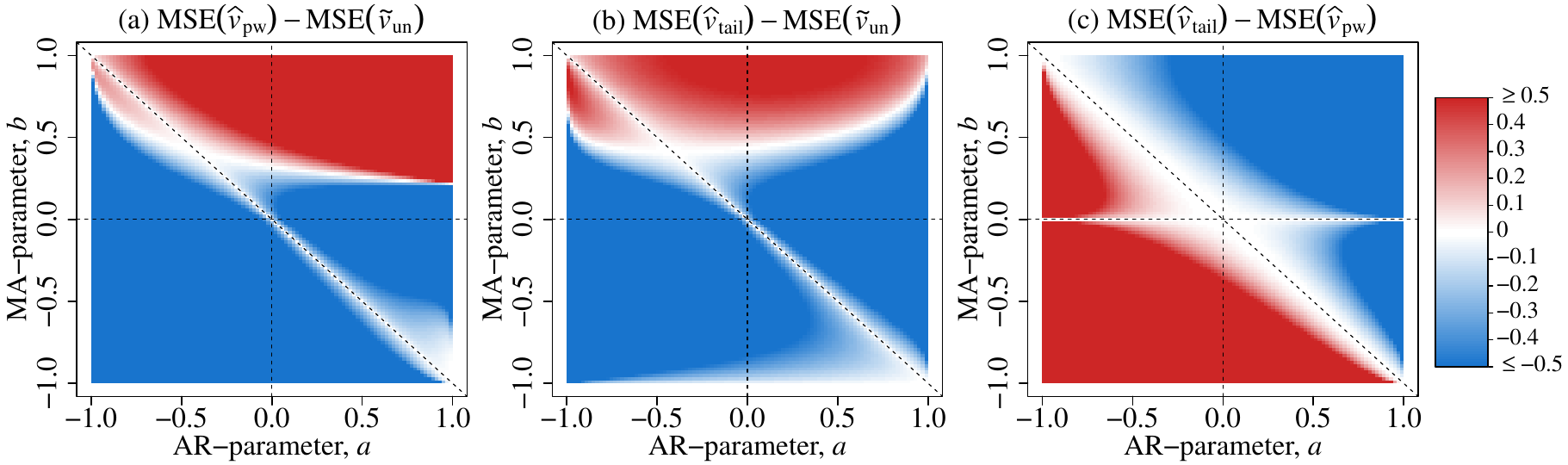}
    \end{center}
    \vspace{-0.5cm}
        \caption{Heat maps display the asymptotic differences of the MSEs after standardizing by the true value of $v^2$. 
        Plot (a), (b), and (c) show the limiting values of 
        $\{ \MSE(\hat{v}_{\textsc{pw}})-\MSE(\tilde{v}_{\textsc{un}}) \}n^{2/3}/v^2$, 
        $\{ \MSE(\hat{v}_{\textsc{tail}})-\MSE(\tilde{v}_{\textsc{un}}) \}n^{2/3}/v^2$, and 
        $\{ \MSE(\hat{v}_{\textsc{tail}})-\MSE(\hat{v}_{\textsc{pw}}) \}n^{2/3}/v^2$, respectively.  
        See Example \ref{rem:CompTPW} for the details. 
        A difference less than zero (i.e., blueish regions) indicates an improvement
        of the first estimator over the second estimator in the difference. 
        Note also that the titles of plots do not show the standardization $n^{2/3}/v^2$ due to space constraints. 
        }\label{fig:ar1pw1}
\end{figure}

\begin{example}[Simulation]\label{exp:simu_ar1}
  Generate the data from the $\ARMA(1,1)$ model: $X_i = a X_{i-1} + \varepsilon_i + b \varepsilon_{i-1}$.
    We consider the Bartlett kernel $K=K_{\Bart}$ and 
    the \textsc{ar}$(1)$ model for postcoloring:
    $X_i(\phi) = \phi X_{i-1}(\phi) + \varepsilon_i$
    with the estimator $\bar{\phi}=\tilde{\gamma}_1/\tilde{\gamma}_0$, which is consistent for $\phi_{\star} = \gamma_1/\gamma_0$;
    see Example~\ref{eg:TPW_AR1}. 
  In this setting, the exact values of $\phi_{\star}$ and $\xi_1$ are 
  \[
    \phi_{\star} = \frac{(a+b)(1+ab)}{1+b^2+2ab}
      \quad\text{and}\quad
    \xi_1  = \frac{2(a+b)(1+ab)}{(1+b)^2(1-a^2)} - \frac{2\phi_{\star}}{1-\phi_{\star}^2}.
  \]
  Let  
  $\phi_{\star z} = {\phi_{\star}^2 (\phi_{\star}-a)}/(1-\phi_{\star}^2)$, 
  $\bar{\sigma}^2_{z} = \sum_{i=2}^{n} (Z_i-\bar{Z})^2 /(n-2)$, and $\bar{Z}=\sum_{i=2}^nZ_i/(n-1)$, 
  where $\{Z_i\}_{i=2}^{n}$ is defined in Definition~\ref{def:SPW}.
  We compare four estimators of $v$:  
  \begin{itemize}[topsep=3pt,itemsep=3pt,partopsep=3pt, parsep=3pt]
	\item[(a)] an $\AR(1)$-based parametric estimator, 
	\item[(b)] the unadjusted estimator with \cite{andrews1991}'s $\AR(1)$ plug-in bandwidth selector,
	\item[(c)] \cite{andrews1992}'s $\AR(1)$-prewhitened version of (b), and
	\item[(d)] the proposed $\AR(1)$-tail postcolored estimator.
	\end{itemize}
	They are, respectively,  computed as  
  \[
    \bar{v}_{\textsc{para}} = \frac{\bar{\sigma}^2_{z}}{(1-\bar{\phi})^2}, \quad
    \tilde{v}_{\textsc{un}} = \tilde{v}(\ell_{\textsc{un}}; K; X_{1:n}), \quad 
    \hat{v}_{\textsc{pw}} = \frac{\tilde{v}(\ell_{\textsc{pw}};K; Z_{2:n})}{(1-\bar{\phi})^2}, \quad
    \hat{v}_{\textsc{tail}} = \hat{v}(\ell_{\textsc{tail}}; K; \bar{\phi}),
  \]
  where the function $\tilde{v}(\cdot)$ is defined in (\ref{eqt:kernelEst}), 
  the function $\hat{v}(\cdot)$ is defined in (\ref{eqt:propDebias}), and 
  \begin{align*}
    \ell_{\textsc{un}} = \left\lceil{\left\{ {6\phi_{\star}^2 n}/{(1-\phi_{\star}^2)^2}  \right\}^{1/3}}\right\rceil, \quad
    \ell_{\textsc{pw}} = \left\lceil{\left\{ {6\phi_{\star z}^2n}/{(1-\phi_{\star z}^2)^2}\right\}^{1/3}}\right\rceil, \quad
    \ell_{\textsc{tail}} = \left\lceil \left( {3\xi_1^2 n }/{ 2} \right)^{1/3} \right\rceil.
  \end{align*}
    In the experiments, we set $a\in\{\pm 0.2, \pm 0.4, \pm 0.8\}$, $b\in\{0,-0.6\}$, and $n=400$. 
    Throughout this article, all experiments are replicated $5000$ times unless otherwise stated.

	First, the true optimal bandwidths are used  
    so as to study the effect solely driven by prewhitening and tail postcoloring.
    The results are shown in Table~\ref{tab:exp1}.
    When the model is well-specified, i.e., when $b=0$,
    both $\hat{v}_{\textsc{pw}}$ and $\hat{v}_{\textsc{tail}}$ achieve good bias-correction effects 
    relative to the unadjusted estimator $\tilde{v}_{\textsc{un}}$.
    These two estimators perform as well as the parametric estimator $\bar{v}_{\textsc{para}}$.
    When the model is misspecified, i.e., when $b=-0.6$,
    the inflation of risk caused by our proposed $\hat{v}_{\textsc{tail}}$ 
    is much weaker than for $\hat{v}_{\textsc{pw}}$.
    In addition, $\hat{v}_{\textsc{tail}}$ achieves a good bias correction effect for the downward bias of $\tilde{v}_{\textsc{un}}$
    while $\hat{v}_{\textsc{pw}}$
    leads to severe over-correction.
  The estimator $\bar{v}_{\textsc{para}}$ is inconsistent under model misspecation.

\begin{table}[t]
\centering
\small
\def~{\hphantom{0}}
\begin{tabular}{cccccccccc}
  &  & \multicolumn{4}{c}{$b=0$ (Well-specified case)} & \multicolumn{4}{c}{$b=-0.6$ (Misspecified case)} \\
  Criteria & $a\backslash$Estimator &  $\bar{v}_{\textsc{para}}$ & $\tilde{v}_{\textsc{un}}$ & $\hat{v}_{\textsc{pw}}$ & $\hat{v}_{\textsc{tail}}$ &  $\bar{v}_{\textsc{para}}$ & $\tilde{v}_{\textsc{un}}$ & $\hat{v}_{\textsc{pw}}$ & $\hat{v}_{\textsc{tail}}$ \\[0.5ex]
   $100\MSE(\cdot)/v^2$ 
    & $  -0.8$ & $0.64$ & $~9.28$ & $0.64$ & $0.64$ & $4146.19$ & $440.86$ & $~79.99$ & $40.03$ \\ 
    & $  -0.4$ & $0.93$ & $~3.69$ & $0.93$ & $0.93$ & $1836.48$ & $141.81$ & $~89.68$ & $18.16$ \\ 
    & $  -0.2$ & $1.15$ & $~2.15$ & $1.14$ & $1.14$ & $1073.58$ & $~99.79$ & $~93.91$ & $14.30$ \\ 
    & $  ~0.2$ & $1.94$ & $~2.47$ & $1.93$ & $1.94$ & $~215.87$ & $~45.33$ & $175.22$ & $~7.88$ \\ 
    & $  ~0.4$ & $2.78$ & $~4.36$ & $2.77$ & $2.79$ & $~~47.36$ & $~21.73$ & $~46.79$ & $~5.62$ \\ 
    & $  ~0.8$ & $9.13$ & $12.42$ & $9.10$ & $9.43$ & $~~29.22$ & $~22.24$ & $~31.08$ & $10.02$ \\[0.5ex]
    $10\{ \E(\cdot)/v - 1 \}$ 
    & $  -0.8$ & $~0.05$ & $~1.46$ & $~0.02$ & $-0.00$ & $64.08$ & $20.30$ & $~8.71$ & $~3.21$ \\ 
    & $  -0.4$ & $~0.03$ & $~0.87$ & $~0.00$ & $-0.01$ & $42.59$ & $11.60$ & $~9.25$ & $~2.01$ \\ 
    & $  -0.2$ & $-0.03$ & $~0.66$ & $-0.05$ & $-0.06$ & $32.53$ & $~9.71$ & $~9.48$ & $~1.76$ \\ 
    & $  ~0.2$ & $-0.03$ & $-0.98$ & $-0.05$ & $-0.04$ & $14.49$ & $~6.48$ & $13.03$ & $~1.32$ \\ 
    & $  ~0.4$ & $-0.06$ & $-1.31$ & $-0.08$ & $-0.05$ & $~6.65$ & $~4.36$ & $~6.60$ & $~1.13$ \\ 
    & $  ~0.8$ & $-0.42$ & $-2.56$ & $-0.44$ & $-0.25$ & $-5.34$ & $-4.61$ & $-5.52$ & $-2.22$ \\[0.5ex]
\end{tabular}
\caption{The standardized mean-squared error $100\MSE(\cdot)/v^2$
and standardized bias $10\{ \E(\cdot)/v - 1 \}$
of 
the parametric estimator $\bar{v}_{\textsc{para}}$,
the unadjusted estimator $\tilde{v}_{\textsc{un}}$,
the prewhitened estimator $\hat{v}_{\textsc{pw}}$, and 
the proposed tail postcolored estimator $\hat{v}_{\textsc{tail}}$
with their own theoretical optimal bandwidths; see Example \ref{exp:simu_ar1}.
The sample size is $n=400$.
}\label{tab:exp1}
\end{table}

\begin{table}[t]
\centering
\small
\def~{\hphantom{0}}
\begin{tabular}{cccccccccc}
  &  & \multicolumn{4}{c}{$b=0$ (Well-specified case)} & \multicolumn{4}{c}{$b=-0.6$ (Misspecified case)} \\
  Criteria & $a\backslash$Estimator &  $\bar{v}_{\textsc{para}}$ & $\tilde{v}_{\textsc{un}}$ & $\hat{v}_{\textsc{pw}}$ & $\hat{v}_{\textsc{tail}}$ &  $\bar{v}_{\textsc{para}}$ & $\tilde{v}_{\textsc{un}}$ & $\hat{v}_{\textsc{pw}}$ & $\hat{v}_{\textsc{tail}}$ \\[0.5ex]
   $100\MSE(\cdot)/v^2$ 
    & $  -0.8$ & $0.64$ & $~9.25$ & $0.92$ & $1.25$ & $4146.19$ & $444.95$ & $~78.66$ & $42.30$ \\  
    & $  -0.4$ & $0.93$ & $~3.82$ & $0.98$ & $1.40$ & $1836.48$ & $143.94$ & $103.13$ & $20.46$ \\  
    & $  -0.2$ & $1.15$ & $~2.36$ & $1.15$ & $2.00$ & $1073.58$ & $~98.08$ & $117.90$ & $16.30$ \\  
    & $  ~0.2$ & $1.94$ & $~2.67$ & $1.95$ & $2.50$ & $~215.87$ & $~46.30$ & $132.67$ & $10.38$ \\  
    & $  ~0.4$ & $2.78$ & $~4.58$ & $2.83$ & $2.96$ & $~~47.36$ & $~27.57$ & $~44.68$ & $~9.54$ \\  
    & $  ~0.8$ & $9.13$ & $13.17$ & $9.42$ & $9.48$ & $~~29.22$ & $~22.30$ & $~30.47$ & $11.78$ \\ [0.5ex]
    $10\{ \E(\cdot)/v - 1 \}$ 
    & $  -0.8$ & $~0.05$ & $~1.41$ & $-0.05$ & $-0.06$ & $64.08$ & $20.55$ & $~8.57$ & $~3.04$ \\  
    & $  -0.4$ & $~0.03$ & $~0.92$ & $-0.03$ & $-0.16$ & $42.59$ & $11.68$ & $~9.85$ & $~1.90$ \\  
    & $  -0.2$ & $-0.03$ & $~0.62$ & $-0.06$ & $-0.38$ & $32.53$ & $~9.59$ & $10.50$ & $~1.68$ \\  
    & $  ~0.2$ & $-0.03$ & $-0.94$ & $-0.05$ & $~0.11$ & $14.49$ & $~6.45$ & $11.05$ & $~1.18$ \\  
    & $  ~0.4$ & $-0.06$ & $-1.33$ & $-0.06$ & $-0.03$ & $~6.65$ & $~4.77$ & $~6.43$ & $~0.87$ \\  
    & $  ~0.8$ & $-0.42$ & $-2.51$ & $-0.37$ & $-0.29$ & $-5.34$ & $-4.53$ & $-5.45$ & $-2.09$ \\[0.5ex]
\end{tabular}
\caption{The standardized mean-squared error $100\MSE(\cdot)/v^2$
and standardized bias $10\{ \E(\cdot)/v - 1 \}$
of $\bar{v}_{\textsc{para}}$, $\tilde{v}_{\textsc{un}}$, $\hat{v}_{\textsc{pw}}$, and $\hat{v}_{\textsc{tail}}$
with their optimal bandwidths (if any) estimated by the parametric plug-in method; 
see Example \ref{exp:simu_ar1} and Remark \ref{rem:est_opt_bw}.
The sample size is $n=400$.
}
\label{tab:exp1_estimatedBW}
\end{table}

Next, 
the above simulation experiment is repeated; 
however, each estimator is computed with its estimated optimal bandwidth.
The bandwidth used for the proposal is  
based on the parametric plug-in estimator of $\xi_1$ as stated in Remark~\ref{rem:est_opt_bw}.
The results are shown in Table \ref{tab:exp1_estimatedBW}. 
The MSE and bias of $\tilde{v}_{\textsc{un}}$, $\hat{v}_{\textsc{pw}}$, and $\hat{v}_{\textsc{tail}}$ are only slightly inflated
compared to those in Table~\ref{tab:exp1},  
but the conclusions remain the same. 
We also perform additional simulation experiments when 
$b=0.6$, $n=200$ and under some nonlinear time series models; 
see {\ifnum\isXr=1{Sections \ref{sec:sim_MSE_estBW}, \ref{sec:sim_MSE_n200}, and \ref{sec:sim_nonL}}\else{Sections B.1, B.2 and B.4}\fi}  
of the supplement, respectively.

  \end{example}

\section{Generalization and extension}\label{sec:general}
\subsection{Multi-model tail postcolored estimator}\label{sec:dRobust}
As a parametric approach,
a reasonably good specification of the model is crucial for tail postcoloring to be beneficial.
In this section, we propose a general estimator that utilizes multiple tail postcoloring models.
We show that, under some conditions, as long as one of the models is well-specified,
the estimator achieves the highest efficiency among those coloring models.
The multiple robustness is ensured in this manner. 
  \begin{definition}\label{def:mulModel}
    Let $X^{[j]}_i(\theta_j) = h_{\theta_j}^{[j]}(\mathcal{F}_i)$ be the $j$th parametric tail postcoloring model
  where $J\in\mathbb{N}$,  
  $h_{\theta_1}^{[1]}, \ldots, h_{\theta_J}^{[J]}$ are measurable functions, and 
  $\theta_j \in \mathcal{H}_j \subseteq \mathbb{R}^{q_j}$ $(j=1,\ldots, J)$ 
  are parameter vectors for $q_j \in \mathbb{N}$. 
    The values of $v$, $v_p$ and $M_{\ell,K}$ under the $j$th model are  
    \[
      v^{[j]}({\theta}_j) = \sum_{k\in\mathbb{Z}} \gamma^{[j]}_k(\theta_j), \quad
    v^{[j]}_p({\theta}_j) = \sum_{k\in\mathbb{Z}} {|k|}^p\gamma^{[j]}_k(\theta_j), \quad
    M^{[j]}_{\ell,K}({\theta}_j) = \sum_{k\in\mathbb{Z}} K (k/\ell) \gamma^{[j]}_k(\theta_j), 
  \]
  respectively, 
    where $\gamma^{[j]}_k(\theta_j) = \Cov\{X^{[j]}_i(\theta_j), X^{[j]}_{i-k}(\theta_j)\}$.
    Assume $v^{[j]}({\theta}_j)\in\mathbb{R}^+$ for $\theta_j\in\mathcal{H}_j$ and $j=1,\ldots, J$.
    Suppose that 
    $\bar{\theta}_j$
    satisfies $n^{1/2}(\bar{\theta}_j - \theta_{\star j})=O_p(1)$ for some $\theta_{\star j}\in \mathcal{H}_j$ $(j=1,\ldots, J)$.
    Denote the tail postcolored estimator in (\ref{eqt:proposalsimple}) based on the $j$th model as
    \begin{equation}\label{eqt:submodelEst}
      \hat{v}^{[j]}_{\textsc{tail}} 
      	\equiv \hat{v}^{[j]}\left(\ell_{ j}; K; \bar{\theta}_j\right) = \frac{v^{[j]}(\bar{\theta}_j)}{M^{[j]}_{\ell_{ j},K}(\bar{\theta}_j)}
      \tilde{v}\left(\ell_{ j}; K\right).
    \end{equation}
    The optimal $\ell_{j}$ is $\ell_{\textsc{tail},j} \sim \{ pB^2 (\xi^{[j]}_{p})^2n/(2A)\}^{1/(2p+1)}$ with $\xi^{[j]}_{p} = \kappa_p - \kappa^{[j]}_p(\theta_{\star j})$ and $\kappa^{[j]}_p(\theta_{\star j}) = v^{[j]}_p(\theta_{\star j})/ v^{[j]}(\theta_{\star j})$; see (\ref{eqt:OptBW}).  
    Let $\hat{w}_1, \ldots, \hat{w}_J \geq 0$ be weights satisfying $\sum_{j=1}^{J}\hat{w}_j=1$.
    Denote 
    $\ell_{1:J} = \{\ell_{ 1}, \ldots, \ell_{ J}\}$ and 
    $\bar{\theta}_{1:J}=\{\bar{\theta}_1,\ldots, \bar{\theta}_J\}$.
    The multi-model tail postcolored estimator is 
    \begin{align}\label{eqt:generalWeightedEst}
      \hat{v}^{[1:J]}_{\textsc{tail}} 
      \equiv \hat{v}^{[1:J]}(\ell_{ 1:J}; K;\bar{\theta}_{1:J}) 
		=  \sum_{j=1}^{J} \hat{w}_j \hat{v}^{[j]}_{\textsc{tail}} . 
    \end{align}
  \end{definition}
The properties of the proposal are presented in the following proposition.
In particular, we show that if the weight $\hat{w}_j$ satisfies certain conditions,
when one of the $J$ models is well-specified,
the new estimator automatically selects the best model.
  \begin{proposition}\label{prop:dRobust}
    Suppose the conditions in Theorem~\ref{thm:BiasVar} and Definition~\ref{def:mulModel} hold and in addition we have
    $\sum_{j=1}^{J} u^{[j]}_p(\theta_{\star j}) = \sum_{j=1}^{J} \sum_{k\in\mathbb{Z}} {|k|}^p \gamma^{[j]}_k(\theta_{\star j}) < \infty$.
    Also assume that $\sup_{\theta\in\mathcal{H}_j} \left|\partial v^{[j]}(\theta)/\partial \theta \right| < \infty$
      and $\sup_{\theta\in\mathcal{H}_j} \left|\partial M^{[j]}_{\ell_{j},K}(\theta)/\partial \theta \right|<\infty$ for all $j=1,\ldots,J$.
    If there exist $w_1,\ldots, w_J\in[0,1]$ such that $\| \hat{w}_j - {w}_j\|  = o\left\{n^{-p/(2p+1)}\right\}$ for all $j=1,\ldots,J$
    and $\sum_{j=1}^J w_j = 1$, 
    then we have (i) 
    $$
              \lim_{h\to\infty} \lim_{n\to\infty} \MSE_h\left\{ \hat{v}^{[1:J]}(\ell_{\textsc{tail}, 1:J}; K;\bar{\theta}_{1:J})\right\}
              = \lim_{n\to\infty} n^{{2p}/{(2p+1)}} \MSE\left\{ \hat{v}^{[1:J]}(\ell_{\textsc{tail}, 1:J}; K;\theta_{\star 1:J}) \right\},
    $$  
    where $\ell_{\textsc{tail}, 1:J} = \{\ell_{\textsc{tail}, 1}, \ldots, \ell_{\textsc{tail},J}\}$ and 
    $\theta_{\star 1:J}=\{{\theta}_{\star 1},\ldots, {\theta}_{\star J}\}$; 
    and (ii)  
    \begin{align*}
              \limsup_{n\to\infty} n^{2p/(2p+1)}\Var\left\{\hat{v}^{[1:J]}(\ell_{\textsc{tail}, 1:J}; K;\theta_{\star 1:J})\right\} 
             &\leq 2p v^2  C^2 
             \left\{ \sum_{j=1}^{J} w_j {\left\vert {\xi^{[j]}_{p}}\right\vert}^{-1/(2p+1)}\right\}^2, \\
      n^{p/(2p+1)} \bias\left\{ \hat{v}^{[1:J]}(\ell_{\textsc{tail}, 1:J}; K;\theta_{\star 1:J})\right\} 
      &= v C
      \sum_{j=1}^{J} w_j \left\vert{\xi^{[j]}_{p}}\right\vert^{-1/(2p+1)} \sgn(B\xi^{[j]}_{p})  
      + o(1) , 
            \end{align*}
      where $C = \{ (2A/p)^p |B| \}^{1/(2p+1)}$ and
      $\sgn(\cdot)$ denotes the signum function.
  \end{proposition}

When $J=1$ and $w_1=1$, the estimator reduces to the single-model tail postcolored estimator in Corollary~\ref{cor:optimalBW}.
Some examples of weights $\{\hat{w}_j\}$ in (\ref{eqt:generalWeightedEst}) are given below:

  \begin{example}[Choice of weights]\label{rem:weightChoice}
    (i) (Simple average) A trivial option is to use a simple average, i.e., $\hat{w}_j = 1/J \ (j = 1,\ldots,J)$.
    Although using a data-independent weight cannot select the best model, 
    it is easy to implement and achieves an averaged effect of all models.

    (ii) (Adaptive choice) 
    Let $\hat{w}_j = {|\hat{\xi}^{[j]}_{p}|^{-2}}/{\sum_{j=1}^J |\hat{\xi}^{[j]}_{p}|^{-2}}$ and $w_j = |{\xi}^{[j]}_{p}|^{-2} / {\sum_{j=1}^J |{\xi}^{[j]}_{p}|^{-2}}$,
    where $\hat{\xi}^{[j]}_{p} = \tilde{v}_p/\tilde{v} - v^{[j]}_p(\bar{\theta}_j)/v^{[j]}(\bar{\theta}_j)$.
    By Proposition~\ref{prop:dRobust} with $\theta_{j} = \theta_{\star j}$ ($j=1, \ldots, J$),
    if there exists $j\in\{1,\ldots,J\}$ such that $\xi^{[j]}_{p} = 0$, 
    then $w_j = 1$ and $\limsup_{n\rightarrow\infty}\MSE( \hat{v}^{[1:J]}_{\textsc{tail} } )  / \MSE(\hat{v}^{[j]}_{\textsc{tail} }) \leq 1$.
  \end{example}

{
 \begin{example}[Two-model example]\label{eg:2model_AR1_MA5}
    Consider the following two coloring models:
    \begin{itemize}[topsep=3pt,itemsep=3pt,partopsep=3pt, parsep=3pt]
    \item[(i)] $\AR(1)$ model: $X_i = a_1 X_{i-1} + \varepsilon_i$ with independent $\varepsilon_i\sim \Normal(0,\sigma^2)$; and
      \item[(ii)] $\MA(5)$ model: $X_i = \sum_{j=1}^5 b_j \varepsilon_{i-j} + \varepsilon_i$ with independent $\varepsilon_i\sim \Normal(0,\varsigma^2)$.
  \end{itemize}
    Through maximum likelihood estimation,  
    $(a_1, \sigma^2)$ and $(b_1, \ldots, b_5, \varsigma^2)$ can be consistently estimated.  
    Let the data be generated as follows:
    $X_i = \sqrt{1-c^2}X_i^{(1)} + cX_i^{(2)}$, 
    where $c\in[0,1]$, and $X_i^{(1)}$ and $X_i^{(2)}$ 
    are independently generated from models (i) and (ii), respectively, 
    with $a_1=0.6$ and $(b_1, \ldots, b_5)=(0.6, 0, 0, 0.3, -0.3)$. 
    Their innovation variances $\sigma^2$ and $\varsigma^2$ are chosen such that their LRVs are equal to one. 
    Thus, $v=1$ for all $c\in[0,1]$.    
    The coloring models (i) and (ii) are well-specified only when $c=0$ and $c=1$, respectively.    
    The estimators defined in (\ref{eqt:submodelEst}) for models (i) and (ii)
    are denoted as $\hat{v}^{[1]}_{\textsc{tail}}$ and $\hat{v}^{[2]}_{\textsc{tail}}$, respectively.
    Additionally, we compute the combined estimator $\hat{v}^{[1:2]}_{\textsc{tail}}$ with weight function defined in Example~\ref{rem:weightChoice} (ii).
    All estimators are equipped with non-parametric plug-in bandwidths; see Remark \ref{rem:est_opt_bw}. 
    We set $n=200$ in the simulation. 
    The results are plotted in Figure~\ref{fig:2model}.

  \begin{figure}[t]
    \centering
    \includegraphics[width=0.8\textwidth]{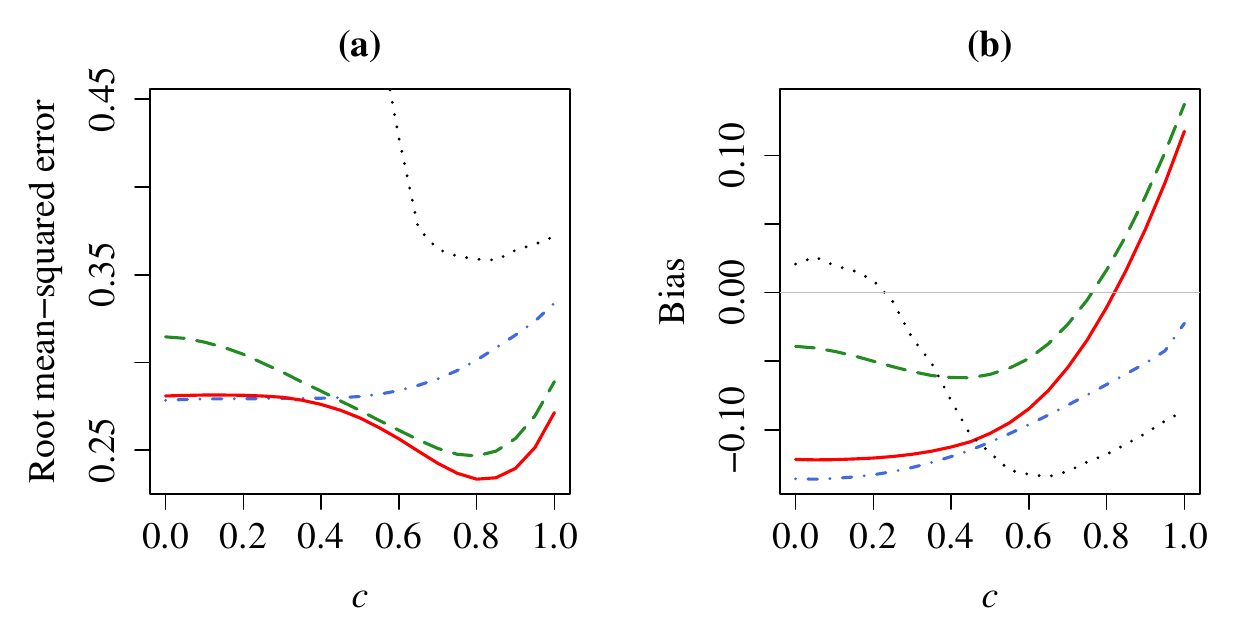}
    \vspace{-0.3cm}
    \caption{
    Plots (a) and (b) show the root mean-squared error and bias, i.e., $\MSE(\cdot)^{1/2}$ and 
     $\Bias(\cdot)$, of 
     $\hat{v}^{[1]}_{\textsc{tail}}$ (green dashed),
      $\hat{v}^{[2]}_{\textsc{tail}}$ (blue dot-dashed),  
      $\hat{v}^{[3]}_{\textsc{tail}}$ (black dotted),  
      and the proposed multi-model tail postcolored estimator $\hat{v}^{[1:2]}_{\textsc{tail}}$
     (red solid) against different $c$, respectively.
     The value of $\MSE(\hat{v}^{[3]}_{\textsc{tail}})^{1/2}$, which ranges from 0.36 to 0.72, 
     is obviously larger the root mean squared errors of the other estimators.
     Thus, it is not plotted in panel (a) when $c\lesssim 0.5$. 
     }\label{fig:2model}
  \end{figure}

  We observed that $\hat{v}^{[1:2]}_{\textsc{tail}}$ 
  automatically switches between the better coloring models and   
  perform better or nearly as well as both single-model estimators $\hat{v}^{[1]}_{\textsc{tail}}$ and $\hat{v}^{[2]}_{\textsc{tail}}$
  in terms of MSE.
  The bias of $\hat{v}^{[1:2]}_{\textsc{tail}}$ lies between that of $\hat{v}^{[1]}_{\textsc{tail}}$ and $\hat{v}^{[2]}_{\textsc{tail}}$, 
  which aligns with intuition.
  For reference, we also consider the following coloring model: 
  \begin{itemize}[topsep=3pt,itemsep=3pt,partopsep=3pt, parsep=3pt]
    \item[(iii)] $\ARMA(1,5)$: 
      $X_i = a'_1 X_{i-1} + \sum_{j=1}^5 b'_j \varepsilon_{i-j} + \varepsilon_i$ with independent 
      $\varepsilon_i\sim \Normal(0,\varpi^2)$. 
  \end{itemize}
  Denote the tail postcolored estimator based on model (iii) as $\hat{v}^{[3]}_{\textsc{tail}}$.
  This model covers models (i) and (ii) as special cases, 
  but $\hat{v}^{[3]}_{\textsc{tail}}$ may not outperform $\hat{v}^{[1]}_{\textsc{tail}}$ and $\hat{v}^{[2]}_{\textsc{tail}}$, 
  as observed in Figure~\ref{fig:2model}.
  When $c\in(0,1)$, the sum of independent $\AR(1)$ and $\MA(5)$ is not an $\ARMA(1,5)$ process. 
  Moreover, the more complicated model (iii) may not fit well when $n$ is small. 
  In some cases, fitting multiple simpler models is easier and more stable than fitting one complicated model. 
  This illustrates the usefulness of multi-model tail postcoloring. 
  \end{example}
}

\subsection{Generalized tail postcolored estimator}\label{sec:generalEst}
In (\ref{eqt:proposalsimple}), we use the same kernel $K$ in the quantity $M_{\ell, K}$
as in computing $\tilde{v}(\ell; K )$.
In fact, they do not have to be the same.
A general form of tail postcolored estimator is
  \begin{equation}\label{eqt:proposalgeneral}
    \hat{v}(\ell ; K; H; \bar{\theta}) = \frac{v(\bar{\theta})}{M_{\ell, H}(\bar{\theta})} \tilde{v}(\ell;K),
  \end{equation}
where $H\in\mathcal{K}_{p'}$.
So, (\ref{eqt:proposalgeneral}) reduces to (\ref{eqt:proposalsimple}) if $H=K$, i.e., 
$\hat{v}(\ell ; K; K; \bar{\theta}) = \hat{v}(\ell ; K; \bar{\theta})$. 
The asymptotic properties of $\hat{v}(\ell ; K; H; \bar{\theta})$ are shown in the following theorem:

  \begin{theorem}[Generalized estimator]\label{thm:BiasVarGeneral}
    Let $K\in\mathcal{K}_{p}$ and $H\in\mathcal{K}_{p'}$ for some $p \in \mathbb{N}$ and $p'\in\bar{\mathbb{N}}$.
    Let $B'=\lim_{t\downarrow 0} \{ H(t) - H(0) \}/ t^{p'}$.
    Assume that $\bar{\theta}$ satisfies $n^{1/2} (\bar{\theta} - \theta_{\star}) = O_p(1)$ for some $\theta_{\star} \in \mathcal{H}$ 
    as $n\to\infty$.
    Assume
    $u_{p'}(\theta_{\star}) < \infty$,
    $\sup_{\theta\in\mathcal{H}}\left|{\partial v(\theta)}/{\partial \theta}\right|<\infty$ and $\sup_{\theta\in\mathcal{H}} \left| {\partial M_{\ell,H}(\theta)}/{\partial \theta} \right|<\infty$.
    If Assumption~\ref{assp:AspSeries} holds, $u_{p}< \infty$, $u_{p'}(\theta_{\star})<\infty$ and $\ell \asymp n^{1/(2p+1)}$,  
    we have 
    (i)
    $$
      \lim_{h\to\infty} \lim_{n\to\infty} \MSE_h\{\hat{v}(\ell ; K; H; \bar{\theta})\}
        = \lim_{n\to\infty} n^{{2p}/({2p+1})} \MSE\{\hat{v}(\ell ; K; H; \theta_{\star})\};
    $$
    (ii) 
    $\Var\left\{ \hat{v}(\ell ; K; H; {\theta}) \right\} \sim \Var\left\{ \tilde{v}(\ell ; K ) \right\}$ for all $p'$; 
    and (iii)
    \begin{equation}\label{eqt:generalbias}
      \bias\left\{ \hat{v}(\ell ; K; H; \theta_{\star}) \right\}
      = \left\{ \begin{array}{ll}
      	\displaystyle \frac{Bv}{\ell^p}\left\{ \kappa_p - \frac{B'\ell^{p-p'}}{B} \kappa_{p'}(\theta_{\star}) \right\} + o\left(\frac{1}{\ell^p}\right) + O\left(\frac{\ell}{n}\right), & \text{if $p'\in \mathbb{N}$}; \\
		\bias\left\{ \tilde{v}(\ell ; K;X_{1:n}) \right\} \{1+o(1)\}, 
		& \text{if $p'=\infty$}. 
		\end{array} \right.
    \end{equation}

  \end{theorem}

Theorem~\ref{thm:BiasVarGeneral} suggests that generalized tail postcoloring 
does not affect the variance of the estimator asymptotically, 
while the bias may exhibit different properties  
depending on the kernel order $p'$. 
The following corollary presents the bias behaviors in different cases: 

  \begin{corollary}\label{cor:TailBVgeneral}
    Under the conditions of Theorem~\ref{thm:BiasVarGeneral}, consider three cases:

      (i) When $p' > p$, $\bias\left\{ \hat{v}(\ell ; K; H; \theta_{\star}) \right\} \sim \bias\left\{ \tilde{v}(\ell;K) \right\}$;

      (ii) When $p' = p$, 
          $\bias\left\{ \hat{v}(\ell ; K; H; \theta_{\star}) \right\} = Bv\left\{ \kappa_p - B'\kappa_{p'}(\theta_{\star}) / B\right\}/\ell^p + o\left({1}/{\ell^p}\right) + O({\ell}/{n})$;

      (iii) When $p' < p$, $\bias\left\{ \hat{v}(\ell ; K; H; \theta_{\star}) \right\} = O(1/\ell^{p'})$.
  \end{corollary}

Case (i) suggests that tail postcoloring has a negligible effect asymptotically
when the kernel $H$ has a higher order than $p$; see Example~\ref{eg:negligibleTailPW}
for a special case.
It protects users from using bandwidth that is too short without any risks of model misspecification.
Case (iii) is not recommended as it 
leads to a sub-optimal rate of convergence.
Case (ii) is our recommended proposal and would be studied in the simulation experiments.

  \begin{example}\label{eg:negligibleTailPW}
    Let $H(t)=\mathbb{1}(|t| \leq 1)$.
    The \textsc{ar}$(1)$-tail postcolored estimator is
        \[
        	\hat{v}^{\circ}_{\textsc{tail}} = \frac{ 1 + {2\bar{\phi}^{\ell+1}}}{ {1+\bar{\phi}-2\bar{\phi}^{\ell+1}} } \tilde{v},
		\]
		where $\bar{\phi}=\tilde{\gamma}_1/\tilde{\gamma}_0$.
    The postcoloring coefficient still satisfies $\eta_{\ell, H}(\bar{\theta}) = \hat{v}^{\circ}_{\textsc{tail}}/\tilde{v}  \to 1$ but 
    with a higher convergence rate compared with the case where $p'=p$.
  \end{example}

\subsection{Multivariate time series}\label{sec:MVextension}
Tail postcoloring applies readily to multivariate data.
We consider a $d$-dimensional stationary and ergodic time series
$\{ X_i \in\mathbb{R}^d\}_{i\in\mathbb{Z}}$ with ${\mu} = \E({X}_1)$ and $d\in\mathbb{N}$.
The autocovariances are ${\Gamma}_{k} = {\Gamma}^{\T}_{-k} = \E\{( X_k-{\mu})( X_0-{\mu})^{\T}\}$ for $k\in\mathbb{Z}$,
and the sample autocovariances are 
\[
	\tilde{\Gamma}_k = \tilde{\Gamma}^{\T}_{-k} = \frac{1}{n}\sum_{i=k +1}^{n} ({X}_i - \bar{{X}}_n)({X}_{i-k } - \bar{{X}}_n)^{\T}
\]
for $k = 0,\ldots,n-1$.
The long-run covariance matrix that we want to estimate is
$$V = \lim_{n\to\infty}n\Var(\bar{{X}}_n).$$ 

Let $C^{(u,v)}$ denote the $(u,v)$th component of a matrix $C$
and $c^{(u)}$ denote the $u$th entry of a vector $c$.
The physical dependence measure for the $u$th component is 
$\delta^{(u)}_{p,i}= { \| X^{(u)}_i-  {X}'^{(u)}_i \| }_p$,
where $X'_i=g(\mathcal{F}'_i)$ is the coupled series of $X_i = g(\mathcal{F}_i)$ 
and $\mathcal{F}_i$ and $\mathcal{F}'_i$ are defined as in Section~\ref{sec:notation}.
Let $\Delta^{(u)}_p=\sum_{i=0}^{\infty}\delta^{(u)}_{p,i} \ (u=1, \ldots, d)$.
We extend $v_p$ and $u_p$ to 
${V}_p=\sum_{k\in\mathbb Z}{|k|}^{p}{\Gamma}_{k}$
and
${U}_p=\sum_{k\in\mathbb Z}{|k|}^{p}|{\Gamma}_{k}|$, 
where $|{\Gamma}_{k}|$ denotes the entry-wise absolute value of ${\Gamma}_k$. 

\begin{definition}[Covariance matrix estimator]\label{def:CovVar}
Let $X_i(\theta) = h_\theta(\mathcal{F}_i)$ be a parametric tail postcoloring model, 
  where $h_\theta$ is a measurable function and the autocovariances are $\Gamma_k(\theta) = \Cov\{X_i(\theta), X_{i-k}(\theta)\}$ for $k\in\mathbb{Z}$
  such that $\sum_{k\in\mathbb{Z}}\Gamma_k(\theta)$ is positive definite for $\theta \in \mathcal{H}$.
  Let $\bar{\theta} \in \mathcal{H}$ be an estimator of $\theta$, 
  where $\mathcal{H} \subseteq \mathbb{R}^{r}$ is the parameter space of $\theta$ and $r\in\mathbb{N}$.
  Let 
  \[
  	\tilde{V}(\ell;K) = \sum_{|k| <n} K(k/\ell) \tilde{ \Gamma}_{k}, \quad
	V(\bar{\theta}) = \sum_{k \in\mathbb{Z}}{\Gamma}_k(\bar{\theta}), \quad
  	M_{\ell,K}(\bar{\theta}) = \sum_{k\in\mathbb{Z}} K(k/\ell) {\Gamma}_k(\bar{\theta}).
  \]
  Assume that $M_{\ell,K}(\theta)$ is invertible for all $\theta\in\mathcal{H}$.
  The tail postcolored estimator for $V$ is 
  \begin{align}\label{eqt:MVproposal}
    \hat{V}_{\textsc{tail}} 
    = \hat{V}(\ell;K;\bar{\theta}) 
    = \bigg\{ V(\bar{\theta}) M_{\ell,K}^{-1}(\bar{\theta}) \tilde{V}(\ell;K) 
    +
    \tilde{V}(\ell;K)  
    M_{\ell,K}^{-1}
    (\bar{\theta}) V(\bar{\theta})    
    \bigg\} / 2.
  \end{align}
\end{definition}

Note that symmetry is ensured with the formula in (\ref{eqt:MVproposal}).
The conclusions in Theorem~\ref{thm:BiasVar} can be extended to the covariance estimator
in Proposition~\ref{prop:multivariateBiasVar}.
To do this, we
define the weighted MSE and its truncated version for the matrix estimator $\hat{V}$ for $V$ as
\begin{align*}
	\WMSE(\hat{V}; W) &= \E\left\{ \vect(\hat{V} - V)^{\T} W \vect(\hat{V} - V) \right\}, \\
	\WMSE_h(\hat{V}; W) &= \E\left[ \min\left\{ n^{ {2p}/(2p+1)} \vect(\hat{V} - V)^{\T} W \vect(\hat{V} - V) , h\right\} \right], 
\end{align*}
where $W$ is a $d^2\times d^2$ weight matrix, $h\geq 0$, and 
$\vect(M) = (M_1^{\T}, \ldots, M_d^{\T})^{\T}$ 
is the column-by-column vectorization of the matrix $M = (M_1 \cdots M_d)$
with columns $M_1, \ldots, M_d$.
We denote $\otimes$ as the Kronecker product,  
$I_{dd}$ as the $d^2\times d^2$ identity matrix, and 
$C_{dd} = \sum_{u=1}^d \sum_{v=1}^d (e_ue_v^{\T}) \otimes (e_ve_u^{\T})$ as the $d^2\times d^2$ commutation matrix, 
where $e_u\in\mathbb{R}^d$ is the $u$th elementary vector.
The properties of $\hat{V}_{\textsc{tail}}$ defined in (\ref{eqt:MVproposal}) are shown below:

  \begin{proposition}\label{prop:multivariateBiasVar}
  Let $K\in\mathcal{K}_{p}$ for some $p \in \mathbb{N}$. 
    Assume  
    $\E \{  \vert X_i^{(u)} \vert^{\omega}  \}<\infty$ for some $\omega>4$, 
    $U_{p}^{(u,u)}<\infty$, 
    and $\Delta^{(u)}_4 < \infty$ for $u = 1,\ldots,d$.
    Suppose $n^{1/2} (\bar{\theta}- {\theta}_{\star}) = O_p(1)$ 
    for some ${\theta}_{\star} \in \mathcal{H}$, 
    $\sup_{\theta\in\mathcal{H}} \left| {\partial V^{(u,v)}(\theta)}/{\partial  \theta^{(h)}} \right| <\infty$ and 
    $\sup_{\theta\in\mathcal{H}} \left| {\partial M^{(u,v)}(\theta)}/{\partial  \theta^{(h)}} \right| < \infty$ for all $u,v,h\in\{1, \ldots, d\}$.
    Let $\ell \asymp n^{1/(2p+1)}$ as $n\to\infty$,
    then (i) 
    \[
      \lim_{h\to\infty} \lim_{n\to\infty} \WMSE_h\left\{ \hat{V}(\ell;K;\bar{\theta}); W\right\}
      = \lim_{n\to\infty} n^{{2p}/{(2p+1)}} \WMSE\left\{ \hat{V}(\ell;K;{\theta}_{\star}); W \right\};
    \]
    (ii) $\Var\left[ \vect\{\hat{V}(\ell;K;\theta_{\star}) \} \right] \sim 
          \Var\left[ \vect\{\tilde{V}(\ell;K) \}\right]$ for all $u,v=1,\ldots, d$, where
    \begin{equation}\label{eqt:mvvariance}
      \Var\left[ \vect \{\tilde{V}(\ell;K) \}\right]
      = 2A (I_{dd} + C_{dd}) (V\otimes V) \frac{\ell}{n} + o\left(\frac{\ell}{n}\right) ; 
    \end{equation}  
    and (iii)
    \begin{equation}\label{eqt:mvbias}
      \bias\left\{ \hat{V}(\ell;K;\theta_{\star}) \right\}
      =
      \frac{B}{\ell^p} \frac{1}{2}\left( V\Xi_p^- + \Xi_p^+V\right) + o\left(\frac{1}{\ell^p}\right) + O\left(\frac{\ell}{n}\right),
    \end{equation}
    where $\Xi_p^- = V^{-1}V_p - V^{-1}(\theta_{\star}) V_p(\theta_{\star})$ and $\Xi_p^+ = V_pV^{-1} - V_p(\theta_{\star}) V^{-1}(\theta_{\star})$.
  \end{proposition}

\begin{remark}\label{rem:positiveDef}
The estimator in (\ref{eqt:MVproposal}) may not be positive-definite as in many bias-corrected estimators. 
To ensure positive-definiteness of $\hat{V}_{\textsc{tail}} = \hat{V}(\ell;K;\bar{\theta})$, 
we lower bound the eigenvalues of $\hat{V}_{\textsc{tail}}$ by $1/n$ as in \cite{politis2011}.
Denote $\{\lambda_i \}_{i=1}^{d}$ and $\{s_j \}_{j=1}^{d}$ as the eigenvalues and eigenvectors 
of $\hat{V}_{\textsc{tail}}$
such that $\hat{V}_{\textsc{tail}} = Q\Lambda Q^{\T}$,
where $Q$ is a $d\times d$ matrix with columns $s_1, \ldots, s_d$ and 
$\Lambda = \diag(\lambda_1, \ldots,\lambda_d)$ is a diagonal matrix with diagonal entries 
$\lambda_1, \ldots,\lambda_d$.
The adjusted estimator is $\hat{V}_{\textsc{tail}}^{+} = Q\Lambda^{+} Q^{-1}$ 
where $\Lambda^+ = \diag\{ \max(\lambda_1, 1/n), \ldots, \max(\lambda_d, 1/n)\}$. 
{Note that this correction step may incur additional finite-sample bias; 
see, e.g., \cite{politis2011} and \cite{LiuChan2022} for discussions and some empirical results.} 
\end{remark}

\subsection{Robustness to heteroskedasticity}\label{sec:supp_het}
  Now we do not assume stationarity while other conditions in Proposition~\ref{prop:multivariateBiasVar} still hold.
  The physical dependence measures are changed to $\delta_{p}^{\mathdag} = \sup_i\| X_i-X_i'\| _p$. 
  Let $\Delta_p^{\mathdag}=\sum_{i=0}^{\infty}\delta_{p}^{\mathdag}$.
  Instead of estimating $V =\lim_{n\to\infty}n\Var(\bar{{X}}_n)$, 
  we consider estimation of its finite-$n$ version:
  \[
    V_n = n\Var(\bar{{X}}_n)
    = \frac{1}{2n}\sum_{k=-n+1}^{n-1}\sum_{i=|k| +1}^{n} \E\{({X}_i-{\mu})({X}_{i-|k| }-{\mu})^{\T} + ({X}_{i-|k| }-{\mu})({X}_i-{\mu})^{\T}\}.
  \]
  Similar to  $V_{p}$ and $U_{p}$, 
  define the counterparts 
  $V_{p,\mathdag}$ and $U_{p,\mathdag}$ such that, for all $u,v=1,\ldots,d$,
    \begin{align*}
      V^{(u,v)}_{p,\mathdag} = \sum_{k\in\mathbb{Z}}|k|^p
          \sup_{i\geq 1}\Cov\left\{ X^{(u)}_{i+k}, X^{(v)}_{i} \right\} 
         \qquad \text{and} \qquad
      U^{(u,v)}_{p,\mathdag} = \sum_{k\in\mathbb{Z}}|k|^p
          \sup_{i\geq 1} \left\vert \Cov\left\{ X^{(u)}_{i+k}, X^{(v)}_{i} \right\} \right\vert. 
    \end{align*}
	The results for our proposed $\hat{V}_{\textsc{tail}}$ in Proposition~\ref{prop:multivariateBiasVar} can be modified as follows:

\begin{proposition}\label{prop:hacRobust}
  Assume all conditions in Proposition~\ref{prop:multivariateBiasVar} with the following modifications:
  Let $\sum_{u=1}^{d}\sum_{v=1}^{d}U^{(u,v)}_{p,\mathdag}<\infty$
  but $\{X_i\}$ is no longer stationary.
  Suppose that $\sqrt{n}(\bar{\theta} - \theta_{\star}) = O_p(1)$ still holds 
  for some $\theta_{\star}\in\mathcal{H}$
  and 
  the series $X_i(\theta) = h_{\theta}(\mathcal{F}_i)$
  satisfies Assumption~\ref{assp:AspSeries} for any $\theta\in\mathcal{H}$.
  (i) If $1/\ell + \ell/n^{1/(p+1)} \to 0$, then
  \begin{gather*}\label{eqt:HACBV}
    \limsup_{n\to\infty}\frac{n}{\ell}\Var\{\hat{V}^{(u,v)}(\ell;K; \theta_{\star})\} \leq 2A \left[ V_{0,\mathdag}^{(u,u)} V_{0,\mathdag}^{(v,v)} + \left\{V_{0,\mathdag}^{(u,v)}\right\}^2\right], \\
    \limsup_{n\to\infty} \ell^p\left\mid \E\{\hat{V}^{(u,v)}(\ell;K; \theta_{\star})\}- V^{(u,v)}_n\right\mid 
     \leq |B| \left\{U^{(u,v)}_{p,\mathdag} + \Psi^{(u,v)}_{p,\mathdag}\right\} .
  \end{gather*}
  for all $u,v\in\{1,\ldots,d\}$
  where $
  \Psi_{p,\mathdag} = \left\vert V_p(\theta_{\star})V^{-1}(\theta_{\star}) V_{0,\mathdag} 
  			+ V_{0,\mathdag} V^{-1}(\theta_{\star})V_p(\theta_{\star}) \right\vert /2$.
  (ii) If 
  $\ell\asymp n^{1/(2p+1)}$, then $n^{p/(2p+1)} \left\{ \hat{V}(\ell;K; \bar{\theta}) - \hat{V}(\ell;K; \theta_{\star}) \right\} = o_p(1)$.
\end{proposition}

The rate-optimality of $\tilde{V}(\ell;K)$ 
under heteroskedasiticity is shown in, e.g., \cite{LiuChan2022}.  
By Proposition~\ref{prop:hacRobust}, 
together with slight modifications of the kernel function as in the proof of Theorem~\ref{thm:BiasVar}, 
we know that 
$\hat{V}(\ell;K; \bar{\theta})$ is also rate-optimal when $\ell \asymp n^{1/(2p+1)}$.

\begin{remark}\label{rem:revisionByCasini2022}
Under non-stationarity, the theories for 
    the unadjusted estimator $\tilde{v}_{\textsc{un}}$ and the prewhitened estimator $\hat{v}_{\textsc{pw}}$
    were also developed in \cite{andrews1991} and \cite{andrews1992}, 
    where the results were recently discussed and revised in \cite{casini2022}.
\end{remark}

\section{Discussion and comparison with similar approaches}\label{sec:discussion}
Combining parametric and nonparametric estimators is a common technique in statistics; 
see \citet{tsiatis2006} for examples of semiparametric methods. 
This approach takes advantage of both types of methods. 
In the literature, there is a vast variety of such approach; 
however, they combine parametric and nonparametric components in different ways. 
Some examples that are closer to our setting are discussed as follows:
\cite{andrews1992}’s prewhitening approach tries to nonparametrically fit the parametric residuals. 
The key idea is to use a parametric model to handle most serial dependence and 
then use a nonparametric tool to cope with the remaining unmodeled structure. 
On the other hand, our approach parametrically fits the unhandled tail autocovariance leftover by the nonparametric tool. 
Hence, the order of using parametric and nonparametric methods is the reverse of the standard prewhitening approach; 
thus, it is named postcoloring.
Our principle favors trusting the nonparametric tool in principle and 
asking a parametric model to provide auxiliary assistance on the tail structure. 
Such assistance becomes relevant when the asymptotic effect of the nonparametric tool has not yet kicked in. 
It appears to us that this idea is new in the literature.
In the context of density estimation and nonparametric regression, there are many classes of hybrid estimators. 
Some examples are listed below:
\begin{itemize}[topsep=3pt,itemsep=3pt,partopsep=3pt, parsep=3pt]
	\item The class of weighted estimators uses a weighted average 
			or geometric average of parametric and nonparametric estimators; 
			see \cite{OlkinSpiegelman1987}. 
			This principle is certainly applicable for estimating the long-run variance; 
			however, the optimal weights are not known. 
			Compared to our approach, we use the bandwidth parameter, 
			which is already needed in the nonparametric component, 
			to bridge parametric and nonparametric estimators. Our approach appears to be more automatic. 
	\item The class of ratio estimators is first studied in \citet{HjortGlad1995}. 
			It tries to use a nonparametric correction factor to adjust a parametric initial estimator; 
			see also \citet{glad1998} and \citet{MishraSuUllah2010} for a similar approach, 
			and \citet{Martins-FilhoMishra2008Ullah} for a generalized version. 
			In principle, it is similar to \cite{andrews1992}’s prewhitening approach, 
			as both nonparametrically correct a parametric start. 
			Hence, they differ from our principle of parametrically correcting a nonparametric start. 
			On the other hand, this approach shares a similar spirit with ours, as both use a multiplicative correction factor. 
			Moreover, our approach is designed specifically 
			for handling the tail autocovariance in strongly correlated time series.
	\item The class of bias-corrected estimators is considered by \citet{DelaigleHall2014}. 
				It modifies a nonparametric estimator by subtracting a parametric estimator of its bias. 
				This approach shares a similar favor with ours as both attempt some form of correction parametrically. 
				However, our approach explicitly aims at correcting the bias due to tail autocovariances 
				in $\tilde{v}_{\textsc{un}}$ via multiplication instead of subtraction. 
\end{itemize}
Compared to these existing approaches, we provide a new perspective 
for integrating parametric and nonparametric estimators
into one improved estimator. 

Employing multiple models to assist a main estimator  
is also a common technique. 
It shares a similar spirit with model averaging; 
see \citet{ClaeskensHjort2008} for a review. 
For example, \citet{hansen2007} discussed averaging least squares estimators obtained from multiple models 
and showed that it achieves the lowest possible squared error among the individual estimators considered. 
This is similar to Example \ref{rem:weightChoice}, 
where the multi-model tail postcolored estimator achieves the smallest mean squared error 
among all individual estimators considered. 
Besides, it shares a similar principle to doubly robust estimation \citep{RobinsRotnitzkyZhao1994,ScharfsteinRotnitzkyRobins1999} 
and multiply robust estimation \citep{han2014} in the context of missing data. 
These methods are consistent if any one of the models is correctly specified, 
whereas our estimator achieves the parametric rate of convergence 
if any one of the parametric models is correctly specified. 
Hence, they provide multiple protections for our nonparametric estimator.
As far as we concern, our proposal is the only available method 
that achieves this property in the context of long-run variance estimation.  

\section{Simulation experiments and real-data application}\label{sec:simulation}

\subsection{HAC estimation}\label{subsec:HAC}

We consider the linear regression model: $Y_i = {X}^{\T}_i {\beta}_0 + \varepsilon_i$,
where ${X}_i \in\mathbb{R}^d \ (i = 1,\ldots,n)$ are the regressors
and  ${\beta}_0\in\mathbb{R}^d$ is a vector of coefficients.
Similar models are also considered in \cite{andrews1991}, \cite{andrews1992} and \cite{vats2022lugsail}. 
The least-square estimator of ${\beta}_0$ is
  $
      \hat{\beta} = S_n^{-1}  \sum_{i=1}^{n} {X}_iY_i, 
  $
  where $S_n = \sum_{i=1}^{n} {X}_i {X}_i^{\T}/n$. 
It satisfies that 
  \begin{equation*}\label{eqt:RegressVar}
      \Sigma_n \equiv  \Var\left\{ n^{1/2}(\hat{\beta} - {\beta}_0) \right\} 
       = S_n^{-1}
      V_n 
      S_n^{-1}, 
      \quad \text{where} \quad
      V_n = \frac{1}{n} \sum_{i=1}^{n} \sum_{j=1}^{n} \E\left(\varepsilon_i {X}_i {X}_j^{\T} \varepsilon_j^{\T}\right).
  \end{equation*}
If $\hat{V}$ is a long-run covariance matrix estimator applied on  
$\{(Y_i - {X}^{\T}_i \hat{\beta}) {X}_i\}_{i=1}^{n}$, 
then $\Sigma_n$ can be estimated by 
$\hat{\Sigma} = S_n^{-1} \hat{V} S_n^{-1}$.
We are interested in testing $H_0: {\beta}_0 = 0$ against $H_1: {\beta}_0 \neq 0$.
We reject $H_0$ if 
$n\hat{\beta}^\T\hat{\Sigma}^{-1}\hat{\beta} > d(n-1)F_{1-\alpha,d,n-d}/(n-d)$,
where $F_{1-\alpha,r_1,r_2}$ is the $(1-\alpha)$th quantile of $F$-distribution
with degrees of freedom $r_1$ and $r_2$.
We study four candidates for $\hat{V}$:
\begin{enumerate}[topsep=3pt,itemsep=3pt,partopsep=3pt, parsep=3pt]
	\item[(a)] the unadjusted estimator $\tilde{V}_{\textsc{un}}$ with \cite{andrews1991}'s 
    $\textsc{ar}(1)$ plug-in bandwidth,
	\item[(b)] \cite{andrews1992}'s $\textsc{var}(1)$-prewhitened version of (a) denoted as $\hat{V}_{\textsc{pw}}$,  
	\item[(c)] \cite{Flegal:2021uk}'s lugsail kernel estimator $\tilde{V}_{\textsc{lugs}}$, and  
	\item[(d)] the proposed $\textsc{var}(1)$-tail postcolored estimator $\hat{V}_{\textsc{tail}}$,
where $\textsc{var}(1)$ refer to order-1 vector \textsc{ar} model; 
see {\ifnum\isXr=1{Section~\ref{sec:simulation_multi}}\else{Section B.5}\fi}  in the supplement for the detailed formulas. 
\end{enumerate}

The estimators $\tilde{V}_{\textsc{un}}$, $\hat{V}_{\textsc{pw}}$, and $\hat{V}_{\textsc{tail}}$ use $K=K_{\Bart}$, 
whereas $\tilde{V}_{\textsc{lugs}}$ uses the default kernel in \cite{Flegal:2021uk}. 
The resulting tests are denoted as tests (a)--(d), respectively. 
We remark that our proposed estimator is robust to heteroskedasticity;
see Proposition~\ref{prop:hacRobust}. 

In the simulation experiments, 
we set $\varepsilon_i = 5 (i/n)^2 \varepsilon_i'  \ (i = 1,\ldots,n)$, 
where $\varepsilon'_i$ is a unit-variance process generated from the stationary \textsc{arma}$(1,1)$ model,
i.e., $\varepsilon_i' = \varepsilon_i''/c$ and 
$\varepsilon_i'' = a \varepsilon_{i-1}'' + e_i + b e_{i-1}$, where 
$|a|, |b|<1$,  
$c^2 = 1+(a+b)^2/(1-a^2)$,
and $e_i \sim \Normal(0,1)$ independently.
Hence, the noise $\varepsilon_i$ is  heteroskedastic.
The covariates are $X_i =  ( 1, {X'_i}^{\T} )^{\T}$, 
where $X_i'\in\mathbb{R}^2$ are correlated and are 
generated from the \textsc{varma}$(1,1)$ model: 
$X_i' = \Phi X_{i-1}' + \Upsilon \epsilon_{i-1} + \epsilon_{i}$,
where $\epsilon_{i} \sim \Normal_2(0, \Sigma_{\epsilon})$ independently and 
$\Phi, \Upsilon,\Sigma_{\epsilon} \in \mathbb{R}_{2\times 2}$ such that  
\[
	\Phi = \begin{bmatrix} 0.7 &0 \\0 & 0.7 \end{bmatrix},\quad
	\Upsilon = \begin{bmatrix} 0.3 & 0.1 \\ 0.1 & 0.3 \end{bmatrix},\quad
	\Sigma_{\epsilon} = \begin{bmatrix} 1 & 0.5 \\ 0.5 & 1\end{bmatrix}.
\]
Let ${\beta}_0 = (0,\delta,0)^{\T}$ and $n=400$.
Three cases are considered: $a=b\in\{0.2, 0.4, 0.8\}$, denoted as Cases 1--3, respectively, corresponding to weak, medium, and strong serial dependence.
This data structure is challenging for prewhitened estimators using \textsc{var}$(1)$ model
as the autocorrelation structure is substantially different from that of the \textsc{var}$(1)$ model. 

The power curves are shown in the upper row of plots in Figure~\ref{fig:hartest}. 
In Case 1, all tests perform similarly with accurate size and promising power. 
However, when the serial dependence increases to Cases 2--3, 
the size distortion of tests (a)--(c) becomes increasingly severe. 
In particular, the unadjusted test (a) and the lugsial test (c) are oversized, 
while the classical prewhitened test (b) is undersized.
Our proposed tail postcolored test (d) remains size accurate. 
This makes our proposed test more powerful than test (c) and 
more reliable than tests (a) and (b).
We also plot the empirical distribution functions of $p$-values of the tests under $H_0$; 
see the lower row of plots in Figure~\ref{fig:hartest}. 
The empirical distribution of the $p$-value of our proposed test is closest to 
the a standard uniform distribution $\Unif(0,1)$. 
This explains the high size accuracy.
See {\ifnum\isXr=1{Section~\ref{sec:supp_Simulation}}\else{Section B}\fi}  of the supplement for more experiments.

\begin{figure}[t]
  \begin{center}
    \includegraphics[width=\textwidth]{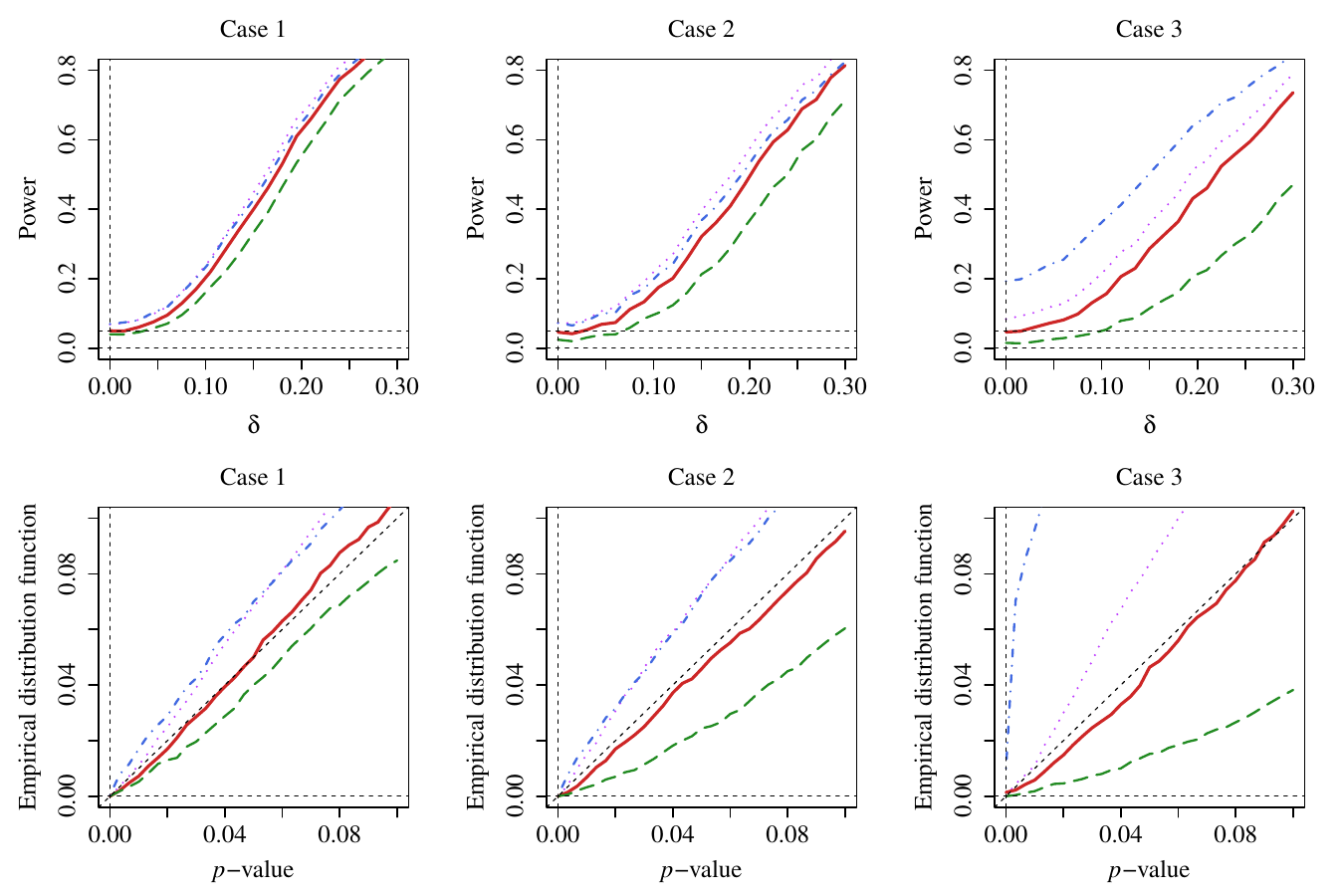}
    \vspace{-0.3cm}
    \caption{Upper row: Power curves for testing $\beta_0 = 0$, 
    where the black dashed horizontal lines designate the nominal size 5\% and 0. 
    Lower row: Empirical distribution functions of the $p$-values under $H_0$, 
    where the black dashed line is the distribution function of $\Unif(0,1)$. 
    Tests (a)--(d) use 
    $\tilde{V}_{\textsc{un}}$ (purple dotted curves),  
    $\hat{V}_{\textsc{pw}}$ (green dashed curves),
    $\tilde{V}_{\textsc{lugs}}$ (blue dash-dotted curves), and 
    $\hat{V}_{\textsc{tail}}$ (red solid curves), respectively.     
    Cases 1--3 correspond to $a=b \in\{0.2, 0.4, 0.8\}$, respectively. 
    }
    \label{fig:hartest}
  \end{center}
\end{figure}

\subsection{Convergence test for Bayesian tobit quantile regression}

We implement the fixed-width output analysis \citep{Galin2006} to determine the chain length $n$
in a Markov chain Monte Carlo simulation problem.
The convergence test relies on the central limit theorem
$n^{1/2}(\bar{g}_n - \E_{\pi} g) \Rightarrow \Normal(0,v_g)$,
where $v_g$ is the long-run variance of $\bar{g}_n$ and $\E_{\pi} g$ is the target estimand.
We terminate the simulation at the first time when
  \begin{equation}\label{eqt:FWterm}
    z_{1-\alpha/2} \left( {\hat{v}_g}/{n}\right)^{1/2} + \epsilon \mathbb{1}(n\leq n^{\dagger}) < \epsilon,
  \end{equation}
where $\hat{v}_g$ is an estimator of $v_g$,
$\alpha\in(0,1)$ is the nominal size and 
$z_{1-\alpha/2}$ is the $(1-\alpha/2)100$th quantile of standard normal distribution.
The target half-width is denoted by $\epsilon$ and 
$n^{\dagger}$ is the minimum sample size set in order to prevent early termination.

\begin{table}[t]
  \centering
  \small
  \def~{\hphantom{0}}
  \begin{tabular}{ p{1.7cm} c cccccc}
    & $\epsilon$ & Proposal & (a) Original  & (b) Andrews &   (c) \textsc{obm}($n^{1/2}$)  & (d) \textsc{obm}($n^{1/3}$) & (e) coda \\
    \multirow{4}{2cm}{Coverage Probability  ($\%$)} 
    & $3.0$ & $99.2$ & $98.8$ & $98.6$  & $98.9$ & $96.2$ & $98.8$ \\ 
    & $4.5$ & $98.9$ & $97.8$ & $98.3$  & $97.8$ & $96.7$ & $97.2$ \\ 
    & $6.0$ & $99.0$ & $98.3$ & $98.7$  & $98.7$ & $95.2$ & $91.8$ \\ 
    & $7.5$ & $99.0$ & $97.7$ & $98.8$  & $97.9$ & $93.6$ & $81.2$ \\ [1ex]
    \multirow{4}{2cm}{Average number of iterations} 
    & $3.0$ & $11356$ & $10693$ & $10450$  & $11000$ & $~8628$  & $10526$ \\   
    & $4.5$ & $~5116$ & $~4699$ & $~4782$  & $~4848$ & $~3583$  & $~3894$ \\   
    & $6.0$ & $~2937$ & $~2662$ & $~2786$  & $~2716$ & $~1927$  & $~1458$ \\   
    & $7.5$ & $~1952$ & $~1728$ & $~1872$  & $~1753$ & $~1177$  & $~~638$ \\ [1ex]
    \multirow{4}{2cm}{MSE of $\bar{\beta}$ ($\times 10^{-4}$)}
    & $3.0$ & $14.38$ & $15.58$ & $16.01$ & $14.92$ & $~19.97$  & $~15.63$ \\  
    & $4.5$ & $32.35$ & $35.89$ & $35.06$ & $34.59$ & $~46.65$  & $~43.29$ \\  
    & $6.0$ & $57.04$ & $64.26$ & $61.63$ & $61.60$ & $~88.54$  & $119.02$ \\  
    & $7.5$ & $86.39$ & $98.06$ & $90.74$ & $99.03$ & $138.80$  & $270.34$ 
  \end{tabular}
  \caption{Results of convergence tests including (1) coverage probabilities of the 99\% confidence interval $\left[\bar{\beta} - z_{1-\alpha/2}\surd{\hat{v}_g/n}, \bar{\beta} + z_{1-\alpha/2}\surd{\hat{v}_g/n}\right]$, 
  (2) the average number of Gibbs iterations for meeting the termination criterion, and 
  (3) the mean-squared error of  $\bar{\beta}_3$ at the point of termination.  
}\label{tab:mcmc}
\end{table}

We consider a Gibbs sampler for Bayesian tobit quantile regression in \cite{kozumi2011gibbs}. 
Suppose the response $y_i$ is left-censored at $0$ and let $y^{\star}_i$ be the latent response.
Then we model the response as 
$y_i = \max(y^{\star}_i,0)$ and $y^{\star}_i = x^{\T}_i \beta + \varepsilon_i$ $(i = 1,\ldots,n)$.
Assume that the noise follows an asymmetric Laplace distribution with median zero. 
Representing $\varepsilon_i$ as a location-scale mixture of normal random variables and assuming a normal prior for the coefficients,
the hierarchical model with parameters $(\beta, \sigma, z)$ considered is as follows:
$\varepsilon_i=\theta z_i + \tau{(\sigma z_i)}^{1/2}u_i$, where $\theta=(1-2p)/\{p(1-p)\}$, $\tau^2=2/\{p(1-p)\}$, 
$z_i\mid \sigma \sim \sigma \Exp(1)$, 
$u_i \sim \Normal (0,1)$, 
$\sigma\sim\text{IG}(n_0/2, s_0/2)$ and $\beta \sim \Normal (\mu_0, \Sigma_{0})$.
The Gibbs sampler iterates as follows: 
\begin{align*}
  y^{\star}_i \mid y_i, \beta, \sigma, z_i &\sim y_i \mathbb{1}(y_i>0) + 
                      \Normal (x_i^{\T}\beta + \theta z_i , \tau^2\sigma\theta z_i)\mathbb{1}(y_i=0);\\
  \beta \mid y_i^{\star}, z_i, \sigma &\sim \Normal(\hat{\mu}_{\beta}, \hat{\Sigma}_{\beta}), \quad\text{where }
  \hat{\Sigma}_{\beta}^{-1}=\sum_{i=1}^{n} \frac{x_i x_i^{\T}}{\tau^2\sigma z_i} + \Sigma_0^{-1}\quad\text{and} \quad
  \hat{\mu}_{\beta} = \hat{\Sigma} \left\{ \sum_{i=1}^{n} \frac{x_i(y_i-\theta z_i)}{\tau^2\sigma z_i} \right\}; \\
  \sigma \mid y_i^{\star}, \beta, z_i &\sim \text{IG}\left[
  \frac{n_0+3n}{2},
          \frac{1}{2}\left\{ s_0 + 2\sum_{i=1}^{n}z_i + \sum_{i=1}^{n}\frac{(y_i - x_i^{\T}\beta-\theta z_i)^2}{\tau^2z_i}\right\} \right]. 
\end{align*}
Here $\text{IG}(a,b)$ denotes the inverse gamma distribution with shape and scale parameters $a$ and $b$.
The algorithm is implemented using the function \texttt{Blqtr} in the R package \texttt{Brq} \citep{Alhamzawi2018}.
The regression data we used is the women labour data in \cite{mroz1984sensitivity}.
The dataset \texttt{Mroz} can be found in the R package \texttt{Ecdat} \citep{CroissantGraves}.
The following covariates for $n=753$ observations are selected for the quantile regression:
the number of children less than 6 years old in household,
the number of children between ages 6 and 18 in household,
wife's age,
actual years of wife's previous labor market experience and 
wife's educational attainment (in years).

In particular, the convergence test is applied on the 
coefficient $\hat{\beta}^{(3)}$ with $\alpha=0.01$, i.e., the third element of the estimated coefficient vector. 
  The evaluation in (\ref{eqt:FWterm}) is done every 500 iterations with $n^{\dagger}=100$. 
  The burn-in size of the chain is $3\times 10^4$. 
  The true value of $\beta$ is approximated by simulating $100$ chains of length $10^7$ and use the mean of their sample averages.
  We compare our $\textsc{ar}(1)$ tail postcolored estimator with the following estimators: 
  \begin{enumerate}[topsep=3pt,itemsep=3pt,partopsep=3pt, parsep=3pt]
	\item[(a)] \texttt{lrvar} in the R package \texttt{"sandwich"} 
			with Bartlett kernel and prewhitening.
  	\item[(b)] The same estimator in (a) is used except that tail postcoloring is not used.

	\item[(c)] \texttt{olbm} in package \texttt{mcmc} \citep{mcmcPack} which is the overlapping batch means (\textsc{obm}) estimator with a batch size of $\lfloor{n^{1/2}}\rfloor$.
	\item[(d)] The same estimator in (c) is used except that 
				the batch size is $\lfloor{n^{1/3}}\rfloor$.
	\item[(e)] \texttt{spectrum0} in  package \texttt{coda} \citep{coda2006} which considers a generalized linear regression on the spectrum \citep{heidelberger1981spectral} and \texttt{order=0}.

  \end{enumerate}

The results for the convergence test are displayed in 
Table~\ref{tab:mcmc}.
For all choices of the half-width $\epsilon$,
our estimator converges the latest and has the most sufficient coverage
probabilities as well as the smallest mean-squared error of coefficient estimates.

\vspace{-0.4cm}
\section*{Acknowledgements}
\vspace{-0.4cm}
{
This research was partially supported by grants
GRF-14306421 and 14307922
provided by the Research Grants Council of HKSAR.
The authors would like to thank the referees, an associate editor, and the editor for their constructive comments that 
improved the scope of the paper.
The authors report there are no competing interests to declare.
}

\newpage 

\renewcommand{\thefigure}{\thesection.\arabic{figure}}
\renewcommand{\thetable}{\thesection.\arabic{table}}
\appendix

\section{Further discussions}\label{sec:supp_furtherDiscussion}

{
\subsection{Comparison with the classical prewhitening approach}\label{sec:comparison_pw}
Throughout this subsection, 
we write 
$v^X = v$, $v_p^X = v_p$, $\gamma_k^X = \gamma_k$ and $\kappa_p^X = \kappa_p$. 

\begin{remark}[Improvement condition]\label{rem:uni}
Consider the classical prewhitening estimator \citep{andrews1992}
using the model $X_{i}(\theta) = h_{\theta}(\mathcal{F}_i) = g_{\theta}(Z_{1:n})$, where $Z_{1:n} = \{Z_i\}_{i=1}^{n}$ denotes the prewhitened time series.
Let $v^X = \lim_{n\rightarrow\infty}n\Var(\bar{X}_n)$ and $v^Z = \lim_{n\rightarrow\infty}n\Var(\bar{Z}_n)$ 
be the LRVs for the series $\{X_i\}$ and $\{Z_i\}$, respectively. 
Suppose that $v^X$ and $v^Z$ are related as $v^{X} = G(\theta) v^{Z}$, where 
$G: \mathcal{H} \to \mathbb{R}$ denotes the recoloring coefficient.
Then 
the prewhitening estimator can then be written as 
\[
  \hat{v}_{\textsc{pw}}(\ell; K; \bar{\theta})
  = G(\bar{\theta})\tilde{v}(\ell; K; Z_{1:n}).
\]
Recall that $\bar{\theta}$ is a $n^{1/2}$-consistent estimator for $\theta$
and $n^{1/2}(\bar{\theta} - \theta_{\star}) = O_{p}(1)$
for some $\theta_{\star} \in {\mathcal{H}}$.
In \cite{andrews1992},
it has been proved (can also be similarly demonstrated under our framework) that 
the classical prewhitening estimator in Definition~\ref{def:SPW}
has the asymptotic property 
\begin{align*}
	\lim_{h\to\infty} \lim_{n\to\infty} \MSE_h\{\hat{v}_{\textsc{pw}}(\ell_{\textsc{pw}}; K; \bar{\theta})\} 
	&=  \lim_{n\to\infty} n^{2p/(2p+1)} \MSE\{\hat{v}_{\textsc{pw}}(\ell_{\textsc{pw}}; K; \theta_{\star})\}  \\ 
	&= (2p+1)\{{(2A/p)}^p B \kappa^{Z}_{p}\}^{2/(2p+1)} (v^X)^2,
\end{align*}
where $\kappa^{Z}_{p} = v^{Z}_{p}/v^{Z}$
denotes the dependence ratio on the prewhitened time series
and 
\[
  \ell_{\textsc{pw}} \sim \left\{ p B^2 (\kappa^{Z}_{p})^2 n /(2A) \right\}^{1/(2p+1)}
\]
denotes the MSE-optimal bandwidth for $\hat{v}_{\textsc{pw}}(\ell; K; \bar{\theta})$,
assuming the conditions in Theorem~\ref{thm:BiasVar} hold. 
Note that $\ell_{\textsc{pw}}$ is indeed also the optimal bandwidth for the kernel estimator $\tilde{v}(\ell; K; Z_{1:n})$ of $v^{Z}$.
Therefore,
the asymptotic MSE is reduced after classical prewhitening if and only if 
\begin{equation}\label{eqt:pw_range}
  |\kappa^{Z}_p| < |\kappa^{X}_p|
\end{equation}
whereas the asymptotic MSE is reduced after tail postcoloring if and only if 
\begin{equation}\label{eqt:ptw_range}
  |\xi_p| < |\kappa^{X}_p|.
\end{equation}

Note that when the MSE-optimal bandwidths are used, 
we have 
\[
	\frac{\MSE(\hat{v}_{\textsc{tail}})}{\MSE(\tilde{v})} 
	\sim \frac{{\Bias}^2(\hat{v}_{\textsc{tail}})}{{\Bias}^2(\tilde{v})} 
	\sim \frac{\Var(\hat{v}_{\textsc{tail}})}{\Var(\tilde{v})}. 
\] 
Hence, the condition for tail postcoloring to reduce mean squared error
is equivalent to the condition for it to reduce squared bias and variance.
Indeed, the optimal MSE, squared bias and variance satisfy the following ratios asymptotically: 
\begin{align}\label{eqt:ratioMSEBiasVar}
	\Var(\hat{v}_{\textsc{tail}}) \sim \frac{2q}{1+2q}\MSE(\hat{v}_{\textsc{tail}}) 
	\qquad \text{and} \qquad
	{\Bias}^2(\hat{v}_{\textsc{tail}}) \sim \frac{1}{1+2q}\MSE(\hat{v}_{\textsc{tail}}) .
\end{align}
Therefore, it suffices for us to determine the cases where tail postcoloring reduces asymptotic MSE. 
This phenomenon is similar to the classical prewhitening. 

\end{remark}

With the results in Remark~\ref{rem:uni},
we can compare the MSEs, biases, and variances of the estimators 
$\tilde{v}_{\textsc{un}}$, $\hat{v}_{\textsc{pw}}$ and $\hat{v}_{\textsc{tail}}$. 
In particular, we consider three specific models to illustrate the results:
\begin{itemize}
	\item \textsc{arma}(1,1): Example \ref{rem:CompTPW} in the main text as well as Examples \ref{exp:ar11} and \ref{eg:biasComparison} below. 
	\item \textsc{ar}(2): Example \ref{exp:ar2}.
	\item \textsc{ma}(2): Example \ref{exp:ma2}.
\end{itemize}

\begin{example}[\ARMA$(1,1)$ data]\label{exp:ar11}\leavevmode
Suppose the data are generated from the \ARMA$(1,1)$ model: $X_i = a X_{i-1} + \varepsilon_i + b \varepsilon_{i-1}$ 
where $a,b\in(-1,1)$ and $\varepsilon_i \sim \Normal(0,1)$ independently.
Note that the autocovariance $\gamma^{X}_k = \Cov(X_0, X_k)$ can be found as follows:
\begin{align}\label{eqt:acvf_AR11}
	\gamma_0^X = \frac{1+b^2 + 2ab}{1-a^2} , \qquad 
    \gamma_1^X = \gamma_{-1}^X = \frac{(a+b)(1+ab)}{1-a^2}, \qquad 
    \gamma_k^X = a^{|k|-1} \gamma^{X}_1 \quad (|k|\geq 2). 
\end{align}
If we use the \AR$(1)$ model, i.e.,  
$X_i = \phi X_{i-1} + Z_i$ for $Z_i \simIID \Normal(0,\sigma^2)$, 
to perform the prewhitening as well as tail postcoloring procedures,
we have 
\begin{align}
  \gamma^{Z}_k  
  &= \Cov(Z_0, Z_k)
  = (1+\phi^2) \gamma^{X}_k
	  - \phi\left(\gamma^{X}_{k-1} + \gamma^{X}_{k+1}\right) \nonumber\\
  &= \begin{cases}
    \{(1+\phi^2)(1+b^2+2ab) - 2\phi(a+b)(1+ab)\}/(1-a^2) \qquad & \text{if $k=0$};\\
    \{(1+\phi^2-a\phi)(a+b)(1+ab) - \phi(1+b^2+2ab)\}/(1-a^2) \qquad & \text{if $|k|=1$};\\
    \left\{(1-\phi^2)a - \phi(1+a^2)\right\}a^{|k|-2}(a+b)(1+ab)/(1-a^2) \qquad & \text{if $|k|\geq 2$}.\\
  \end{cases} \label{eqt:acvf_z_AR11}
\end{align}
it can be verified by straightforward algebras that 
\begin{align}\label{eqt:vX_vZ}
  v^{Z}_0 = (1-\phi)^2 v^{X}_0  \qquad \text{and} \qquad 
  v^Z_1 = (1-\phi)^2 v_1^X - 2\phi\gamma_0^X, 
\end{align}
where 
$v_p^X = \sum_{k\in\mathbb{Z}}|k|^p\gamma_k^X$ and 
$v_p^Z = \sum_{k\in\mathbb{Z}}|k|^p\gamma_k^Z$ for $p\in\mathbb{N}_0$. 
We remark that both relationships in (\ref{eqt:vX_vZ}) are general for all data generating mechanism. 
Putting (\ref{eqt:acvf_AR11}) and (\ref{eqt:acvf_z_AR11}) into (\ref{eqt:vX_vZ}), we can evaluate $v_0^X$, $v_0^Z$, $v_1^X$, and $v_1^Z$ in close forms:
\begin{align*}
  v^X_0 &= \frac{(1+b)^2}{(1-a)^2}, &
  v^Z_0 &= \frac{(1-\phi)^2(1+b)^2}{(1-a)^2}, \\
  v^X_1 &=
  \frac{2(a+b)(1+ab)}{(1-a^2)(1-a)^2}, &
  v^Z_1 
  &=
  \frac{2\left\{ (a+b)(1+ab)\,(\phi-1)^2 -\phi\,(1+b^2+2ab)\,(1-a)^2\right\}}{(1-a^2)(1-a)^2}.
\end{align*}
For tail postcoloring, the quantity $\xi_1 = \kappa_1^X - \kappa_1(\phi)$ 
is important for determining the improvement region stated in (\ref{eqt:ptw_range}), where
\[
	\kappa_1^X \equiv \frac{v_1^X}{v_0^X} = \frac{2\{ a + b + a(a+b)^2 /(1-a^2)\}}{(1+b)^2}
	\qquad \text{and} \qquad 
	\kappa_1(\phi) = \frac{v_1(\phi)}{v_0(\phi)} = \frac{2\phi}{1-\phi^2}. 
\] 

By checking (\ref{eqt:pw_range}),
the prewhitening technique leads to a reduction in MSE 
if $\phi\in(-1,1)$ satisfies that 
$v_1^X \phi^3 - (\gamma_0^X + 2 v_1^X) \phi^2 + v_1^X \phi >0$.
Solving the inequality, we obtain  
\begin{align*}
	\left\{ U^-\mathbb{1}(0<U< 1/4)-\mathbb{1}(U\geq 1/4)\right\}\mathbb{1}(a+b<0) 
		<\phi<  U^+ \mathbb{1}(a+b>0), 
\end{align*}
where $U^{\pm} = \{ (2U \pm 1) \mp \sqrt{1\pm 4U} \}/(2U)$ and $U=2|(a+b)(1+ab)|/\{(1+b^2+2ab)(1-a)^{2}\}$.

On the other hand,
by checking (\ref{eqt:ptw_range}),
the tail postcoloring technique leads to a reduction in MSE 
if $\phi\in(-1,1)$ satisfies that $\kappa_1^X \phi^3 + \phi^2 - \kappa_1^X \phi <0$. 
Solving the inequality, we obtain  
\[
  \frac{-1+\{4(\kappa^X_1)^2 + 1\}^{1/2}}{2\kappa^X_1}\mathbb{1}(\kappa^X_1<0)
     < \phi <
     \frac{-1+\{4(\kappa^X_1)^2 + 1\}^{1/2}}{2\kappa^X_1}\mathbb{1}(\kappa^X_1>0).
\]

In particular, if we use $\phi=\gamma^X_1/\gamma^X_0 = (a+b)(1+ab)/(1+b^2+2ab)$, 
the the difference in mean-squared error after prewhitening and tail postcoloring are visualized in 
Figure~\ref{fig:ar1pw1}.
Similar to Figure~\ref{fig:ar1pw1}, we also show the difference in squared bias and 
the difference in variance in Figure \ref{fig:ar1pw1_biasVar}. 
Owing to the properties stated in (\ref{eqt:ratioMSEBiasVar}), 
the behaviors of squared bias and variance are the similar to that of MSE when the optimal bandwidths are used.

\begin{figure}[t]
    \begin{center}
          \includegraphics[width=\linewidth]{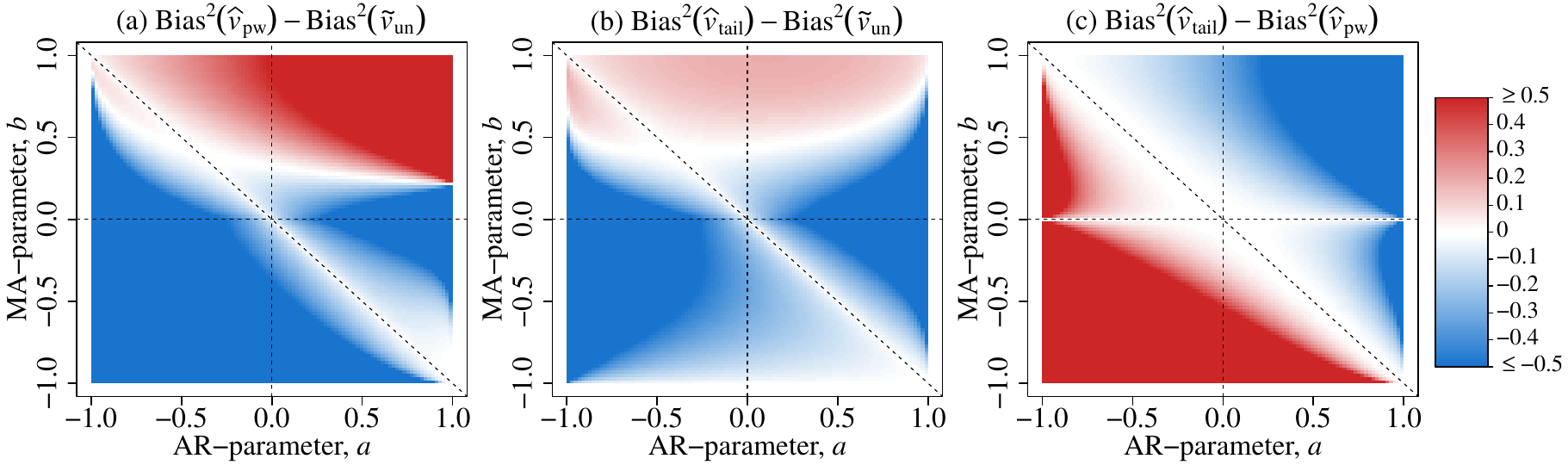}
    \end{center}
    \begin{center}
          \includegraphics[width=\linewidth]{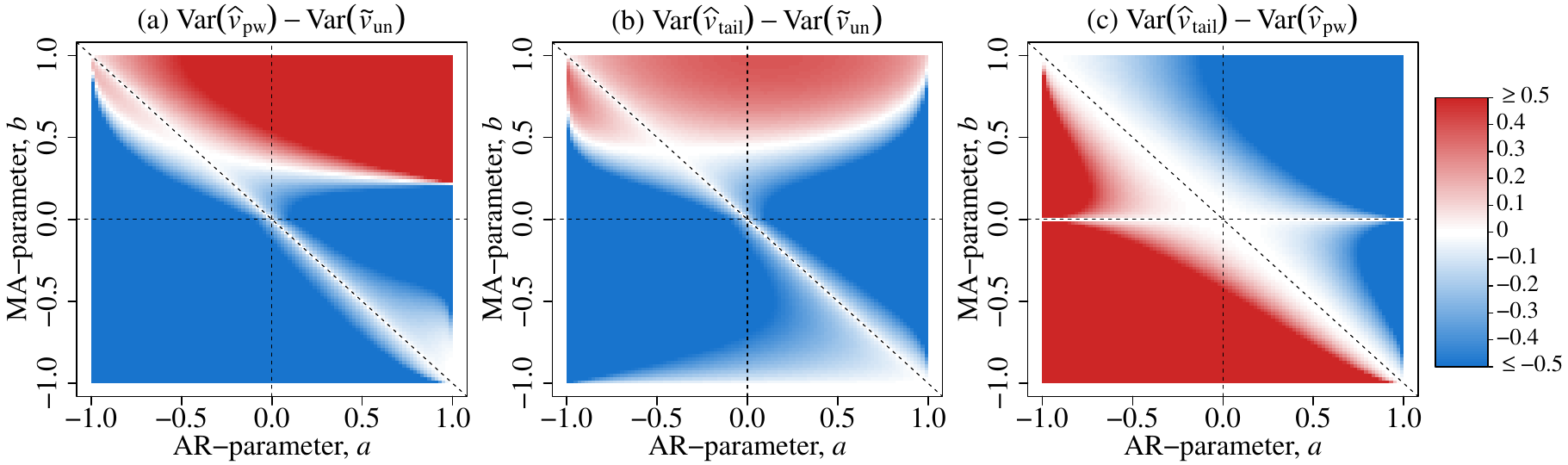}
    \end{center}
    \vspace{-0.3cm}
        \caption{\label{fig:ar1pw1_biasVar}Heat maps display the asymptotic differences of the squared bias and variances 
        after standardizing by the true value of $v^2$, 
        i.e., 
        $\{ \Bias^2(\cdot)-\Bias^2(\cdot) \}n^{2/3}/v^2$ and 
        $\{ \Var(\cdot)-\Var(\cdot) \}n^{2/3}/v^2$, 
        in the first row and second row of plots, respectively.  
        The first column, second column, and third column compare 
        $(\hat{v}_{\textsc{pw}},\tilde{v}_{\textsc{un}})$, 
        $(\hat{v}_{\textsc{tail}},\tilde{v}_{\textsc{un}})$, and  
        $(\hat{v}_{\textsc{tail}},\hat{v}_{\textsc{pw}})$, respectively. 
        See Example \ref{exp:ar11} for the details. 
        A difference less than zero (i.e., blueish regions) indicates an improvement
        of the first estimator over the second estimator in the difference. 
        Note that the comparison of MSE is shown in Figure~\ref{fig:ar1pw1} of the main text. 
        Note also that the titles of plots do not show the standardization $n^{2/3}/v^2$ due to space constraints. 
        }
\end{figure}

\end{example}

\begin{example}[Bias comparison for \ARMA(1,1) data]\label{eg:biasComparison}
In this example, we compare the asymptotic biases of 
the standard estimator $\tilde{v}$, prewhitten estimator $\hat{v}_{\textsc{pw}}$, 
and tail post colored estimator $\hat{v}_{\textsc{tail}}$. 
Both $\hat{v}_{\textsc{pw}}$ and $\hat{v}_{\textsc{tail}}$ employ $\AR(1)$ whitening/coloring model. 
All estimators are equipped with their respective optimal bandwidthes. 

Similar to Example \ref{exp:ar11}, 
let the data be generated from an $\ARMA(1,1)$ model:
$X_i = a X_{i-1} + \varepsilon_i + b \varepsilon_{i-1}$ 
where $\varepsilon_i \sim \Normal(0,1)$ independently. 
The limiting values of $n^{1/3}\Bias(\cdot)/v$ of the three estimators are shown in Figure \ref{fig:ar1pw1_biasSign}. 
The estimators $\tilde{v}$, $\hat{v}_{\textsc{pw}}$ and $\hat{v}_{\textsc{tail}}$ have 
positive asymptotic biases if and only if 
$(a,b)\in\mathcal{B}^+_{\textsc{un}}$, $(a,b)\in\mathcal{B}^+_{\textsc{pw}}$, 
and $(a,b)\in\mathcal{B}^+_{\textsc{tail}}$, respectively, where
\begin{align*}
	\mathcal{B}^+_{\textsc{un}} 
		&= \left\{ (a,b)\in(-1,1)^2:  -v_1 > 0 \right\} 
		= \left\{ (a,b)\in(-1,1)^2:  a+b<0  \right\} ;\\
	\mathcal{B}^+_{\textsc{pw}} 
		&= \left\{ (a,b)\in(-1,1)^2:  -v_1^Z > 0 \right\} \nonumber\\
		&= \left\{ (a,b)\in(-1,1)\times(0,1):  a+b>0  \right\} \cup \left\{ (a,b)\in(-1,1)\times(-1,0):  a+b<0 \right\}; \\
	\mathcal{B}^+_{\textsc{tail}} 
		&= \left\{ (a,b)\in(-1,1)^2:  -\xi_1 > 0 \right\} \nonumber\\
		&= \left\{ (a,b)\in(-1,1)\times(0,1):  a+b>0  \right\} \cup \left\{ (a,b)\in(-1,1)\times(-1,0):  a+b<0 \right\}.
\end{align*}
Note that $\mathcal{B}^+_{\textsc{pw}}=\mathcal{B}^+_{\textsc{tail}}$, 
hence, the signs of the asymptotic biases of $\hat{v}_{\textsc{pw}}$ and $\hat{v}_{\textsc{tail}}$ are equal.
It can be observed that $\hat{v}_{\textsc{pw}}$ and $\hat{v}_{\textsc{tail}}$ have positive asymptotic biases 
when $a,b>0$, which are arguably the most commonly encountered situation in practice. 
For variance estimation, over-estimation is less risky than having an under-estimation 
as it leads to more conservative rather than more aggregative inference on the mean. 

\begin{figure}[t]
    \begin{center}
          \includegraphics[width=\linewidth]{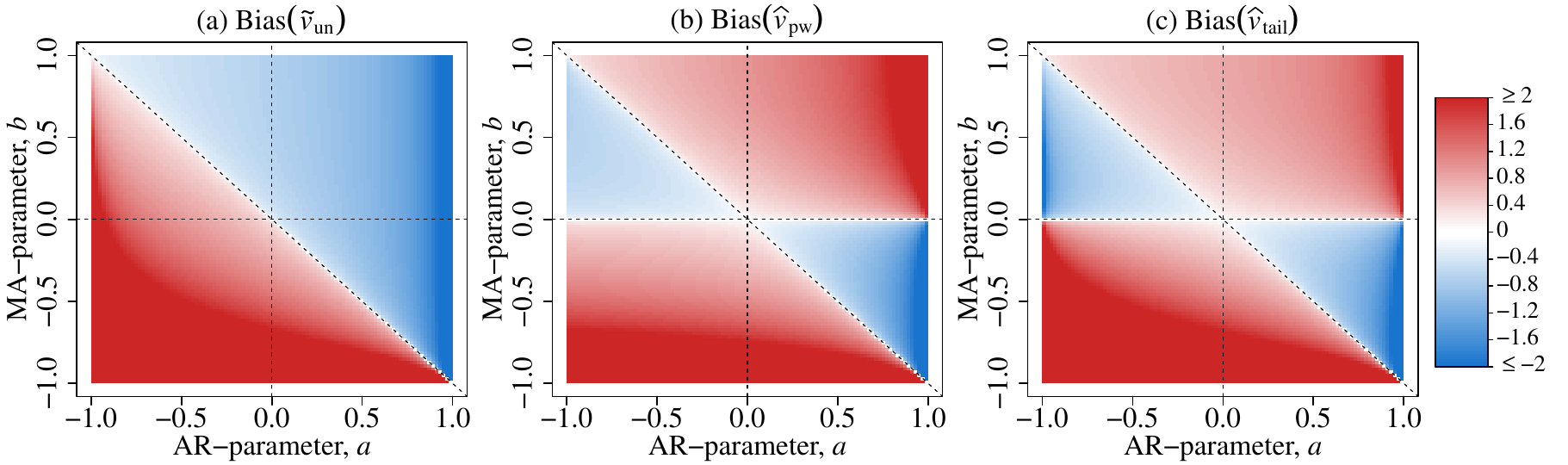}
    \end{center}
    \vspace{-0.3cm}
        \caption{Heat maps display the limiting values of $\Bias(\cdot)n^{1/3}/v$ for
        $\tilde{v}$, $\hat{v}_{\textsc{pw}}$, and $\hat{v}_{\textsc{tail}}$
        in plots (a), (b), and (c), respectively; 
        see Example \ref{exp:ar11} for the details.  
        More negative asymptotic bias appears more bluish, while more positive asymptotic bias appears more reddish.
		Note that the titles of plots do not show the standardization $n^{1/3}/v$ due to space constraints. 
        }\label{fig:ar1pw1_biasSign}
\end{figure}
\end{example}

\clearpage 

\begin{example}[\AR$(2)$ data]\label{exp:ar2}\leavevmode
Suppose the data are generated from the \AR$(2)$ model: $X_i = a_1 X_{i-1} + a_2 X_{i-2} +\varepsilon_i$ 
where $\varepsilon_i \sim \Normal(0,1)$ independently and $a_1,a_2\in\mathbb{R}^2$ such that $|a_2|<1$, $a_2+a_1<1$ and $a_2-a_1<1$.
Note that 
\begin{align}\label{eqt:acvf_ar2}
	\gamma_0^X = \frac{1-a_2}{A}, \qquad 
	\gamma_1^X = \frac{a_1}{A}, \qquad 
	\gamma_k^X = a_1 \gamma^X_{k-1}+ a_2 \gamma_{k-2}\quad (k\geq 2),
\end{align}
where $A=(1+a_2)\{(1-a_2)^2-a_1^2\}$. 
The analytical form of $\gamma_k^X$ for $k\geq 2$ can be solved easily by solving the difference equation. 

Similarly, we use the \AR$(1)$ model $X_i = \phi X_{i-1} + Z_i$ 
to perform the prewhitening as well as tail postcoloring procedures. 
To find the best \AR$(1)$ projection, we set $\phi = \phi_{\star}$, where $\phi_{\star} \equiv \gamma^X_1/\gamma^X_0 = a_1/(1-a_2)$. 
As in (\ref{eqt:vX_vZ}), we also have $v^Z = {v^X}{(1-\phi_{\star})^2}$, 
which can be found by using (\ref{eqt:acvf_ar2}) and the property that 
$\gamma_k^Z = (1+\phi_{\star}^2)\gamma_k^X - \phi_{\star}^2(\gamma_{k+1}^X+\gamma_{k-1}^X)$.
In addition, we have $\xi_1 = \kappa_1^X - \kappa_1(\phi_{\star})$, where  
$\kappa_1^X = {v_1^X}/{v_0^X} $ can be found by using (\ref{eqt:acvf_ar2}), and 
\[
	\kappa_1(\phi_{\star}) = \frac{v_1(\phi_{\star})}{v_0(\phi_{\star})} = \frac{2\phi_{\star}}{1-\phi_{\star}^2} = \frac{2a_1(1-a_2)}{(1-a_2)^2-a_1^2}.  
\] 

We consider three estimators $\tilde{v}_{\textsc{un}}$, $\hat{v}_{\textsc{pw}}$, and $\hat{v}_{\textsc{tail}}$
with their respective optimal bandwidths.
We then compare their MSEs, squared biases, and the variances in 
Figure \ref{fig:ar1pw1_ar2data}.
The results are similar to those in Example \ref{exp:ar11}. 
In particular, 
the proposed $\hat{v}_{\textsc{tail}}$ improves upon $\hat{v}_{\textsc{pw}}$
when $a_1>0$, which covers cases where time series are positively related to the lag-one observation.  

In addition, as in Example \ref{eg:biasComparison}, 
we compute the limiting values of $n^{1/3}\Bias(\cdot)/v$ for the three estimators.  
The results are displayed in Figure \ref{fig:ar1pw1_ar2data_biasSign}. 
We observed that 
$\tilde{v}$, $\hat{v}_{\textsc{pw}}$ and $\hat{v}_{\textsc{tail}}$ always have the same sign of asymptotic bias
under this data generating model.
This phenomenon is the same as that in Example \ref{eg:biasComparison}. 
The asymptotic biases of $\tilde{v}$, $\hat{v}_{\textsc{pw}}$ and $\hat{v}_{\textsc{tail}}$ are positive when $a_2<0$. 
Both prewhitening and tail postcoloring successfully turn the negative asymptotic bias of $\tilde{v}_{\textsc{un}}$
in the region $(a_1, a_2)\in(0,2)\times (-1,0)$ to positive values. 

\begin{figure}[t]
    \begin{center}
          \includegraphics[width=\linewidth]{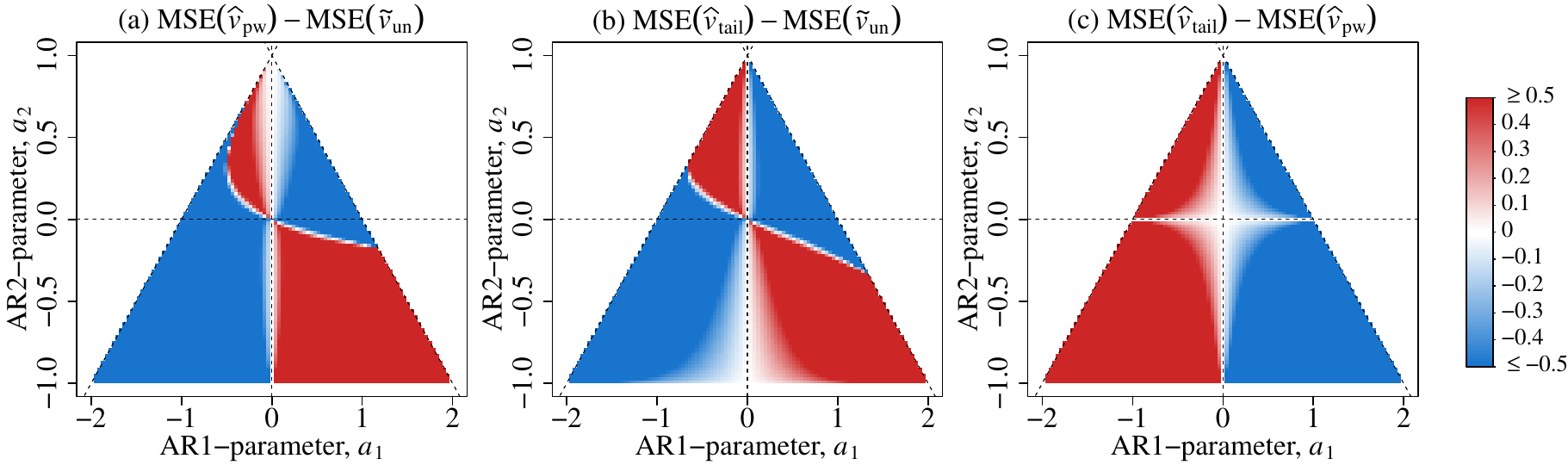}
    \end{center}
    \begin{center}
          \includegraphics[width=\linewidth]{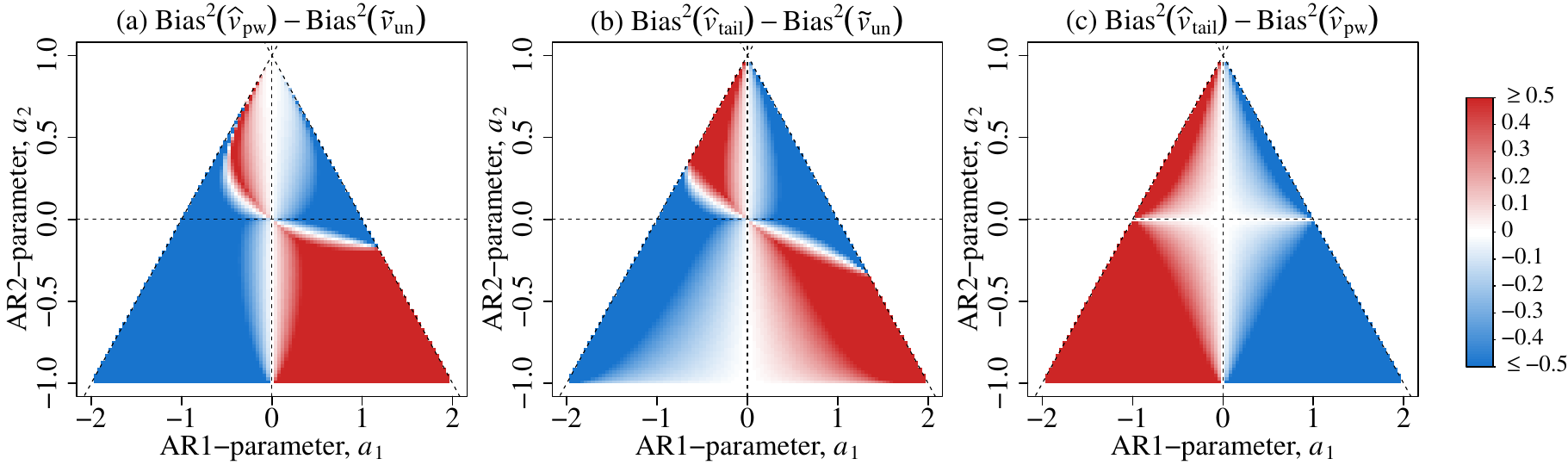}
    \end{center}
    \begin{center}
          \includegraphics[width=\linewidth]{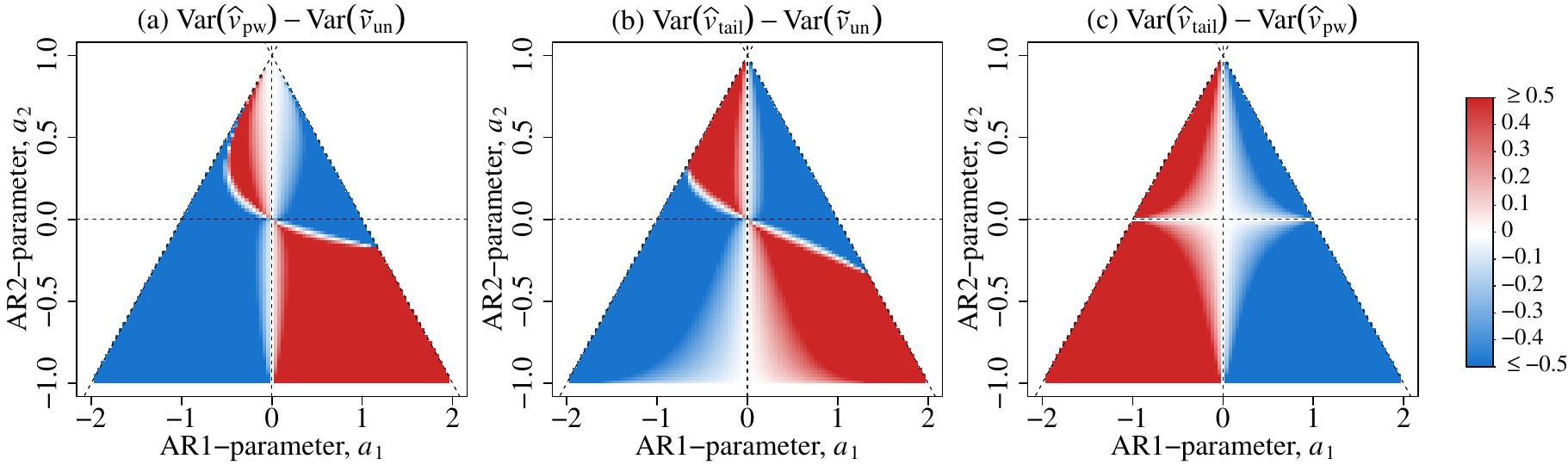}
    \end{center}
    \vspace{-0.3cm}
        \caption{\label{fig:ar1pw1_ar2data}
        Heat maps display the asymptotic differences of the MSEs, squared bias, and variances 
        after standardizing by the true value of $v^2$, 
        i.e., 
        $\{ \MSE(\cdot)-\MSE(\cdot) \}n^{2/3}/v^2$,
        $\{ \Bias^2(\cdot)-\Bias^2(\cdot) \}n^{2/3}/v^2$, and 
        $\{ \Var(\cdot)-\Var(\cdot) \}n^{2/3}/v^2$, 
        in the first row, second row, and third rows of plots, respectively. 
        The first column, second column, and third column compare 
        $(\hat{v}_{\textsc{pw}},\tilde{v}_{\textsc{un}})$, 
        $(\hat{v}_{\textsc{tail}},\tilde{v}_{\textsc{un}})$, and  
        $(\hat{v}_{\textsc{tail}},\hat{v}_{\textsc{pw}})$, respectively. 
        See Example \ref{exp:ar2} for a detailed description. 
        A difference less than zero (i.e., blueish regions) indicates an improvement
        of the first estimator over the second estimator in the difference. 
        Note also that the titles of plots do not show the standardization $n^{2/3}/v^2$ due to space constraints. 
        }
\end{figure}

\begin{figure}[t]
    \begin{center}
          \includegraphics[width=\linewidth]{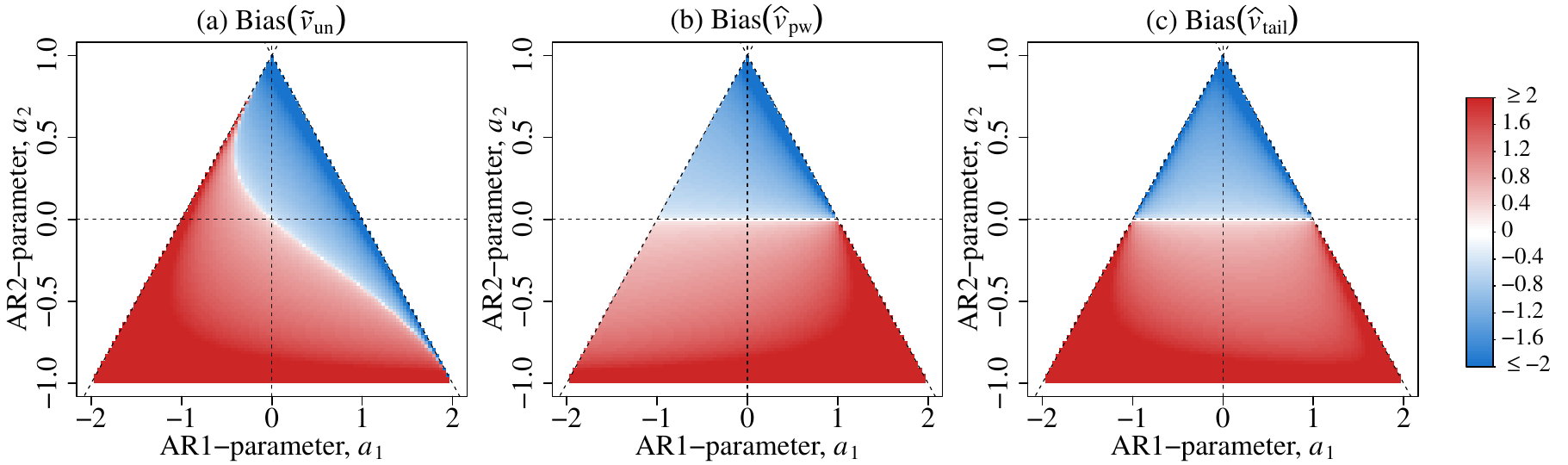}
    \end{center}
    \vspace{-0.3cm}
        \caption{Heat maps display the limiting values of $\Bias(\cdot)n^{1/3}/v$ for
        $\tilde{v}$, $\hat{v}_{\textsc{pw}}$, and $\hat{v}_{\textsc{tail}}$
        in plots (a), (b), and (c), respectively; 
        see Example \ref{exp:ar2} for the details.  
        More negative asymptotic bias appears more bluish, while more positive asymptotic bias appears more reddish.
		Note that the titles of plots do not show the standardization $n^{1/3}/v$ due to space constraints. 
        }\label{fig:ar1pw1_ar2data_biasSign}
\end{figure}

\end{example}

\clearpage 

\begin{example}[\MA$(2)$ data]\label{exp:ma2}\leavevmode
Suppose the data are generated from the \MA$(2)$ model: $X_i = b_1 \varepsilon_{i-1} + b_2 \varepsilon_{i-2} +\varepsilon_i$ 
where $\varepsilon_i \sim \Normal(0,1)$ independently and $b_1,b_2\in[-1,1]^2$.
Note that 
\begin{align}\label{eqt:acvf_ma2}
	\gamma_0^X = b_1^2 + b_2^2 + 1, \qquad 
	\gamma_1^X = b_1b_2 + b_1, \qquad 
	\gamma_2^X = b_2, \qquad
	\gamma_k^X = 0 \quad (k\geq 3).
\end{align}

Similarly, we use the \AR$(1)$ model $X_i = \phi X_{i-1} + Z_i$ 
to perform the prewhitening as well as tail postcoloring procedures. 
To find the best \AR$(1)$ projection, we set 
$\phi = \phi_{\star}$, where $\phi_{\star} \equiv \gamma^X_1/\gamma^X_0 = (b_1b_2+b_1)/(b_1^2+b_2^2+1)$. 
As in (\ref{eqt:vX_vZ}), we also have $v^Z = {v^X}{(1-\phi_{\star})^2}$, 
where $v^X = (1+b_1+b_2)^2$. 
The analytic form of $v_1^Z$ can be easily found in view of (\ref{eqt:acvf_ma2}) and 
the property that $\gamma_k^Z = (1+\phi_{\star}^2)\gamma_k^X - \phi_{\star}^2(\gamma_{k+1}^X+\gamma_{k-1}^X)$.  
Also, we have $\xi_1 = \kappa_1^X - \kappa_1(\phi_{\star})$, where  
\[
	\kappa_1^X \equiv \frac{v_1^X}{v_0^X} = \frac{2(b_1+2b_2+b_1b_2)}{(b_1+b_2+1)^2} 
	\quad \text{and} \quad 
	\kappa_1(\phi_{\star}) = \frac{v_1(\phi_{\star})}{v_0(\phi_{\star})} = \frac{2\phi_{\star}}{1-\phi_{\star}^2} 
	= \frac{2(b_1b_2+b_1)(b_1^2+b_2^2+1)}{(b_1^2+b_2^2+1)^2-(b_1b_2+b_1)^2}.  
\] 

We consider three estimators $\tilde{v}_{\textsc{un}}$, $\hat{v}_{\textsc{pw}}$, and $\hat{v}_{\textsc{tail}}$
with their respective optimal bandwidths.
We then compare their MSEs, squared biases, and the variances in 
Figure \ref{fig:ar1pw1_ma2data}.
We observe that $\hat{v}_{\textsc{tail}}$ improves upon $\hat{v}_{\textsc{pw}}$ for nearly all cases where $b_1>0$, 
corresponding to time series that are positively related to the lag-one innovation. 
However, the improvement regions for $\hat{v}_{\textsc{pw}}$ and $\hat{v}_{\textsc{tail}}$ over $\tilde{v}_{\textsc{un}}$
are non-trivial. 
In particular, when $b_2=0$ and $b_1\gg 0$, the data will exhibit a significant lag-one autocorrelation, 
which incorrect guides the $\AR(1)$ whitening/coloring models. 
Thus it leads to an inflation of MSE. 
However, such inflation is negligible when $n\rightarrow\infty$. 

In Figure \ref{fig:ar1pw1_ma2data_biasSign}, we show the limiting values of their biases. 
Both $\hat{v}_{\textsc{pw}}$ and $\hat{v}_{\textsc{tail}}$ turn the negative bias of $\tilde{v}_{\textsc{un}}$
to positive values for the majority of the regions where $a_1,a_2>0$. 
Although $\hat{v}_{\textsc{pw}}$ and $\hat{v}_{\textsc{tail}}$ behaves similarly in terms of bias correcting effect, 
their effects are not equivalent. 
Figure \ref{fig:heatMA2_BiasDiff} shows that the sign of the asymptotic biases of 
$\hat{v}_{\textsc{pw}}$ and $\hat{v}_{\textsc{tail}}$ can be different. 
This phenomenon is different from the $\ARMA(1,1)$ data in Example \ref{exp:ar11} and 
the $\AR(2)$ data in Example \ref{exp:ar2}.

\begin{figure}[t]
    \begin{center}
          \includegraphics[width=\linewidth]{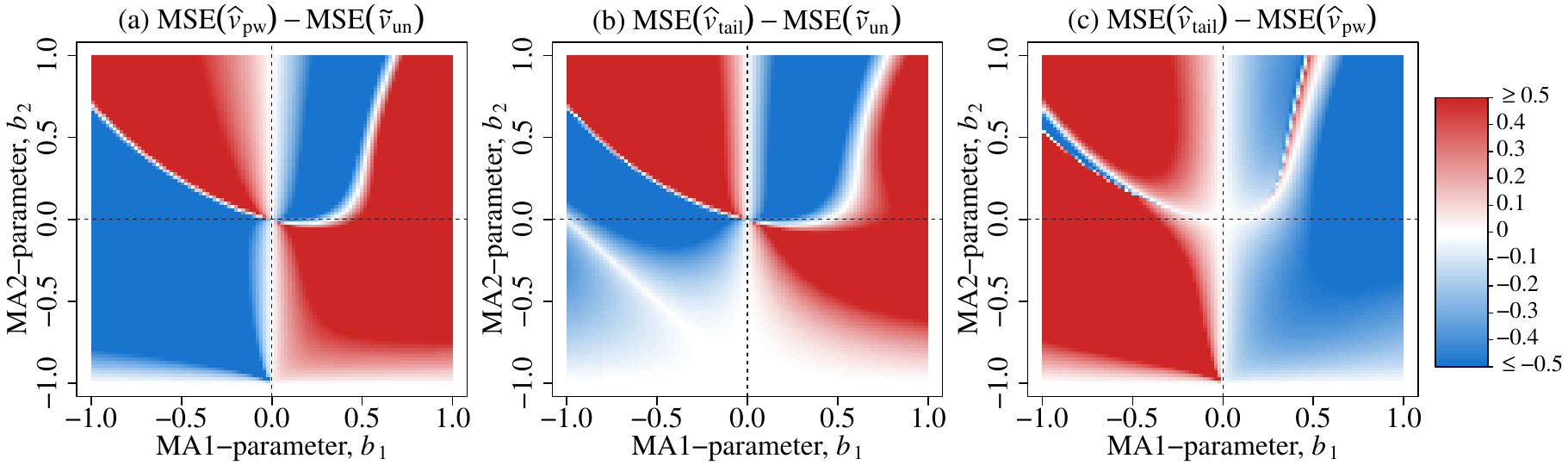}
    \end{center}
    \begin{center}
          \includegraphics[width=\linewidth]{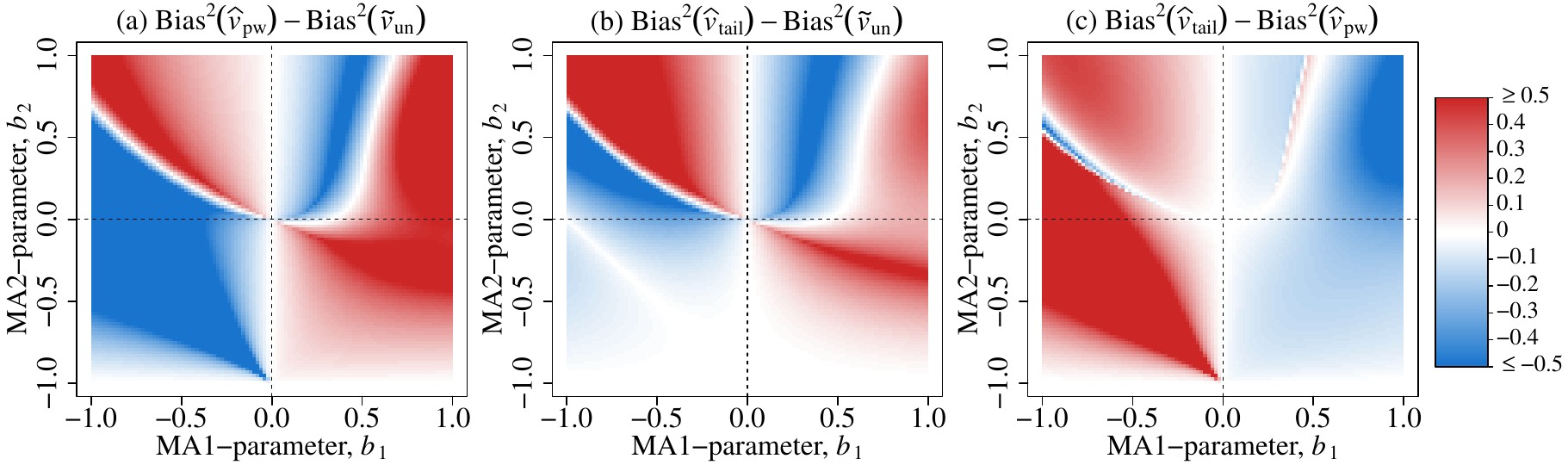}
    \end{center}
    \begin{center}
          \includegraphics[width=\linewidth]{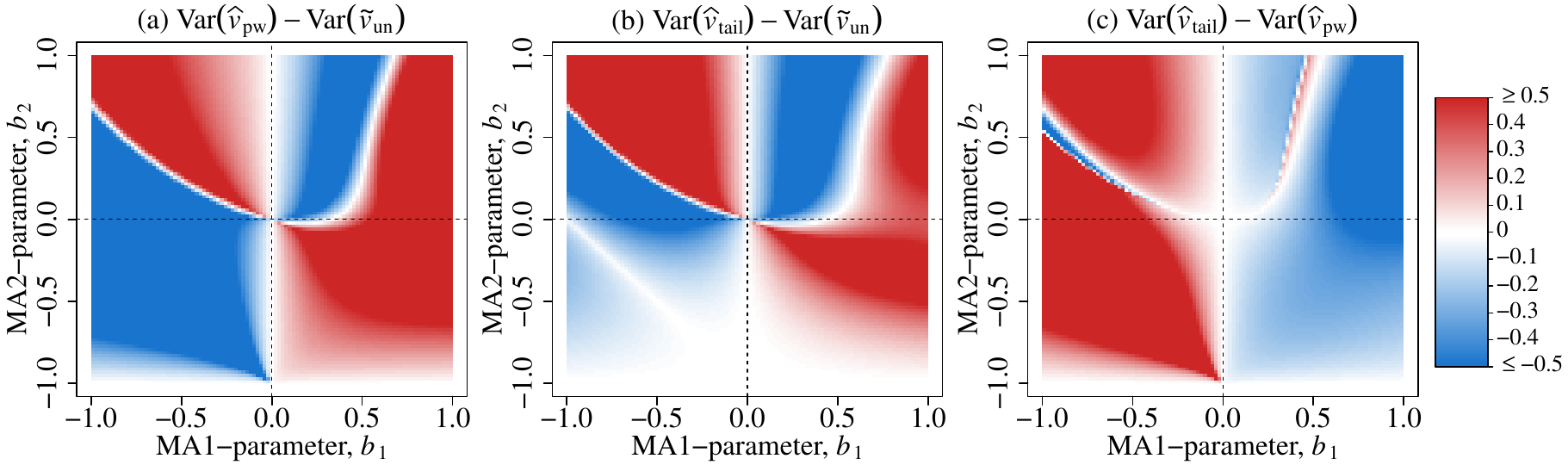}
    \end{center}
    \vspace{-0.3cm}
        \caption{\label{fig:ar1pw1_ma2data}
        Heat maps display the asymptotic differences of the MSEs, squared bias, and variances 
        after standardizing by the true value of $v^2$, 
        i.e., 
        $\{ \MSE(\cdot)-\MSE(\cdot) \}n^{2/3}/v^2$,
        $\{ \Bias^2(\cdot)-\Bias^2(\cdot) \}n^{2/3}/v^2$, and 
        $\{ \Var(\cdot)-\Var(\cdot) \}n^{2/3}/v^2$, 
        in the first row, second row, and third rows of plots, respectively.  
        The first column, second column, and third column compare 
        $(\hat{v}_{\textsc{pw}},\tilde{v}_{\textsc{un}})$, 
        $(\hat{v}_{\textsc{tail}},\tilde{v}_{\textsc{un}})$, and  
        $(\hat{v}_{\textsc{tail}},\hat{v}_{\textsc{pw}})$, respectively. 
        See Example \ref{exp:ma2} for a detailed description. 
        A difference less than zero (i.e., blueish regions) indicates an improvement
        of the first estimator over the second estimator in the difference. 
        Note that the titles of plots do not show the standardization $n^{2/3}/v^2$ due to space constraints. 
        }
\end{figure}

\begin{figure}[t]
    \begin{center}
          \includegraphics[width=\linewidth]{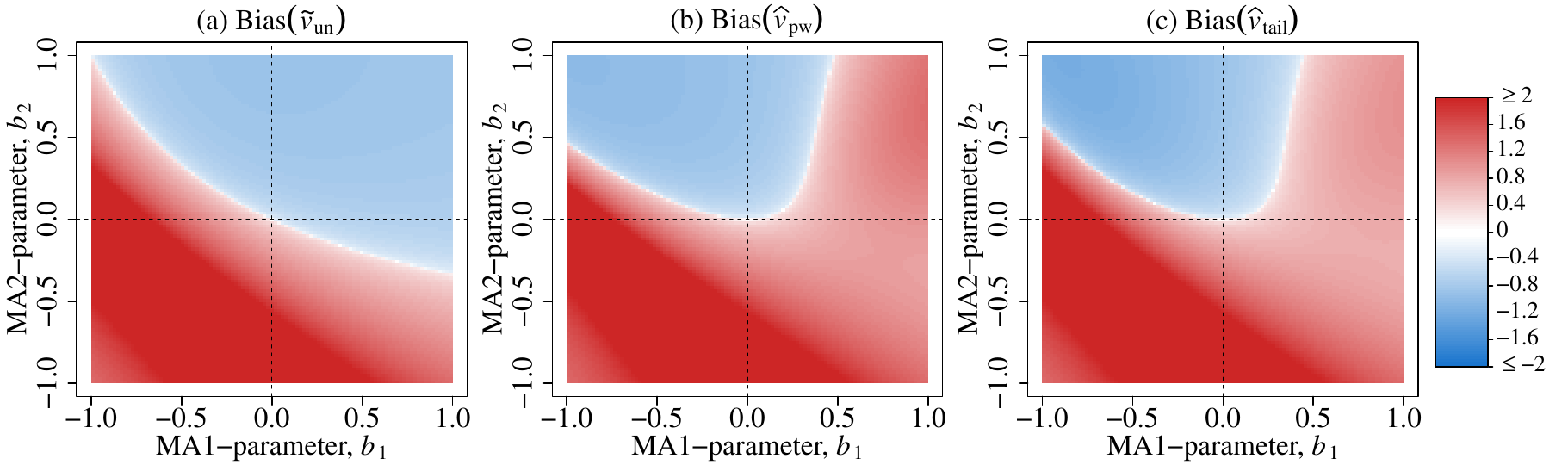}
    \end{center}
    \vspace{-0.3cm}
        \caption{Heat maps display the limiting values of $\Bias(\cdot)n^{1/3}/v$ for
        $\tilde{v}$, $\hat{v}_{\textsc{pw}}$, and $\hat{v}_{\textsc{tail}}$
        in plots (a), (b), and (c), respectively; 
        see Example \ref{exp:ma2} for the details.  
        More negative asymptotic bias appears more bluish, while more positive asymptotic bias appears more reddish.
        Note that the titles of plots do not show the standardization $n^{1/3}/v$ due to space constraints. 
        }\label{fig:ar1pw1_ma2data_biasSign}
\end{figure}

\begin{figure}[t]
    \begin{center}
          \includegraphics[width=.7\linewidth]{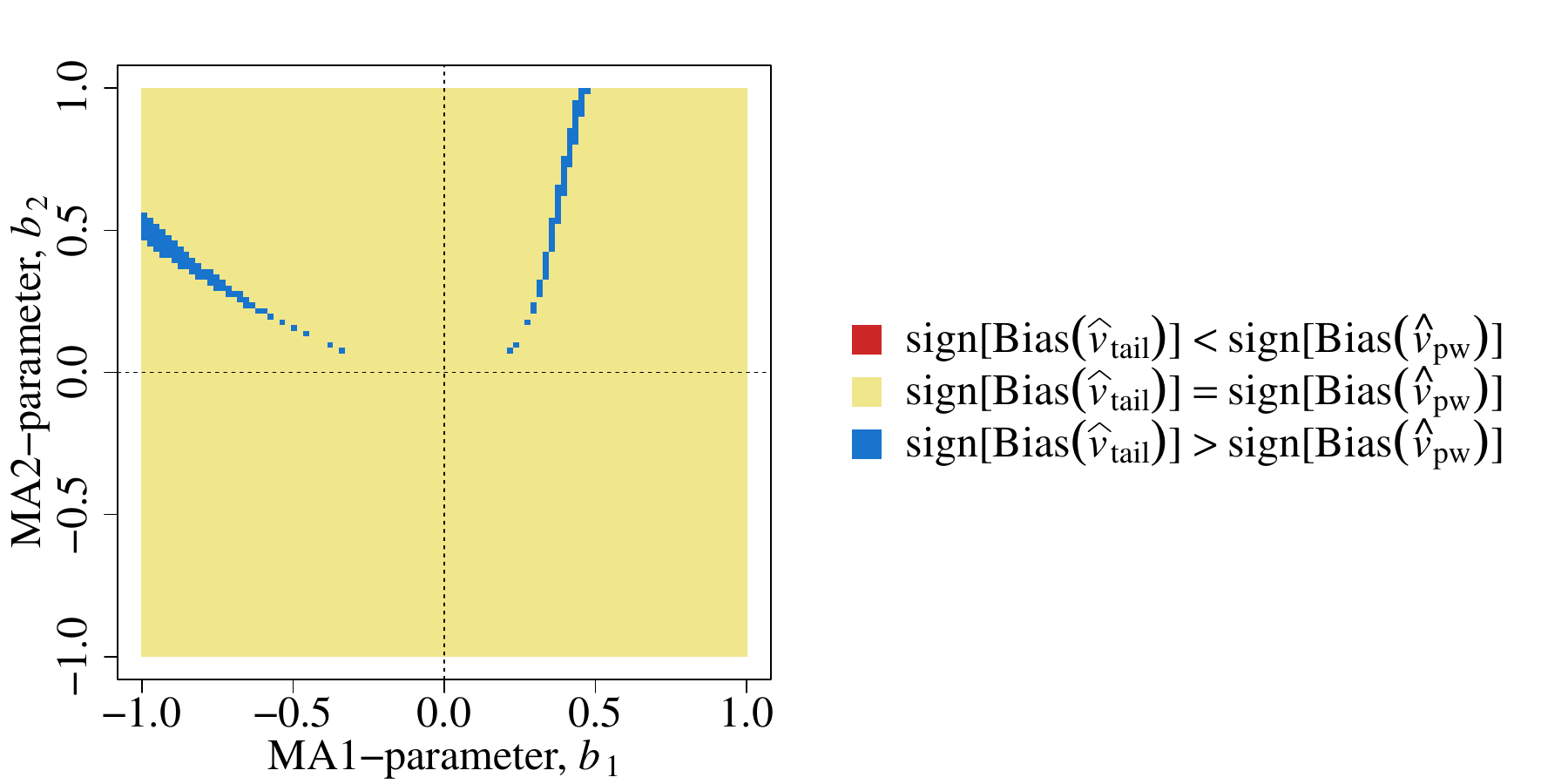}
    \end{center}
    \vspace{-0.3cm}
        \caption{The difference of the signs of the limiting values of $\Bias(\cdot)n^{1/3}/v$ for
        $\hat{v}_{\textsc{pw}}$ and $\hat{v}_{\textsc{tail}}$.
        The yellow region corresponds to the cases where 
        the signs of
        $\lim_{n\rightarrow\infty}n^{1/3}\Bias(\hat{v}_{\textsc{pw}})$ and 
        $\lim_{n\rightarrow\infty}n^{1/3}\Bias(\hat{v}_{\textsc{tail}})$ are the same.  
        The blue regions corresponds to the cases where 
        $\lim_{n\rightarrow\infty}n^{1/3}\Bias(\hat{v}_{\textsc{pw}}) < 0 <\lim_{n\rightarrow\infty}n^{1/3}\Bias(\hat{v}_{\textsc{tail}})$. 
        Note that there exists no $(b_1,b_2)$ such that $\lim_{n\rightarrow\infty}n^{1/3}\Bias(\hat{v}_{\textsc{pw}}) > 0 >\lim_{n\rightarrow\infty}n^{1/3}\Bias(\hat{v}_{\textsc{tail}})$. 
        }\label{fig:heatMA2_BiasDiff}
\end{figure}

\end{example}

\clearpage

\subsection{Additive form of tail postcoloring}\label{sec:additiveTailprecolor}
Our proposed tail postcolored estimator $\hat{v}_{\textsc{tail}}$ in Definition \ref{def:SimpleTPW} 
admits a multiplicative form: $\hat{v}_{\textsc{tail}} = {\eta}_{\ell, K}(\bar{\theta}) \tilde{v}(\ell; K)$. 
It is possible to perform tail postcoloring additively as:
\[
    \hat{v}_{\textsc{tail}}' = \tilde{v}(\ell; K) + \left[\sum_{ |k| < n} \left\{1-K(k/\ell)\right\} {\gamma}_k(\bar{\theta})
    + \sum_{|k| \geq n} {\gamma}_k(\bar{\theta})\right].  
\]
However, we recommend the multiplicative form $\hat{v}_{\textsc{tail}}$ over this 
additive form $\hat{v}_{\textsc{tail}}'$ because of three reasons: 

First, $\hat{v}_{\textsc{tail}}'$ is not guaranteed to be positive semi-definite in finite samples 
even if the kernel $K$ is positive semi-definite.
Hence, an extra layer of positive-definiteness correction may be needed. 
This is the major drawback of $\hat{v}_{\textsc{tail}}'$. 
Second,  
the multiplicative form $\hat{v}_{\textsc{tail}}$ resembles the classical prewhitened estimator $\hat{v}_{\textsc{pw}}$, 
as both estimators can be expressed as $C \tilde{v}$ for some $C$. 
This makes it more memorable.
Third, the expression of $\hat{v}_{\textsc{tail}}'$ is more complicated than our proposed $\hat{v}_{\textsc{tail}}$.
Arguably, this formula is more cumbersome to state and to compute, so a neater form is preferable.

}

{
\subsection{Spectral density estimation with tail postcoloring}\label{sec:specDen}
We generalize the proposed tail postcoloring method to spectral density estimation in this section.
Assume that $u_0 = \sum_{k=-\infty}^{\infty} |\gamma_{k}| < \infty$.
The spectral density is
\begin{equation}\label{eqt:spectral}
  f(\omega) = \frac{1}{2\pi} \sum_{k=-\infty}^{\infty}\gamma_k e^{\iota k\omega} = \frac{1}{2\pi} \sum_{k=-\infty}^{\infty}\gamma_k \cos(k\omega),
\end{equation}
where $\omega \in \mathbb{R}$ and $\iota = \sqrt{-1}$ denotes the imaginary unit.
Our proposed tail postcoloring estimator is  
\begin{align}\label{eqt:SpecDen}
  \hat{f}_{\textsc{tail}}(\omega) &= \hat{f}(\omega; \ell; K; \bar{\theta})  
  = \frac{\sum_{k=-\infty}^{\infty}\gamma_k(\bar{\theta})e^{\iota k\omega}}{\sum_{k=-\infty}^{\infty}K\left({k}/{\ell}\right)\gamma_k(\bar{\theta})e^{\iota k\omega}} \cdot \tilde{f}(\omega),
\end{align}
where 
\[
  \tilde{f}(\omega) = \frac{1}{2\pi} \sum_{k=-\infty}^{\infty}K\left(\frac{k}{\ell}\right)\tilde{\gamma}_k e^{\iota k\omega}
\]
is the standard kernel spectral density estimator. 
We remark that both $f$ and $\hat{f}_{\textsc{tail}}$ are even and periodic functions with a period of $2\pi$. 
It suffices for use to study the properties of $\hat{f}_{\textsc{tail}}(\omega)$ when $\omega\in[0, \pi]$. 

For $p\in\mathbb{N}_0$, we also define 
\[
	f_p(\omega; \theta_{\star}) = \frac{1}{2\pi} \sum_{k=-\infty}^{\infty}\gamma_k(\theta_{\star}) e^{\iota k\omega} 
    \qquad\text{and}\qquad
    f_p(\omega) = \frac{1}{2\pi} \sum_{k=-\infty}^{\infty}\gamma_ke^{\iota k\omega} ,
\]
which are based on the tail ppstwhitening model and the true data generating process, respectively. 
Define 
the truncated mean-squared error of any estimator $\hat{f}(\omega)$ of $f(\omega)$ as 
\[
	\MSE_{h}\{\hat{f}(\omega)\} = \E\left[\min\left\{n^{2p/(2p+1)} (\hat f(\omega)-f(\omega))^2,h\right\}\right].
\]
We then have the following theorem describing the asymptotic properties of the estimator in (\ref{eqt:SpecDen}).

\begin{theorem}\label{thm:spectralBV}
Assume the conditions in Theorem~\ref{thm:BiasVar} hold.
Let $\ell \asymp n^{1/(2p+1)}$.  
For any $\omega\in[0, \pi]$, we have 
(i)
\[
  \lim_{h\to\infty} \lim_{n\to\infty}  \MSE_h \{\hat{f}(\omega; \ell; K; \bar{\theta})\}
  = \lim_{n\to\infty} n^{{2p}/({2p+1})} \MSE\{\hat{f}(\omega; \ell; K; \theta_{\star} )\};
\]
{
(ii)  
$
	\Var\left\{ \hat{f}(\omega; \ell; K; \theta_{\star} ) \right\} 
	\sim \Var\left\{ \tilde{f}(\omega) \right\} 
	\sim 4 A \zeta(\omega)  f^2(\omega) \ell/n,
$
where $\zeta(\omega) = \{1+\mathbb{1}(\omega/\pi\in\mathbb{Z})\}/2$;
and (iii)
    \begin{equation}\label{eqt:NonnegligibleBias}
      \bias\left\{ \hat{f}(\omega; \ell; K; \theta_{\star} ) \right\}
      =  \frac{B \chi_p(\omega)  }{\ell^p} f(\omega) + o\left(\frac{1}{\ell^p}\right) + O\left(\frac{\ell}{n}\right), 
    \end{equation}
    where 
    \[
    \chi_p(\omega) = 
      \frac{f_p(\omega)}{f_0(\omega)} 
      - \frac{f_p(\omega; \theta_{\star})}{f_0(\omega; \theta_{\star})}. 
    \]}
  \end{theorem}

}

\section{Additional simulation results}\label{sec:supp_Simulation}

\subsection{Simulation results of Example \ref{exp:simu_ar1} when  $b=0.6$}\label{sec:sim_MSE_estBW}

The simulation experiment in Example \ref{exp:simu_ar1} is repeated
when $b=0.6$. 
{
Recall that the estimators considered are as follows:
  \begin{itemize}
  \item[(a)] an $\AR(1)$-based parametric estimator;
  \item[(b)] the unadjusted estimator with the nonparametric plug-in bandwidth selector;
  \item[(c)] \cite{andrews1992}'s $\AR(1)$-prewhitened estimator with their parametric $\AR(1)$ bandwidth selector; and
  \item[(d)] the proposed $\AR(1)$-tail postcolored estimator.
  \end{itemize}
  They are computed as  
  \begin{align}\label{eqt:estimator_list_sim}
    \bar{v}_{\textsc{para}} = \frac{\bar{\sigma}^2_{z}}{(1-\bar{\phi})^2}, \quad
    \tilde{v}_{\textsc{un}} = \tilde{v}(\hat{\ell}_{\textsc{un}}; K; X_{1:n}), \quad 
    \hat{v}_{\textsc{pw}} = \frac{\tilde{v}(\hat{\ell}_{\textsc{pw}};K; Z_{2:n})}{(1-\bar{\phi})^2}, \quad
    \hat{v}_{\textsc{tail}} = \hat{v}(\hat{\ell}_{\textsc{tail}}; K; \bar{\phi}),
  \end{align}
  respectively, 
  where the function $\tilde{v}(\cdot)$ is defined in (\ref{eqt:kernelEst}), 
  the function $\hat{v}(\cdot)$ is defined in (\ref{eqt:propDebias}), and 
  \begin{align*} 
   \hat{\ell}_{\textsc{un}} = \left\lceil \left( {3\tilde{\kappa}_1^2 n }/{ 2} \right)^{1/3} \right\rceil, \quad
    \hat{\ell}_{\textsc{pw}} = \left\lceil{\left\{ {6{\bar{\phi}_{Z}}^2n}/{(1-{\bar{\phi}_{Z}}^2)^2}\right\}^{1/3}}\right\rceil, \quad
    \hat{\ell}_{\textsc{tail}} = \left\lceil \left( {3\hat{\xi}_1^2 n }/{ 2} \right)^{1/3} \right\rceil.
  \end{align*}
  where 
  \[
  \tilde{\kappa}_1
  = \tilde{v}_{1,\#} /\tilde{v}_{0,\#}
  \qquad\text{and}\qquad
  \hat{\xi}_1 
      = \tilde{v}_{1,\#} /\tilde{v}_{0,\#}  - 2\bar{\phi}(1-\bar{\phi}^2)^{-1},
  \]
  with $\tilde{v}_{r,\#} \equiv \sum_{|k|< n} K_r(k/\ell_r^{\#}) {|k|}^r \tilde{\gamma}_k$
  and $\ell_r^{\#} = \left\lceil 2n^{1/(2r+3)} \right\rceil
  $ for $r\in\{0,1\}$.
}
In other words,
each estimator is computed with its estimated optimal bandwidth.
The results are shown in Table~\ref{tab:exp1_estimatedBW_2}.

\begin{table}[t]
\centering
\small
\def~{\hphantom{0}}
\begin{tabular}{cccccccccc}
  &   & \multicolumn{4}{c}{$b=0.6$ (Misspecified case)} \\
  Criteria & $a\backslash$Estimator &  $\bar{v}_{\textsc{para}}$ & $\tilde{v}_{\textsc{un}}$ & $\hat{v}_{\textsc{pw}}$ & $\hat{v}_{\textsc{tail}}$ \\[0.5ex]
   $100\MSE(\cdot)/v^2$ 
    & $  -0.8$ & $~3.86$ & $~2.19$ & $~~2.69$ & $~2.02$ \\  
    & $  -0.4$ & $~3.83$ & $~1.73$ & $~~4.34$ & $~2.29$ \\  
    & $  -0.2$ & $~9.26$ & $~2.65$ & $~~8.89$ & $~3.23$ \\  
    & $  ~0.2$ & $27.18$ & $~4.75$ & $~17.90$ & $~5.58$ \\  
    & $  ~0.4$ & $38.80$ & $~6.08$ & $~29.97$ & $~7.40$ \\  
    & $  ~0.8$ & $74.33$ & $14.96$ & $163.42$ & $21.34$ \\  [0.5ex]
    $10\{ \E(\cdot)/v - 1 \}$ 
    & $  -0.8$  & $-1.71$ & ~$0.18$ & $-1.32$ & $-0.78$ \\  
    & $  -0.4$  & $~1.32$ & -$0.41$ & $~1.41$ & $~0.69$ \\  
    & $  -0.2$  & $~2.52$ & -$0.58$ & $~2.46$ & $~0.81$ \\  
    & $  ~0.2$  & $~4.61$ & -$0.98$ & $~3.43$ & $~1.03$ \\  
    & $  ~0.4$  & $~5.51$ & -$1.18$ & $~4.52$ & $~1.18$ \\  
    & $  ~0.8$  & $~6.71$ & -$2.18$ & $10.29$ & $~1.64$ \\ [0.5ex]
\end{tabular}
\caption{The standardized mean-squared error $100\MSE(\cdot)/v^2$
and standardized bias $10\{ \E(\cdot)/v - 1 \}$
of 
the parametric estimator $\bar{v}_{\textsc{para}}$,
the unprewhitened estimator $\tilde{v}_{\textsc{un}}$,
the standard prewhitened estimator $\hat{v}_{\textsc{pw}}$, and 
the proposed tail prewhitened estimator $\hat{v}_{\textsc{tail}}$
but with optimal bandwidth estimated by the parametric plug-in method; see Example \ref{exp:simu_ar1}.
The sample size is $n=400$.
Each experiment is replicated $5000$ times.
}
\label{tab:exp1_estimatedBW_2}
\end{table}

In this case, the proposed $\hat{v}_{\textsc{tail}}$ is the second best 
next to the unprewhitened estimator $\hat{v}_{\textsc{un}}$ in terms of mean-squared error. 
We note that the inflation of mean-squared error of $\hat{v}_{\textsc{tail}}$ 
is not disastrous compared to $\hat{v}_{\textsc{un}}$ in the case of $b=0.6$. 
However, the mean-squared error of $\hat{v}_{\textsc{un}}$ may explode to a very large value in other cases, e.g., $b=-0.6$.
In addition, $\hat{v}_{\textsc{tail}}$ is much more promising compared to other estimators and provides satisfactory bias correction.

{
\subsection{Simulation results of Example \ref{exp:simu_ar1} when $n=200$}\label{sec:sim_MSE_n200}
The simulation experiment in Example \ref{exp:simu_ar1} is repeated for sample size $n=200$.
The results are shown in Tables \ref{tab:exp1_n200} and \ref{tab:exp1_estimatedBW_n200}
when the true optimal bandwidth and the parametric plug-in bandwidth are used, respectively. 

\begin{table}[t]
\centering
\small
\def~{\hphantom{0}}
\begin{tabular}{cccccccccc}
  &  & \multicolumn{4}{c}{$b=0$ (Well-specified case)} & \multicolumn{4}{c}{$b=-0.6$ (Misspecified case)} \\
  Criteria & $a\backslash$Estimator &  $\bar{v}_{\textsc{para}}$ & $\tilde{v}_{\textsc{un}}$ & $\hat{v}_{\textsc{pw}}$ & $\hat{v}_{\textsc{tail}}$ &  $\bar{v}_{\textsc{para}}$ & $\tilde{v}_{\textsc{un}}$ & $\hat{v}_{\textsc{pw}}$ & $\hat{v}_{\textsc{tail}}$ \\[0.5ex]
   $100\MSE(\cdot)/v^2$ 
   & $  -0.8$ & $38.94$ & $14.46$ & $~1.26$ & $~1.30$ & $~60.59$ & $814.37$ & $149.63$ & $71.39$ \\  
   & $  -0.4$ & $~4.57$ & $~5.79$ & $~1.86$ & $~1.86$ & $416.08$ & $235.76$ & $131.67$ & $31.04$ \\  
   & $  -0.2$ & $~3.01$ & $~3.60$ & $~2.40$ & $~2.40$ & $426.39$ & $156.75$ & $154.08$ & $23.31$ \\  
   & $  ~0.2$ & $~2.97$ & $~3.94$ & $~3.94$ & $~3.98$ & $167.63$ & $~68.00$ & $183.24$ & $14.76$ \\  
   & $  ~0.4$ & $~5.27$ & $~6.98$ & $~5.65$ & $~5.76$ & $~49.35$ & $~36.62$ & $~53.24$ & $10.66$ \\  
   & $  ~0.8$ & $44.53$ & $19.10$ & $16.78$ & $18.22$ & $~32.28$ & $~26.61$ & $~31.43$ & $14.82$ \\[0.5ex] 
    $10\{ \E(\cdot)/v - 1 \}$ 
     & $-0.8$ & $-6.16$ & $~1.69$ & $~0.05$ & $-0.00$ & $~6.47$ & $27.08$ & $11.86$ & $~4.62$ \\  
     & $-0.4$ & $-1.47$ & $~1.06$ & $~0.02$ & $-0.01$ & $19.57$ & $14.82$ & $11.09$ & $~2.92$ \\  
     & $-0.2$ & $-0.33$ & $~0.86$ & $-0.01$ & $-0.03$ & $19.92$ & $12.03$ & $12.03$ & $~2.39$ \\  
     & $~0.2$ & $-0.47$ & $-1.26$ & $-0.11$ & $-0.08$ & $12.37$ & $~7.77$ & $13.11$ & $~1.89$ \\  
     & $~0.4$ & $-1.72$ & $-1.72$ & $-0.18$ & $-0.11$ & $~6.45$ & $~5.59$ & $~6.83$ & $~1.54$ \\  
     & $~0.8$ & $-6.62$ & $-3.31$ & $-0.80$ & $-0.41$ & $-5.60$ & $-5.02$ & $-5.50$ & $-2.65$ \\[0.5ex] 
\end{tabular}
\caption{
{The standardized mean-squared error $100\MSE(\cdot)/v^2$
and standardized bias $10\{ \E(\cdot)/v - 1 \}$
of 
the parametric estimator $\bar{v}_{\textsc{para}}$,
the unadjusted estimator $\tilde{v}_{\textsc{un}}$,
the standard prewhitened estimator $\hat{v}_{\textsc{pw}}$, and 
the proposed tail postcolored estimator $\hat{v}_{\textsc{tail}}$
with their own theoretical optimal bandwidths; see Example \ref{exp:simu_ar1}.
The sample size is $n=200$.
Each experiment is replicated $5000$ times.}
}\label{tab:exp1_n200}
\end{table}

\begin{table}[t]
\centering
\small
\def~{\hphantom{0}}
\begin{tabular}{cccccccccc}
  &  & \multicolumn{4}{c}{$b=0$ (Well-specified case)} & \multicolumn{4}{c}{$b=-0.6$ (Misspecified case)} \\
  Criteria & $a\backslash$Estimator &  $\bar{v}_{\textsc{para}}$ & $\tilde{v}_{\textsc{un}}$ & $\hat{v}_{\textsc{pw}}$ & $\hat{v}_{\textsc{tail}}$ &  $\bar{v}_{\textsc{para}}$ & $\tilde{v}_{\textsc{un}}$ & $\hat{v}_{\textsc{pw}}$ & $\hat{v}_{\textsc{tail}}$ \\[0.5ex]
   $100\MSE(\cdot)/v^2$ 
  & $-0.8$ & $38.94$ & $14.37$ & $~1.85$ & $~4.90$ & $~60.59$ & $798.34$ & $127.99$ & $1235.51$ \\  
  & $-0.4$ & $~4.57$ & $~6.11$ & $~1.97$ & $~3.56$ & $416.08$ & $248.31$ & $168.62$ & $~103.82$ \\  
  & $-0.2$ & $~3.01$ & $~4.11$ & $~2.43$ & $~3.59$ & $426.39$ & $169.48$ & $195.73$ & $~~54.80$ \\  
  & $~0.2$ & $~2.97$ & $~4.35$ & $~4.01$ & $~4.28$ & $167.63$ & $~81.81$ & $169.69$ & $~~27.51$ \\  
  & $~0.4$ & $~5.27$ & $~7.45$ & $~5.80$ & $~6.24$ & $~49.35$ & $~47.59$ & $~51.20$ & $~~18.10$ \\  
  & $~0.8$ & $44.53$ & $20.31$ & $17.58$ & $18.79$ & $~32.28$ & $~27.13$ & $~30.89$ & $~~28.37$ \\[0.5ex]
    $10\{ \E(\cdot)/v - 1 \}$
  & $-0.8$ & $-6.16$ & $~1.73$ & $-0.09$ & $~0.01$ & $~6.47$ & $27.16$ & $10.77$ & $32.49$ \\  
  & $-0.4$ & $-1.47$ & $~1.06$ & $-0.03$ & $-0.12$ & $19.57$ & $15.20$ & $12.36$ & $~9.13$ \\  
  & $-0.2$ & $-0.33$ & $~0.79$ & $-0.02$ & $-0.23$ & $19.92$ & $12.47$ & $13.23$ & $~6.17$ \\  
  & $~0.2$ & $-0.47$ & $-1.17$ & $-0.09$ & $-0.39$ & $12.37$ & $~8.36$ & $12.37$ & $~3.75$ \\  
  & $~0.4$ & $-1.72$ & $-1.70$ & $-0.14$ & $-0.38$ & $~6.45$ & $~6.11$ & $~6.65$ & $~2.85$ \\  
  & $~0.8$ & $-6.62$ & $-3.28$ & $-0.68$ & $-1.06$ & $-5.60$ & $-4.94$ & $-5.44$ & $-5.15$ \\[0.5ex]  
\end{tabular}
\caption{
{The standardized mean-squared error $100\MSE(\cdot)/v^2$
and standardized bias $10\{ \E(\cdot)/v - 1 \}$
of 
the parametric estimator $\bar{v}_{\textsc{para}}$,
the unadjusted estimator $\tilde{v}_{\textsc{un}}$,
the standard prewhitened estimator $\hat{v}_{\textsc{pw}}$, and 
the proposed tail postcolored estimator $\hat{v}_{\textsc{tail}}$
but with  optimal bandwidth estimated by the parametric plug-in method; see Example \ref{exp:simu_ar1}.
The sample size is $n=200$.
Each experiment is replicated $5000$ times.}}
\label{tab:exp1_estimatedBW_n200}
\end{table}
}

{

\subsection{Simulation results for testing $\mu=0$}\label{sec:unitroottesting}
Suppose that we are interested in testing the mean of a time series is zero, i.e.,  
testing $H_0: \mu = 0$ against $H_1: \mu\neq 0$. 
Let the data be generated from $X_i = \mu + Z_i$ ($i=1, \ldots, n$), 
where $Z_{i} = \phi Z_{i-1} + \varepsilon_{i}$ and $\varepsilon_{i} \sim \Normal(0,\sigma^2)$ independently.
and $\sigma$ is chosen such that the long-run variance of $\bar{Z}_n$ is one, i.e., $v=1$.  
We are interested in knowing the testing performance when
the autoregressive coefficient $\phi$ is close to $1$ and the process is close to having a unit root. 
In particular, we consider $\phi\in\{0, 0.7, 0.8, 0.9, 0.95, 0.97, 0.99\}$.

We consider 
the Wald's test statistic and self-normalized test statistic: 
\begin{align}\label{eqt:testStat_meanTest}
	T_n(\hat{v}_{\textsc{lrv}}) =n\bar{X}_n^2/\hat{v}_{\textsc{lrv}}
	\qquad \text{and} \qquad
	T^{\textsc{sn}}_n = n\bar{X}_n^2/ W_n
\end{align}
where 
$\hat{v}_{\textsc{lrv}}$ is an estimator of $v$ and  
$W_n=n^{-2}\sum_{i=1}^n \left\{\sum_{j=1}^i (X_j - \bar{X}_n) \right\}^2$ is a self-normalizer \citep{lobato2001,shaoxf2010}.
Note that $2W_n = \tilde{v}(n, K_{\Bart})$ is a special case of the normalizer used in the fixed-$b$ approach in \cite{KieferVogelsang2002Econometrica,KieferVogelsang2002}, i.e., the kernel estimator $\tilde{v}(\ell, K)$ 
with $\ell=n$ and $K=K_{\Bart}$. 
The estimator $\hat{v}_{\textsc{lrv}}$ used in (\ref{eqt:testStat_meanTest}) is chosen from the following candidates: 
\begin{enumerate}
	\item[(a)] an $\AR(1)$-based parametric estimator $\bar{v}_{\textsc{para}}$;
	\item[(b)] the unadjusted estimator $\tilde{v}_{\textsc{un}}$ 
				with the parametric plug-in bandwidth selector based on an $\AR(1)$ model;
	\item[(c)] the lugsail lag window estimator $\hat{v}_{\textsc{lug}}$ in \cite{vats2022lugsail}; 
	\item[(d)] \cite{andrews1992}’s \textsc{ar}$(1)$-prewhitened estimator $\hat{v}_{\textsc{pw}}$
				with the parametric plug-in bandwidth selector in (\ref{eqt:estimator_list_sim});
	\item[(e)] the proposed \textsc{ar}$(1)$-tail postcolored estimator $\hat{v}_{\textsc{tail}}$
				with the parametric plug-in bandwidth selector in Remark \ref{rem:est_opt_bw}.  
\end{enumerate}

\begin{table}[t]
\centering
\begin{tabular}{ccccccccc} 
True mean&	$\AR$-parameter, $\phi$ & $T_n^{\textsc{para}}$ & $T_n^{\textsc{un}}$ & $T_n^{\textsc{lug}}$ & $T_n^{\textsc{pw}}$ & $T_n^{\textsc{tail}}$ & $T_n^{\textsc{sn}}$  \\[0.5ex] 
$\mu=0$	&	$0$	&	$ 4.7$	&	$ 4.8$	&	$ 4.9$	&	$ 4.7$	&	$ 4.7$	&	$ 4.9$	\\ 
(Null)	&$0.7$	&	$ 6.5$	&	$10.6$	&	$ 6.1$	&	$ 6.4$	&	$ 4.4$	&	$ 5.6$	\\ 
	&$0.8$	&	$ 8.5$	&	$13.4$	&	$ 7.8$	&	$ 8.4$	&	$ 5.9$	&	$ 6.6$	\\ 
	&$0.9$	&	$11.5$	&	$18.2$	&	$11.4$	&	$11.2$	&	$ 8.3$	&	$ 9.3$	\\ 
	&$0.95$	&	$17.5$	&	$26.7$	&	$18.8$	&	$17.3$	&	$12.7$	&	$13.0$	\\ 
	&$0.97$	&	$24.6$	&	$35.2$	&	$28.7$	&	$24.4$	&	$17.8$	&	$18.0$	\\ 
	&$0.99$	&	$45.6$	&	$54.6$	&	$52.7$	&	$44.9$	&	$35.2$	&	$34.0$	\\[0.5ex]
$\mu=0.2$	&$0$	&	$80.4$	&	$81.2$	&	$81.1$	&	$80.4$	&	$76.0$	&	$61.1$	\\ 
 (Alternative)	&$0.7$	&	$80.8$	&	$87.3$	&	$77.9$	&	$80.7$	&	$75.4$	&	$64.0$	\\ 
	&$0.8$	&	$81.0$	&	$88.0$	&	$78.4$	&	$80.6$	&	$75.8$	&	$64.6$	\\ 
	&$0.9$	&	$83.6$	&	$90.7$	&	$83.8$	&	$83.4$	&	$77.7$	&	$69.7$	\\ 
	&$0.95$	&	$87.5$	&	$94.3$	&	$92.1$	&	$87.4$	&	$81.1$	&	$76.5$	\\ 
	&$0.97$	&	$91.8$	&	$96.7$	&	$96.1$	&	$91.7$	&	$85.8$	&	$83.3$	\\ 
	&$0.99$	&	$99.5$	&	$99.9$	&	$99.9$	&	$99.5$	&	$96.7$	&	$98.0$	\\ 
\end{tabular}
\caption{The size (\%) and power (\%) of the mean tests in Section \ref{sec:unitroottesting}
under $\mu=0$ and $\mu=0.2$, respectively, 
where 
$T_n^{\textsc{para}}$, $T_n^{\textsc{un}}$, $T_n^{\textsc{lug}}$, $T_n^{\textsc{pw}}$, and $T_n^{\textsc{tail}}$ are
Wald's tests, while $T_n^{\textsc{sn}}$ is a self-normalized test. 
The $\AR$ parameter is $\phi\in\{0, 0.7, 0.8, 0.9, 0.95, 0.97, 0.99\}$. 
The nominal size is $\alpha=5\%$.
  The sample size $n=200$.}\label{tab:unitroottest}
\end{table}
Consequently, 
six different tests are considered:
\begin{gather*}
	T_n^{\textsc{para}} = T_n(\bar{v}_{\textsc{para}}), \quad
	T_n^{\textsc{un}} = T_n(\tilde{v}_{\textsc{un}}),\quad
	T_n^{\textsc{lug}} = T_n(\hat{v}_{\textsc{lug}}), \\
	T_n^{\textsc{pw}} = T_n(\hat{v}_{\textsc{pw}}), \quad
	T_n^{\textsc{tail}} = T_n(\hat{v}_{\textsc{tail}}),\quad
	T^{\textsc{sn}}_n.
\end{gather*}
The nominal size is set as $\alpha = 5\%$. 
We reject the $H_0$ based on $T_n^{\bullet}$ if $T_n^{\bullet} > c_{\bullet}$ 
for $\bullet\in\{\textsc{para}, \textsc{un}, \textsc{lug}, \textsc{para}, \textsc{pw}, \textsc{tail}\}$, 
where the finite-sample critical value $c_{\bullet}$ is found by 
simulation with independent data $X_1, \ldots, X_n \sim \Normal(0,1)$ based on 50,000 replications.

The sample size is $n=200$. 
The size and the power are computed when $\mu=0$ and $\mu=0.2$, respectively; 
see Table~\ref{tab:unitroottest}.  
Under the null hypothesis $H_0$, 
all tests demonstrate satisfactory size control when the data are uncorrelated (i.e., $\phi=0$). 
When $\phi$ increases, the size of $T_n^{\textsc{un}}$ is inflated most quickly. 
It shows that the adjustments employed in  
$T_n^{\textsc{un}}$, $T_n^{\textsc{lug}}$, $T_n^{\textsc{pw}}$, $T_n^{\textsc{tail}}$, and $T_n^{\textsc{sn}}$
are effective. 
However, when $\phi$ further increases to the near unit root situation (i.e., $\phi\geq 0.95$), 
all tests show various degrees of inflation in size. 
Among them, the proposed $T_n^{\textsc{tail}}$ and the self-normalized test $T_n^{\textsc{sn}}$
have the best performance and they perform similarly. 
In this example, the near unit root issue has a less impact on the proposed tail postcoloring technique 
compared with the standard prewhitening technique. 

Under the alternative hypothesis $H_1$, 
all Wald's tests (i.e., $T_n^{\textsc{para}}$, $T_n^{\textsc{un}}$, $T_n^{\textsc{lug}}$, $T_n^{\textsc{pw}}$, and $T_n^{\textsc{tail}}$)
show similar power when the data are uncorrelated, 
while the self-normalized test  $T_n^{\textsc{sn}}$ has an obviously lower power. 
It shows that self-normalization, or more boardly fixed-$b$ approach, has a negative impact on power. 
When $\phi$ increases, the proposed $T_n^{\textsc{tail}}$ still shows higher power than 
the self-normalized test $T_n^{\textsc{sn}}$ except the case of $\phi=0.99$.
In a nutshell, 
the proposed tail postcolored estimator has an advantage in mean testing 
when the process is close to having a unit root.

The sample size is $n = 200$. 
The size and the power are computed when $\mu = 0$ and $\mu = 0.2$, respectively; 
see Table~\ref{tab:unitroottest}. 
Under the null hypothesis $H_0$, 
all tests demonstrate satisfactory size control when the data are uncorrelated (i.e., $\phi = 0$). 
When $\phi$ increases, the size of $T_n^{\textsc{un}}$ is inflated most quickly. 
This indicates that the adjustments employed in 
$T_n^{\textsc{un}}$, $T_n^{\textsc{lug}}$, $T_n^{\textsc{pw}}$, $T_n^{\textsc{tail}}$, and $T_n^{\textsc{sn}}$ are effective. 
However, when $\phi$ further increases to a near unit root situation (i.e., $\phi \geq 0.95$), 
all tests show various degrees of inflation in size. 
Among them, the proposed $T_n^{\textsc{tail}}$ and the self-normalized test $T_n^{\textsc{sn}}$ 
demonstrate the best performance and perform similarly. 
We observe that 
the near unit root issue has a lesser impact 
on the proposed tail postcoloring technique compared with the standard prewhitening technique.

Under the alternative hypothesis $H_1$, 
all Wald's tests 
(i.e., $T_n^{\textsc{para}}$, $T_n^{\textsc{un}}$, $T_n^{\textsc{lug}}$, $T_n^{\textsc{pw}}$, and $T_n^{\textsc{tail}}$) 
show similar power when the data are uncorrelated, 
while the self-normalized test $T_n^{\textsc{sn}}$ has a noticeably lower power. 
This indicates that self-normalization, or more broadly a fixed-$b$ approach, 
negatively impacts power. 
When $\phi$ increases, 
the proposed $T_n^{\textsc{tail}}$ still shows higher power than the self-normalized test $T_n^{\textsc{sn}}$, 
except in the case of $\phi = 0.99$, where all tests exhibit seriously inflated size. 

In summary, the proposed tail postcolored estimator has an advantage in mean testing when the process is close to having a unit root
in terms of size control and power performance.
For further discussion of this type of testing problem, 
we refer readers to the discussion in \cite{perron2011irrelevance}.  

\subsection{Simulation results for nonlinear time series}\label{sec:sim_nonL}

In this section, we consider two other models for generating the data and compare the estimators.

\begin{example}[Threshold autoregressive model]\label{exp:tar}\leavevmode
Consider the threshold autoregressive (TAR) model \citep{tong1978threshold}:
$X_{i}=\rho_{1}\max(X_{i-1},0)+\rho_{2}\min(0,X_{i-1})+\varepsilon_{i}$,
where $\varepsilon_i$'s follow $\Normal(0,1)$ independently
and $\rho_1, \rho_2 \in(-1,1)$.
For each observation, it switches between two AR$(1)$ models
with AR parameters $\rho_{1}$ and $\rho_{2}$. 
By Theorem 1 in \cite{wu2005},
if $\max{(\rho_{1},\rho_{2})}<1$ and $\E(\varepsilon_t^4)<\infty$, 
the series is stationary.
Throughout this subsection, we set $\rho_2=0.5$.
The results can be found in Tabel~\ref{tab:tar_simu_results}.

We can observe that in terms of 
the mean-squared error,
our estimator is the best in almost all models.
Regarding bias correction,
our estimator is more conservative compared with the standard prewhitening and the lugsail kernel estimators.
The downward bias is reduced in magnitude without changing the sign.

\begin{table}[t]
\centering
\small
\def~{\hphantom{0}}
\begin{tabular}{cccccccccc}
  Criteria & $\rho_1\backslash$ Estimator  & $\tilde{v}_{\textsc{un}}$ & $\hat{v}_{\textsc{pw}}$ & $\tilde{v}_{\textsc{lug}}$ & $\bar{v}_{\textsc{para}}$ & $\hat{v}_{\textsc{tail}}$ \\[0.5ex]
   $100\MSE(\cdot)/v^2$
  & $~~0$ & $~6.86$ & $~9.46$ & $~9.43$ & $16.02$ & $~5.81$ \\  
  & $0.1$ & $~7.16$ & $~8.24$ & $10.64$ & $11.98$ & $~5.78$ \\  
  & $0.2$ & $~7.76$ & $~7.78$ & $13.49$ & $~9.68$ & $~6.31$ \\  
  & $0.3$ & $~8.36$ & $~7.07$ & $16.53$ & $~7.78$ & $~6.57$ \\  
  & $0.4$ & $~8.81$ & $~7.09$ & $18.40$ & $~7.25$ & $~7.03$ \\  
  & $0.5$ & $10.03$ & $~7.16$ & $17.57$ & $~6.99$ & $~7.07$ \\  
  & $0.6$ & $12.14$ & $~9.74$ & $18.72$ & $~9.90$ & $~8.79$ \\  
  & $0.7$ & $15.59$ & $15.27$ & $19.06$ & $17.65$ & $10.41$ \\  
  & $0.8$ & $23.37$ & $43.08$ & $27.76$ & $55.74$ & $16.53$ \\[0.5ex]
    $100\{ \E(\cdot)/v - 1 \}$
  & $~~0$ & $-17.44$ & $~9.77$ & $-15.94$ & $20.58$ & $-4.52$ \\  
  & $0.1$ & $-18.84$ & $~7.14$ & $-15.12$ & $14.17$ & $-4.66$ \\  
  & $0.2$ & $-20.20$ & $~4.81$ & $-12.17$ & $~9.11$ & $-4.51$ \\  
  & $0.3$ & $-22.22$ & $~1.91$ & $~-7.67$ & $~4.14$ & $-4.97$ \\  
  & $0.4$ & $-23.33$ & $~1.58$ & $~~1.01$ & $~2.75$ & $-3.62$ \\  
  & $0.5$ & $-26.47$ & $~0.37$ & $~~6.31$ & $~0.93$ & $-4.62$ \\  
  & $0.6$ & $-29.99$ & $~2.52$ & $~~9.24$ & $~3.53$ & $-4.83$ \\  
  & $0.7$ & $-35.67$ & $~7.07$ & $~~8.19$ & $10.37$ & $-6.44$ \\  
  & $0.8$ & $-45.48$ & $21.27$ & $~~4.51$ & $29.43$ & $-9.67$ \\ [0.5ex]
\end{tabular}
\caption{
 The standardized mean-squared error $100\MSE(\cdot)/v^2$
and standardized bias $10\{ \E(\cdot)/v - 1 \}$
of 
the unprewhitened estimator $\tilde{v}_{\textsc{un}}$,
the classical prewhitening estimator $\hat{v}_{\textsc{pw}}$,
the lugsail kernel estimator $\tilde{v}_{\textsc{lug}}$ in \cite{vats2022lugsail},
the \AR$(1)$ parametric estimator $\bar{v}_{\textsc{para}}$,
and the proposed tail postcolored estimator $\hat{v}_{\textsc{tail}}$
with optimal bandwidth estimated by the nonparametric plug-in method.
The model parameters are $\rho_{1} \in\{ 0,0.1,0.2, \ldots,0.8\}$
and $\rho_{2}=0.5$.
The sample size is $n=200$.
}
\label{tab:tar_simu_results}
\end{table}

\end{example}

\begin{example}[Geometric \AR$(1)$ model]\label{exp:geo_ar1}
  Consider the geometric \AR$(1)$ model
  specified as $\log X_i = \varphi \log X_{i-1} + \varepsilon_{t}$
  where $\varepsilon_i$'s follow $\Normal(0,1)$ independently.
  The true long-run variance is computed empirically.
  Estimators compared in (\ref{exp:tar}) are computed again.
  The performance results are displayed in Table~\ref{tab:geoar}.
  It can be observed that especially for the positively autocorrelated models,
  our proposal exhibits satisfactory MSE and bias correction results.

  \begin{table}[t]
\centering
\small
\def~{\hphantom{0}}
\begin{tabular}{ccrrrrr}
  Criteria & $\varphi\backslash$ Estimator  & $\tilde{v}_{\textsc{un}}$ & $\hat{v}_{\textsc{pw}}$ & $\tilde{v}_{\textsc{lug}}$ & $\bar{v}_{\textsc{para}}$ & $\hat{v}_{\textsc{tail}}$ \\[0.5ex]
   $100\MSE(\cdot)/v^2$ 
  & $-0.7$ & $40.65$ & $~39.22$ & $42.91$ & $~~19.69$ & $42.56$ \\  
  & $-0.5$ & $47.85$ & $~44.58$ & $46.70$ & $~~~8.06$ & $51.95$ \\  
  & $-0.3$ & $49.32$ & $~42.51$ & $45.55$ & $~~12.37$ & $52.46$ \\  
  & $~0.3$ & $~3.73$ & $~39.54$ & $~5.90$ & $~598.62$ & $~2.20$ \\  
  & $~0.5$ & $~5.27$ & $102.00$ & $~7.11$ & $~348.12$ & $~2.85$ \\  
  & $~0.7$ & $~9.49$ & $796.90$ & $~9.43$ & $1427.85$ & $~4.19$ \\ [0.5ex]
    $100\{ \E(\cdot)/v - 1 \}$
  & $-0.7$ & $-63.45$ & $-62.39$ & $-65.33$ & $-42.29$ & $-64.81$ \\ 
  & $-0.5$ & $-69.00$ & $-66.59$ & $-68.20$ & $-16.44$ & $-71.94$ \\ 
  & $-0.3$ & $-70.02$ & $-64.81$ & $-67.34$ & $~~8.82$ & $-72.27$ \\ 
  & $~0.3$ & $-16.98$ & $~45.32$ & $-19.79$ & $220.60$ & $~-7.40$ \\ 
  & $~0.5$ & $-20.69$ & $~53.70$ & $-15.26$ & $152.14$ & $~-7.40$ \\ 
  & $~0.7$ & $-28.07$ & $~45.28$ & $~-2.25$ & $~83.25$ & $~-9.29$ \\ [0.5ex]
\end{tabular}
\caption{
 The standardized mean-squared error $100\MSE(\cdot)/v^2$
and standardized bias $10\{ \E(\cdot)/v - 1 \}$
of 
the unprewhitened estimator $\tilde{v}_{\textsc{un}}$,
the classical prewhitening estimator $\hat{v}_{\textsc{pw}}$,
the lugsail kernel estimator $\tilde{v}_{\textsc{lug}}$ in \cite{vats2022lugsail},
the \AR$(1)$ parametric estimator $\bar{v}_{\textsc{para}}$,
and the proposed tail postcolored estimator $\hat{v}_{\textsc{tail}}$
with optimal bandwidth estimated by the nonparametric plug-in method.
The geometric \AR$(1)$ has parameter $\varphi \in \{\pm 0.3, \pm 0.5, \pm 0.7\}$
The sample size is $n=200$.
}\label{tab:geoar}
\end{table}
\end{example}

}

\subsection{Simulation results on multivariate time series}\label{sec:simulation_multi}
In this subsection, we compare the following 
long-run covariance matrix estimators, 
which are also used in Section~\ref{subsec:HAC}.

\begin{itemize}
\item[(a)]\label{est:noPW} The unadjusted Bartlett kernel estimator  
      $$
        \tilde{V}_{\textsc{un}} 
          = \tilde{V} (\hat{\ell}_{\textsc{un}};  K_{\Bart}; X_{1:n})
          = \sum_{|k|\leq \hat{\ell}_{\textsc{un}}} K_{\Bart}(k/\hat{\ell}_{\textsc{un}}) \tilde{\Gamma}_k,
    \quad \text{where} \quad
    \hat{\ell}_{\textsc{un}}
        = \left\lceil \left[ \frac{3n}{2} \frac{\sum_{j=1}^{d} \{v_1(\bar{\phi}_{j}, \bar{\sigma}^2_{zj})\}^2 }{ \sum_{j=1}^{d} \{v_0(\bar{\phi}_{j}, \bar{\sigma}^2_{zj})\}^2} \right]^{1/3} \right\rceil
      $$ 
      is \cite{andrews1991}'s parametric plug-in bandwidth estimator using $d$ univariate \textsc{ar}$(1)$ models; and 
      \begin{gather*}
        \bar{\phi}_j = \frac{\tilde{\Gamma}^{(j,j)}_1}{\tilde{\Gamma}^{(j,j)}_0} ,\qquad
    v_1(\bar{\phi}_{j}, \bar{\sigma}^2_{zj}) = \frac{2\bar{\phi}_j\bar{\sigma}^2_{zj}}{ (1-\bar{\phi}_j)^3(1+\bar{\phi}_j)^2 },\qquad
    v_0(\bar{\phi}_{j}, \bar{\sigma}^2_{zj}) = \frac{\bar{\sigma}^2_{zj} }{ (1-\bar{\phi}_j)^2} ,\\
    \bar{\sigma}^2_{zj} = \frac{1}{n-2}\sum_{i=2}^n \left( Z^{*(j)}_i - \bar{Z}^{*(j)}\right)^2 , \qquad
    \bar{Z}^{*(j)} = \frac{1}{n-1}\sum_{i=1}^n Z_i^{*(j)}, \qquad 
    Z^{*(j)}_{i} = {X^{(j)}_i -  \bar{\phi}_j X^{(j)}_{i-1}} .
     \end{gather*}
     When $d=1$, the estimator $\tilde{V}_{\textsc{un}}$ reduces to $\tilde{v}_{\textsc{un}}$ defined in Example \ref{exp:simu_ar1}
     with the optimal bandwidth $\ell_{\textsc{un}}$ estimated by $\hat{\ell}_{\textsc{un}}$.

\item[(b)] \cite{andrews1992}'s \textsc{var}$(1)$ prewhitened version of (a): 
      $$
        \hat{V}_{\textsc{pw}} 
          =  \left(I_d - \bar{\Phi}^{\T} \right)^{-1} 
        \tilde{V} \left(\hat{\ell}_{\textsc{pw}};  K_{\Bart}; Z_{2:n} \right) 
        \left(I_d - \bar{\Phi}^{\T} \right)^{-1},
    \;\; \text{where} \;\;
    \hat{\ell}_{\textsc{pw}}
        = \left\lceil \left[ \frac{3n}{2} \frac{\sum_{j=1}^{d} \{v_1(\bar{\phi}_{zj} , \bar{\sigma}^2_{zzj})\}^2 }
                  { \sum_{j=1}^{d} \{v_0(\bar{\phi}_{zj} , \bar{\sigma}^2_{zzj} )\}^2} \right]^{1/3} \right\rceil
      $$
    and 
    \begin{gather} 
        \bar{\Phi} =  \tilde{\Gamma}_0^{-1}\tilde{\Gamma}_1  , \qquad
        Z_i = X_i - X_{i-1} \bar{\Phi}^{\T}; \label{eqt:defPhiBar}\\
         \bar{\phi}_{zj} = \frac{\sum_{i=3}^n (Z_i^{*(j)} - \bar{Z}^{*(j)})(Z_{i-1}^{*(j)} - \bar{Z}^{*(j)})}{\sum_{i=3}^n (Z_i^{*(j)} - \bar{Z}^{*(j)})^2} , \quad 
        \bar{\sigma}^2_{zzj} = \frac{1}{n-3}\sum_{i=3}^n \left( Z^{*(j)}_i - \bar{\phi}_{zj} Z^{*(j)}_{i-1} \right)^2. \nonumber
      \end{gather}
    The bandwidth $\hat{\ell}_{\textsc{pw}}$
      is computed similarly as $\hat{\ell}_{\textsc{un}}$
      except that the data $\{X_i\}_{i=1}^n$ is replaced with $\{Z_i\}_{i=2}^{n}$.
      We remark that $Z_i$ is not the same as $(Z_i^{*(1)}, \ldots, Z_i^{*(d)})^{\T}$
      as $Z_i$ is the residual under $\textsc{var}(1)$ model whereas 
      $(Z_i^{*(1)}, \ldots, Z_i^{*(d)})^{\T}$ is the residual under $d$ $\textsc{ar}(1)$ models.
      When $d=1$, the estimator $\hat{V}_{\textsc{pw}}$ reduces to $\hat{v}_{\textsc{pw}}$ defined in Example \ref{exp:simu_ar1}
     with the optimal bandwidth $\ell_{\textsc{pw}}$ estimated by $\hat{\ell}_{\textsc{pw}}$.

\item[(c)] \cite{Flegal:2021uk}'s lugsail kernel estimator
    \[
      \tilde{V}_{\textsc{lugs}} = \tilde{V}(\hat{\ell}_{\textsc{l}}; K_{\textsc{l}}; X_{1:n})
      \quad \text{where} \quad 
      K_{\textsc{l}}(t) = ({1-c})^{-1}K_{\Bart}(t) - c({1-c})^{-1}K_{\Bart}(rt)
    \]
    with their default parameters $c=1/2$ and $r=3$. 
        It is computed with the function \texttt{mcse.multi} in the R package \texttt{mcmcse}.

\item[(d)] The proposed \textsc{var}$(1)$ tail precolored estimator:
  \[
    \hat{V}_{\textsc{tail}} = \hat{V}(\hat{\ell}_{\textsc{tail}};K_{\Bart};\bar{\Phi}), 
    \quad \text{where} \quad
    \hat{\ell}_{\textsc{tail}} = \left\lceil {\left( \frac{3 \hat{\xi}_p^2 n}{2}\right)}^{1/3} \right\rceil ,
    \quad 
    \hat{\xi}_p  = \frac{\sum_{j=1}^{d} \tilde{V}^{(j,j)}_{1,\#}}{\sum_{j=1}^{d} \tilde{V}_{0,\#}^{(j,j)}} - \frac{\sum_{j=1}^{d} {V^{(j,j)}_1(\bar{\Phi})}}{ \sum_{j=1}^{d} V^{(j,j)}(\bar{\Phi})},
  \]
  the estimator $\bar{\Phi}$ is defined as in (\ref{eqt:defPhiBar}), and 
  \begin{gather}
    V_1(\bar{\Phi},\bar{\Sigma}_z) 
      = {\Gamma}_1(\bar{\Phi},\bar{\Sigma}_z) (\bar{\Psi}\bar{\Phi}^{\T}\bar{\Psi}+ \bar{\Psi}),  \qquad
      V_0(\bar{\Phi},\bar{\Sigma}_z) 
      = \Gamma_{0}(\bar{\Phi},\bar{\Sigma}_z) +  {\Gamma}_1(\bar{\Phi},\bar{\Sigma}_z) \bar{\Psi} +  \{ {\Gamma}_1(\bar{\Phi},\bar{\Sigma}_z)\bar{\Psi} \}^{\T}, \nonumber\\
    \vect\{\Gamma_0(\bar{\Phi},\bar{\Sigma}_z) \} 
      = (I_{d^2} - \bar{\Phi}\otimes \bar{\Phi})^{-1} \vect(\bar{\Sigma}_z),  \qquad
      {\Gamma}_1(\bar{\Phi},\bar{\Sigma}_z) 
      = \Gamma_0(\bar{\Phi},\bar{\Sigma}_z) \bar{\Phi}^{\T}, \label{eqt:Gamma01}\\
      \bar{\Psi} = ({I}_d -\bar{\Phi}^{\T} )^{-1},  \qquad
    \bar{\Sigma}_z = \frac{1}{n-1}\sum_{i=2}^n (Z_i-\bar{Z}) (Z_i-\bar{Z})^{\T} , \qquad 
    \bar{Z}= \frac{1}{n-1}\sum_{i=2}^nZ_i ,  \nonumber\\
    \tilde{V}^{(j,j)}_{r,\#} = \sum_{|k|\leq \ell^{\#}_{r}}  K_{\Bart} (k/\ell_{r}^{\#})|k|^{r} \tilde{\Gamma}^{(j,j)}_k , \qquad \ell^{\#}_{r} = \lceil{2n^{{1}/({2r+3})}}\rceil \qquad (r=0,1). \nonumber
  \end{gather}
  In (\ref{eqt:Gamma01}), $\vect(\cdot)$ is the vectorization operation of a matrix and $I_r$ is a $r\times r$ identity matrix.
  When $d=1$, the estimator $\hat{V}_{\textsc{tail}}$ reduces to $\hat{v}_{\textsc{tail}}$ defined in Example \ref{exp:simu_ar1}
     with the optimal bandwidth $\ell_{\textsc{tail}}$ estimated by $\hat{\ell}_{\textsc{tail}}$.
\end{itemize}
The eigenvalue adjustment in Section~\ref{sec:MVextension} is applied on the above estimators if they are not positive-definite. 
We remark that, in the simulation experiments presented in Section~\ref{subsec:HAC}, 
the estimators 
$\tilde{V}_{\textsc{un}}$, $\hat{V}_{\textsc{pw}}$, $\tilde{V}_{\textsc{lugs}}$, and $\hat{V}_{\textsc{tail}}$ 
are computed on the time series $\{W_i \in \mathbb{R}^d\}$ as the input data, i.e., setting $X_i = W_i$ ($i=1, \ldots,n$). 

We present simulation results for the following multivariate time series models:

\begin{example}[\textsc{varma}$(1,1)$]\label{exp:varma}
  Consider the model ${X}_i = {\Phi} {X}_{i-1} + {\Upsilon} {\varepsilon}_{i-1} + {\varepsilon}_{i}$,
  where ${\varepsilon}_{i} \sim \Normal_2({0}, {\Sigma}_{\varepsilon})$ independently and 
  ${\Phi, \Upsilon},{\Sigma}_{\varepsilon} \in \mathbb{R}_{2\times 2}$. 
  The true long-run variance matrix for this model can be computed as:
  \begin{eqnarray*}
    {V} 
      &=& {\Gamma}_{0} +  {\Gamma}_{1} {\Psi} +  \{ {\Gamma}_{1} {\Psi} \}^{\T},\\
    \vect({{\Gamma}_0}) &=& ({I}_4 - {\Psi} \otimes {\Psi})^{-1} 
    \vect\{{\Sigma}_{\varepsilon} + ({\Phi} + {\Upsilon}) {\Sigma}_{\varepsilon}({\Phi} + {\Upsilon})^\T 
    - {\Phi} {\Sigma}_{\varepsilon} {\Phi} ^\T \},\\
    {\Psi} &=& ({I}_4 -{\Phi}^{\T} )^{-1}.
  \end{eqnarray*}

  We consider a series of \textsc{varma}$(1,1)$ models with different correlation structure.
  Simulation results are displayed in Table~\ref{tab:VARMA1}.
  Generally, the proposed estimator $\hat{V}_{\textsc{tail}}$ has the lowest weighted mean-squared error in most cases. 
  Besides, $\hat{V}_{\textsc{tail}}$ does not reveal explosive weighted mean-squared error when the magnitude of the diagonal entries of $\Phi$ are close to one. 
  Hence, tail prewhitening is more reliable than the standard prewhitening especially when there exists strong autocorrelation. 
\end{example}

\begin{table}[t]
  \centering
  \small
  \def~{\hphantom{0}}
  \begin{tabular}{l cccc cccc}
                & \multicolumn{4}{c}{$n=200$} & \multicolumn{4}{c}{$n=1000$} \\
    $\vect({\Phi})$$\backslash$Estimator & $\tilde{V}_{\textsc{un}}$ & $\hat{V}_{\textsc{pw}}$ & $\tilde{V}_{\textsc{lugs}}$ & $\hat{V}_{\textsc{tail}}$
      & $\tilde{V}_{\textsc{un}}$ & $\hat{V}_{\textsc{pw}}$ & $\tilde{V}_{\textsc{lugs}}$ & $\hat{V}_{\textsc{tail}}$ \\
  $(0.7, 0, 0, 0.7)^{\T}$     & $~69.86$  & $~54.34$  & $~58.98$  & $ 51.16$  & $ 24.76$  & $~10.25$  & $ 24.16$  & $ 12.61$  \\
  $(0.5, 0, 0, 0.5)^{\T}$     & $~~7.13$  & $~~4.47$  & $~10.39$  & $ ~4.34$  & $ ~2.62$  & $~~1.27$  & $ ~2.51$  & $ ~1.10$  \\
  $(0.3, 0, 0, 0.3)^{\T}$     & $~~3.34$  & $~~1.24$  & $~~6.38$  & $ ~1.08$  & $ ~2.89$  & $~0.41$ & $ ~6.04$  & $ ~0.32$  \\
  $(0.5,0.3,-0.3,-0.5)^{\T}$    & $~~2.37$  & $~~5.24$  & $~~3.56$  & $ ~3.15$  & $ ~0.69$  & $~~3.26$  & $ ~0.72$  & $ ~0.90$  \\
  $(0.7,0.3,-0.3,-0.7)^{\T}$    & $~10.26$  & $~52.37$  & $~20.43$  & $ 32.25$  & $ ~3.10$  & $~39.90$  & $ ~3.37$  & $ 13.83$  \\
  $(0.9,0.3,-0.3,-0.9)^{\T}$    & $110.38$  & $226.05$  & $176.33$  & $ 79.26$  & $ 40.92$  & $138.06$  & $ 61.68$  & $ 24.71$  \\
  \end{tabular}
  \caption{Simulation results for data generated from 
  Example \ref{exp:varma}. 
  The weighted mean-squared error is computed with $W=I_4$.
  Let $\vect({\Theta}) = (0.5,0.7,-0.7,0.5)^\T$ and $\vect({\Sigma}_{\epsilon}) = (2,0,0,1)^{\T}$ for all models.
  The models vary with the value of ${\Phi}$.}\label{tab:VARMA1}
\end{table} 

\subsection{Simulation results under another two-model example}

\begin{example}[Two-model example]\label{eg:twoModel_AR2}
    Consider two prewhitening models:
    \begin{itemize}
		\item[(i)] $X_i = \phi_1 X_{i-1} + 0.5 X_{i-2} + \varepsilon_i$ with independent $\varepsilon_i\sim \Normal(0,\sigma^2)$; and
    	\item[(ii)] $X_i = \phi_2 X_{i-1} - 0.5 X_{i-2} + \varepsilon_i$ with independent $\varepsilon_i\sim \Normal(0,\varsigma^2)$.
	\end{itemize}
    Through minimizing the sum of squared residuals,
    $\phi_1$ and $\phi_2$ can be consistently estimated by $\bar{\phi}_1 = 0.5\bar{\phi}$ 
    and $\bar{\phi}_2 = 1.5\bar{\phi}$, respectively, 
    where $\bar{\phi} = \tilde{\gamma}_1/\tilde{\gamma}_0$.
    Consider the data generated from the \textsc{ar}$(2)$ model,
    i.e., 
    $X_i = a_1 X_{i-1} + a_2 X_{i-2}+ \varepsilon_i$,
    where $\varepsilon_i\sim\Normal(0,1)$ independently.
    The estimator defined in (\ref{eqt:submodelEst}) computed for models (i) and (ii)
    are denoted as $\hat{v}^{[1]}_{\textsc{tail}}$ and $\hat{v}^{[2]}_{\textsc{tail}}$, respectively.
    Also compute the combined estimator $\hat{v}^{[1:2]}_{\textsc{tail}}$ with weight function defined in Example~\ref{rem:weightChoice} (ii).
    Fix $a_1 = 0.1$ and $n=200$.
    The simulation result is plotted in Figure~\ref{fig:2model}.
    It can be observed that the combined estimator dominates both of the single-model estimators
    when $a_2\geq 0$. 
    When $a_2<0$, it performs worse than the well-specified estimator $\hat{v}^{[2]}_{\textsc{tail}}$
    but is much better than the misspecified estimator $\hat{v}^{[1]}_{\textsc{tail}}$.
    The new estimator has a bias in the middle of the two single-model biases, which agrees with intuition.
  \end{example}

  \begin{figure}[t]
    \centering
    \includegraphics[width=0.75\textwidth]{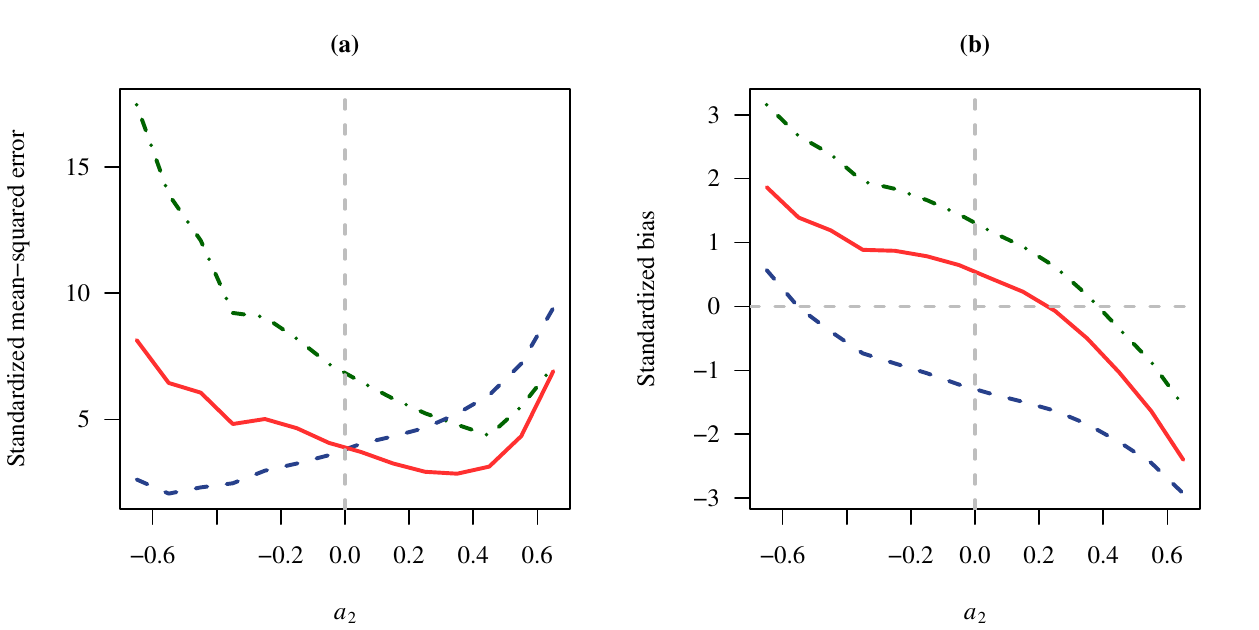}
    \caption{
    (a) and (b) plot the standardized mean-squared error and standardized bias
     $\E(\cdot)/v - 1$ of $\hat{v}^{[1]}_{\textsc{tail}}$ (green dot-dashed),
      $\hat{v}^{[2]}_{\textsc{tail}}$ (blue dashed) and $\hat{v}^{[1:2]}_{\textsc{tail}}$
     (red solid) against different $a$, respectively.
     }\label{fig:2model}
  \end{figure}

\section{Techinical proofs}

\subsection{Proof of main results}

Throughout this section, 
we may omit the specification of the kernel $K$ and bandwidth $\ell$ in the notation of the following estimators
when there is no confusion:
\begin{align*}
	\tilde{v} &= \tilde{v}(\ell; K), &
	\hat{v}(\theta) &= \hat{v}(\ell; K; \theta), &
	\hat{v}^{[1:J]}({\theta}_{1:J}) &= \hat{v}^{[1:J]}(\ell_{ 1:J}; K; {\theta}_{1:J}), \\
	\tilde{V} &= \tilde{V}(\ell; K), &
	\hat{V}( {\theta}) &= \hat{V}(\ell;K; {\theta}).&&
\end{align*}

\begin{proof}[Proof of Theorem~\ref{thm:BiasVar}]
  We first derive the bias and variance of $\hat{v}(\ell; K; \theta_{\star})$.
  First
  let $r(t) = K(t) - \left( 1 + B{|t|}^{p} \right) \mathbb{1}(|t|\leq1)$.
  Since $K(t) = K(-t)$ for all $t\geq 0$ and
  $K(0)=1$,
  $\lim_{t \to 0 } |r(t)| / {|t|}^{p} = 0$.
  The bias is computed following the definition as 
  \[
    \E\left\{ \hat{v}(\theta_{\star})\right\} - v 
      = \left\{ \frac{\sum_{k\in\mathbb{Z}}\gamma_k(\theta_{\star})}{\sum_{k\in\mathbb{Z}} K(k/\ell) \gamma_k(\theta_{\star})} - 1\right\} \E(\tilde{v})
      +\E(\tilde{v}) - v.
  \]
  In particular, consider the coefficient
  \begin{equation}\label{eqt:bvcoef1}
    \frac{\sum_{k\in\mathbb{Z}}\gamma_k(\theta_{\star})}{\sum_{k\in\mathbb{Z}} K(k/\ell) \gamma_k(\theta_{\star})} - 1
      = \frac{\sum_{k\in\mathbb{Z}}\gamma_k(\theta_{\star})}{\sum_{|k| \leq \ell} \left(1 + B{\left\lvert k/\ell \right\vert}^{p}\right)\gamma_k(\theta_{\star}) + \sum_{k\in\mathbb{Z}} r\left( k/\ell \right)\gamma_k(\theta_{\star})} - 1.
  \end{equation}
  Compute the terms separately,
  \begin{align*}
    \sum_{|k| \leq \ell} \left(1 + B{\left\lvert \frac{k}{\ell}\right\vert}^{p}\right)\gamma_k({\theta_{\star}})
      &= \sum_{k\in\mathbb{Z}} \gamma_k({\theta_{\star}}) -  \sum_{|k| > \ell} \gamma_k({\theta_{\star}}) + B\sum_{|k| \leq \ell}{\left\lvert \frac{k}{\ell}\right\vert}^{p}\gamma_k({\theta_{\star}});\\
   B\sum_{|k| \leq \ell} {\left\lvert \frac{k}{\ell}\right\rvert}^{p}\gamma_k({\theta_{\star}})
   & = B \sum_{k \in \mathbb{Z}} {\left\lvert \frac{k}{\ell}\right\rvert}^{p}\gamma_k({\theta_{\star}}) 
      - B \sum_{|k| > \ell} {\left\lvert \frac{k}{\ell}\right\rvert}^{p}\gamma_k({\theta_{\star}})
     = B\frac{v_p}{\ell^p} + o\left(\frac{1}{\ell^p}\right),
  \end{align*}
  where in the last step, we use the assumption that $u_p(\theta_{\star})<\infty$ and 
  by Kronecker's lemma, $|\sum_{|k| > \ell}{|k/\ell|}^p \gamma_k({\theta_{\star}})| \leq \sum_{|k| > \ell} {|k/\ell|}^p |\gamma_k({\theta_{\star}})| = o(1/\ell^p)$.
  Similarly, 
  \[
    \left| \sum_{|k| > \ell} \gamma_k({\theta_{\star}}) \right|  \leq \sum_{|k| > \ell} {\left\lvert \frac{k}{\ell} \right\rvert}^p\left|\gamma_k({\theta_{\star}})\right|  = o\left(\frac{1}{\ell^p}\right).
  \]
  Therefore,
  \[
     \sum_{|k| \leq \ell} \left(1 + B{\left\lvert \frac{k}{\ell}\right\vert}^{p}\right)\gamma_k({\theta_{\star}}) =  \sum_{k\in\mathbb{Z}} \gamma_k({\theta_{\star}}) + B\frac{v_p}{\ell^p} + o\left(\frac{1}{\ell^p}\right).
  \]
  Now consider $\sum_{k\in\mathbb{Z}} r({k}/{\ell})\gamma_k({\theta_{\star}})$,
  \begin{equation}\label{eqt:bvterms1}
     \left\lvert \sum_{k\in\mathbb{Z}} r\left(\frac{k}{\ell}\right)  \gamma_k({\theta_{\star}})\right\rvert  
      \leq \left\lvert \sum_{|k| \leq \sqrt{\ell}} r\left(\frac{k}{\ell}\right)  \gamma_k({\theta_{\star}})\right\rvert  + 
          \left\lvert \sum_{|k|>\sqrt{\ell}} r\left(\frac{k}{\ell}\right)  \gamma_k({\theta_{\star}})\right\rvert.
  \end{equation}
  Since $\lim_{t \to 0 } |r(t)|  / {|t|}^p = 0$ by the definition of constant $B$, as $\ell\to \infty$,
  $$
    \sup_{|k| \leq \sqrt{\ell}} \left\mid \frac{r\left(\frac{k}{\ell} \right)}{\left(\frac{k}{\ell}\right)^{p}}\right\mid  
     = o(1)
     \quad\text{and}\quad
     \left\mid  \sum_{|k| \leq \sqrt{\ell}} r\left(\frac{k}{\ell}\right)  \gamma_k({\theta_{\star}})\right\mid  
     \leq \left\{ \sup_{|k| \leq \sqrt{\ell}} \left\mid \frac{r\left(k/\ell \right)}{\left(k/\ell\right)^{p}}\right\mid  \right\} 
     \sum_{|k| \leq \sqrt{\ell}}\left\mid  \frac{k}{\ell}\right\mid ^{p}\left\lvert\gamma_k({\theta_{\star}})\right\rvert= o\left(\frac{1}{\ell^{p}}\right),
  $$
  where $\sum_{|k| \leq \sqrt{\ell}} {|k|}^{p} \left\mid  \gamma_k({\theta_{\star}}) \right\mid  = O(1)$ 
  as $u_p({\theta_{\star}})<\infty$.
  For the second term in (\ref{eqt:bvterms1}),
  \[
    \left\mid \sum_{|k| > \sqrt{\ell}} r\left(\frac{k}{\ell} \right) \gamma_k({\theta_{\star}})\right\mid  
      = \left\mid \sum_{|k| > \sqrt{\ell}} \frac{r\left(k/\ell\right)}{\left\mid  k/\ell \right\mid^{p}}\left\mid  \frac{k}{\ell}\right\mid ^{p} \gamma_k({\theta_{\star}})\right\mid 
      \leq \left\{ \sup_{|k| > \sqrt{\ell}} \left\mid \frac{r\left( k/\ell \right)}{\left( k/\ell \right)^{p}}\right\mid  \right\} \cdot
      \sum_{|k| > \sqrt{\ell}} \left\mid  \frac{k}{\ell}\right\mid ^{p}\left\mid \gamma_k({\theta_{\star}})\right\mid .
  \]
  Note that 
  \begin{align*}
    \sup_{|k| > \sqrt{\ell}} \left\mid \frac{r\left( k/\ell \right)}{\left( k/\ell \right)^{p}}\right\mid 
    &\leq 
    \sup_{t>0} \frac{|r(t)|}{{|t|}^{p}} 
    = \sup_{t>0} \frac{|K(t) - (1 + B{|t|}^p)\mathbb{1}(|t|\leq 1)|}{{|t|}^{p}} \\
    & \leq \sup_{0<t<1} \frac{|K(t)-1|}{{|t|}^{p}} 
    + \sup_{t\geq 1} \frac{|K(t)|}{{|t|}^{p}}
    + |B| + 1 = O(1)
  \end{align*}
  as $|K(t)| \leq 1$ and $\lim_{t\downarrow 0} |K(t)-K(0)| / {|t|}^p < \infty$ by definition of $p$.
  In addition,
  $\sum_{|k| > \sqrt{\ell}} {|k/\ell|}^{p}\left\lvert \gamma_k({\theta_{\star}})\right\rvert  = o(1/\ell^{p})$
  following previous arguments.
  Hence, 
  $$
    \left\mid \sum_{|k| > \sqrt{\ell}} r\left(\frac{k}{\ell}\right) \gamma_k({\theta_{\star}})\right\mid  = o\left(\frac{1}{\ell^{p}}\right)
      \qquad\text{and}\qquad
    \left\mid  \sum_{|k| \leq \sqrt{\ell}} r\left(\frac{k}{\ell}\right)  \gamma_k({\theta_{\star}})\right\mid  = o\left(\frac{1}{\ell^{p}}\right).
  $$

  Combining the results above,
  \[
    \frac{\sum_{k\in\mathbb{Z}}\gamma_k({\theta_{\star}})}{\sum_{k\in\mathbb{Z}} K(k/\ell) \gamma_k({\theta_{\star}})} - 1 
      = \frac{\sum_{k\in\mathbb{Z}}\gamma_k({\theta_{\star}})}{\sum_{k\in\mathbb{Z}}\gamma_k({\theta_{\star}}) + Bv_p(\theta_{\star})/\ell^p + o(1/\ell^p)} - 1 
      = \frac{-Bv_p(\theta_{\star})/\ell^p}{v(\theta_{\star})} + o\left(\frac{1}{\ell^p}\right).
  \]
  Next we derive the bias of $\tilde{v}$. 
  Use the true mean $\mu$ to replace the sample mean in $\tilde{\gamma}_k$ and 
  define $\ddot{\gamma}_k = \sum_{i=|k|+1}^n (X_i - \mu)(X_{i-k} - \mu)/n$.
  Note that
  \begin{align*}
    \tilde{\gamma}_k - \ddot{\gamma}_k &= \frac{1}{n} \sum_{i=|k|+1}^n\left\{(X_i - \bar{X}_n)(X_{i-|k|} - \bar{X}_n) - (X_i - \mu)(X_{i-|k|} - \mu)\right\}\\
    &= \frac{1}{n} \sum_{i=|k|+1}^n\left\{-(X_i+X_{i-|k|} - \bar{X}_n - \mu)(\bar{X}_n - \mu)\right\}\\
    &= -\frac{1}{n} \left\{(n+|k|)\bar{X}_n - \sum_{i=1}^{|k|} X_i - \sum_{i=n-|k|+1}^n X_i - (n-|k|)\mu\right\}(\bar{X}_n - \mu)\\
    &= -\frac{1}{n} \left\{(n+|k|)\left(\bar{X}_n-\mu\right) - \sum_{i=1}^{|k|} (X_i-\mu) - \sum_{i=n-|k|+1}^n (X_i-\mu)\right\}(\bar{X}_n - \mu).
  \end{align*}
  If $|k|\leq (n-1)/2$,
  \begin{align*}
    \E\left|\tilde{\gamma}_k - \ddot{\gamma}_k\right| 
    &\leq \frac{1}{n} \left\||k|\left(\bar{X}_n-\mu\right) + \sum_{i=|k|+1}^{n-|k|} (X_i-\mu)\right\|\|\bar{X}_n - \mu\|\\
    &\leq \left\{ \frac{|k|}{n} \cdot \left\|\bar{X}_n-\mu\right\| + \frac{1}{n}\left\|\sum_{i=|k|+1}^{n-|k|} (X_i-\mu)\right\| \right\} \|\bar{X}_n - \mu\|;
  \end{align*}
  if $(n-1)/2<|k|<n$.
  \begin{align*}
    \E\left|\tilde{\gamma}_k - \ddot{\gamma}_k\right| 
    &\leq \left\{ \frac{|k|}{n} \cdot \left\|\bar{X}_n-\mu\right\| + \frac{1}{n}\left\|\sum_{i=n-|k|}^{|k|+1} (X_i-\mu)\right\| \right\} \|\bar{X}_n - \mu\|.
  \end{align*} 

  Note that $\|\bar{X}_n - \mu\| = O(1/\sqrt{n})$ and $\|\sum_{i=|k|+1}^{n-|k|} (X_i-\mu)/n\| = O\left\{(n-2|k|)^{1/2}/n\right\}$.
  As $|k|\leq n$, $\E(|\tilde{\gamma}_k - \ddot{\gamma}_k|)=O(1/n)$.
  Then we study the bias of $\ddot{\gamma}_k$. Without loss of generality,
  assume that $\mu=0$, for $\mu\neq 0$,
  we can subtract the mean $\mu$ from the data $\{X_i\}_{i=1}^n$. Then we can compute
  \[
    \E(\ddot{\gamma}_k) - \gamma_k = \frac{1}{n} \sum_{i=|k|+1}^n \E(X_iX_{i-|k|}) - \E(X_iX_{i-|k|}) = \frac{|k|\gamma_k}{n} = O\left(\frac{1}{n}\right)
  \]
  since $u_1 = \sum_{k\in\mathbb{Z}} |k| |\gamma_k| < \infty$ by assumption.
  Then we can conclude that $\E(\tilde{\gamma}_k)=\gamma_k+O(1/n)$ for all $k$.
  So the bias of $\tilde{v}$ is 
  \begin{align*}
    \E(\tilde{v}) - v &= \sum_{k=-n+1}^{n-1} K\left(\frac{k}{\ell}\right) \left\{\gamma_k + O\left(\frac{1}{n}\right)\right\} - v \\
                      &= \sum_{|k| \leq \ell} B {\left|\frac{k}{\ell}\right|}^p \gamma_k
                      + \sum_{\ell<|k|<n} r\left(\frac{k}{\ell}\right)\gamma_k
                      - \sum_{|k|\geq \ell} \gamma_k 
                      + \sum_{k=-n+1}^{n-1} K\left(\frac{k}{\ell}\right)O\left(\frac{1}{n}\right)  \\
                      &= \frac{B\{v_p+o(1)\}}{\ell^p} + o\left(\frac{1}{\ell^p}\right) + O\left(\frac{\ell}{n}\right)
  \end{align*}
  by similar arguments as previously so
  $\E(\tilde{v}) - v = B v_p/\ell^p + o(1/\ell^p) + O(\ell/n)$.
  The bias is 
  \begin{align*}
    \E\left\{ \hat{v}(\theta_{\star}) \right\} - v 
    &= \left\{ \frac{-Bv_p(\theta_{\star})}{v(\theta_{\star})\ell^p} + o\left(\frac{1}{\ell^p}\right) \right\}  \E(\tilde{v}) +  \E(\tilde{v}) - v \\
    &= \left\{ \frac{-Bv_p(\theta_{\star})}{v(\theta_{\star})\ell^p} + o\left(\frac{1}{\ell^p}\right) \right\}\left\{v+o(1)\right\} + B \frac{v_p}{\ell^{p}} + O\left(\frac{\ell}{n}\right) + o\left(\frac{1}{\ell^p}\right) \\
      &= \frac{B}{\ell^p}\left\{ v_p - \frac{v_{p}({\theta_{\star}})}{v({\theta_{\star}})} v \right\} + o\left(\frac{1}{\ell^p}\right) + O\left(\frac{\ell}{n}\right).
  \end{align*}

  As $\hat{v}(\theta_{\star}) = \{1 + O(1/\ell^p)\} \tilde{v}$,
  $\Var\{ \hat{v}(\theta_{\star})\} \sim \Var(\tilde{v})$.
  So next we derive the variance of $\tilde{v}$.
  First we replace $\bar{X}_n$ by $\mu$ and
  approximate $\tilde{v}$ by $\ddot{v}$ defined as
  \begin{align}\label{eqt:definition_vddot}
    \ddot{v}=\sum_{i=1}^{n}\sum_{j=1}^{n}K\left( \frac{|i-j|}{\ell} \right)\frac{(X_i-\mu)(X_j-\mu)}{n}.
  \end{align}
  We use Lemma A.6 in \cite{LiuChan2022} to show that $\tilde{v}$ can be approximated by $\ddot{v}$ with
  $W_n(i,j)=K(|i-j|/\ell)/n$.
  Checking the necessary conditions,
  \begin{align*}
     E_{4,n} &= \frac{1}{n}; \qquad E_{5,n} \leq E^2_{6,n} = O\left(\frac{\ell}{n}\right);\\
    E_{6,n} &= \left\vert\sum_{i=2}^{n}\sum_{j=1}^{i-1}\frac{K(|i-j|/\ell)}{n}\right\vert \leq  \sqrt{n} \left\{ \sum_{i=2}^{n}\sum_{j=1}^{i-1} \frac{K^2(  |i-j|/\ell)}{n^2} \right\}^{1/2} = O\left(\frac{\ell^{1/2}}{n^{1/2}}\right).
  \end{align*}
  Therefore,
  using Minkowski inequality and noting that 
  \[
    \E\left(\left|\tilde{v}-\ddot{v}\right|\right) \leq \|\tilde{v}-\ddot{v}\| = O\left(\frac{\ell^{1/2}}{n}\right),
  \]
  we have 
  \begin{align}\label{eqt:vddot_vtilde_equiv}
    \left\|\tilde{v}-\E(\tilde{v})\right\| = \left\|\ddot{v}-\E(\ddot{v})\right\| 
    + O\left(\frac{\ell^{1/2}}{n}\right)
   \end{align} 
  Without loss of generality,
  we assume $\mu=0$ from now on.
  If $\mu\neq0$, we can always subtract the mean from the data and all results apply.
  Now $\ddot{v} = \sum_{1\leq i,j\leq n}K(|i-j|/\ell)X_iX_j / n$
  and we want to derive the variance of $\ddot{v}$.

  To further decompose $\ddot{v}$ into multiple blocks,
  suppose that $n-\ell=m\ell+r_{n}$ for some $m,r_{n}\in\mathbb{N}_{0}$ such that $r_{n}<\ell$.
  Hence $m\sim n/\ell$.
    For $t\in\mathbb{N} \bigcap [1,s]$,  
    define $\mathcal{A}_t=\{(i,j)\in([1,n]\cap \mathbb N)^2: (t-1)\ell+1 \leq i-j\leq t \ell \}$
    and partition $\mathcal{A}_t$
    into two parts:
    $\bigcup_{k=1}^{m-t+1}\mathcal{B}_{k,t}$ and
    $\mathcal R_t = \mathcal{A}_t - \left(\bigcup_{k=1}^{m-t+1}\mathcal{B}_{k,t}\right)$, 
  where
  \[
    \mathcal B_{k,t}=\left\{(i,j) \in \mathcal{A}_t : i\in\{(k+t-1)\ell+1,\ldots,(k+t)\ell\} \right\}.
  \]
  Following this partition, define the quantities
  \begin{align}\label{DefOfGk}
    \ddot{v}_s 
      &= 2\underset{(i,j)\in\bigcup_{t=1}^s\mathcal{A}_t}{\sum\sum} K\left( \frac{|i-j|}{\ell} \right)\frac{X_iX_j}{n} \nonumber 
      + \sum_{i=1}^n K(0) \frac{X_i^2}{n}\\
      &= 2 \underbrace{\sum_{t=1}^{s} \frac{\ell}{n} \sum_{k=1}^{m-t+1} \left\{\frac{1}{\ell}\underset{(i,j)\in\mathcal B_{k,t}}{\sum\sum}  K\left(\frac{|i  -j|}{\ell}\right) X_i X_{j}\right\}}_{M_{n,s}}
      +  \underbrace{\frac{2\ell}{n}\sum_{t=1}^{s}\left\{\frac{1}{\ell} \underset{(i,j)\in\mathcal R_t}{\sum\sum}  K\left(\frac{|i-j|}{\ell}\right) X_i X_{j}\right\}}_{R_{1,n,s}} \\
      &\quad +  \underbrace{\frac{1}{n}\sum_{i=1}^n K(0)X_i^2}_{R_{2,n}} , 
  \end{align}
   so we can decompose the main part $M_{n,s}$ into multiple blocks as
   \begin{align}\label{def_Mns}
    M_{n,s} 
      &= \frac{\ell}{n} \sum_{t=1}^{s} \sum_{k=1}^{m-t+1} \underbrace{\left\{\frac{1}{\ell}\sum_{i=(k+t-1)\ell+1}^{(k+t)\ell} \sum_{j=i-t\ell }^{i-(t-1  )\ell-1} 
      K\left(\frac{|i-j|}{\ell}\right) X_i X_{j}\right\}}_{G_{k,t}}.
  \end{align}

  By construction, for fixed $s$, and $r_n<\ell$ thus
  $\|R_{1,n,s}-\E(R_{1,n,s})\| 
  =o(\|2 M_{n,s}  - \E(2 M_{n,s} )\|)$.
  We also have
  $\|R_{2,n}-\E(R_{2,n})\|=O(1/\sqrt{n})$.
  Therefore,
  $$
    \Var(\ddot{v}_s) \sim
    \Var(2M_{n,s}) =
    \Var\left(\frac{2\ell}{n}\sum_{t=1}^{s} \sum_{k=1}^{m-t+1} G_{k,t}\right).
  $$
  Further decompose the variance of $\sum_{t=1}^{s} \sum_{k=1}^{m-t+1}  G_{k,t}$ into
  \begin{equation}
    \left\|\sum_{t=1}^{s} \sum_{k=1}^{m-t+1} \left\{ G_{k,t} - \E (G_{k,t}) \right\} \right\| \lesseqgtr I_1 \pm I_0,
  \end{equation}
  where $a \lesseqgtr b \pm c$ symbolizes jointly that $a \leq b+c$ and $a \geq b-c$;
  $$
    I_1 = \left\| \sum_{t=1}^{s} \sum_{k=1}^{m-t+1} \left\{ G_{k,t}  - \E (G_{k,t}  \mid \mathcal{F}_{\varsigma_{k,t}})\right\}\right\| 
      \qquad\text{and}\qquad
    I_0 = \left\| \sum_{t=1}^{s} \sum_{k=1}^{m-t+1} \left\{ \E (G_{k,t}  \mid \mathcal{F}_{\varsigma_{k,t}}) - \E (G_{k,t}) \right\}\right\|,
  $$
  and let $\varsigma_{k,t}$ be the smallest value of the index $i$ of $X_i$ in defining $G_{k,t}$.
  Using Lemma~\ref{lemma:hoc2} and Lemma~\ref{lemma:hoc1} (iii), as $n,\ell\to\infty$ and $s$ is fixed, we have 
  \[
    \Var(\ddot{v}_s) = {4\int_{0}^{s}K^2(t)\dd t v^2} \frac{\ell}{n}+o\left\{ \frac{\ell}{n} \int_{0}^{s}K^2(t)\dd t \right\}.
  \]  
  For all $s$, the statistic $\ddot{v}_s$ can be bounded as follows:
  \[
        |\ddot{v}_s| 
        = \left\lvert \underset{|i-j|\leq s\ell}{\sum\sum} K\left( \frac{|i-j|}{\ell} \right)\frac{X_iX_j}{n} \right\rvert 
        \leq \underset{|i-j|\leq n}{\sum\sum} \left\lvert K\left( \frac{|i-j|}{\ell} \right)
        \right\rvert  \left\lvert \frac{X_iX_j}{n} \right\rvert 
        \equiv U ,
  \]
  where the upper bound $U$ satisfied that 
  \begin{align*}
     \E(U)
     &=  
     \underset{|i-j|\leq n}{\sum\sum} \left\lvert K\left( \frac{|i-j|}{\ell} \right)
       \right\rvert  \E \left\{ \left\lvert \frac{X_iX_j}{n} \right\rvert \right\}
     < \infty , \\
     \E(U^2)
     &\leq  
     \underset{|i-j|\leq n}{\sum\sum} \left\lvert K\left( \frac{|i-j|}{\ell} \right)
       \right\rvert^2  \E \left\{ \left\lvert \frac{X_iX_j}{n} \right\rvert^2 \right\} 
      \leq    
     \underset{|i-j|\leq n}{\sum\sum} \left\lvert K\left( \frac{|i-j|}{\ell} \right)
       \right\rvert^2 \frac{\| X_i \|_4^2 \| X_j\|_4^2}{n^2} 
     <\infty 
  \end{align*}
  for all $n$ since $\E(X_i^4) <\infty$.
  Therefore, by dominated convergence theorem,
  \[
    \E\left(\lim_{s\to\infty} \ddot{v}_s\right) = \lim_{s\to\infty} \E(\ddot{v}_s)
  \quad \text{and} \quad
    \E\left(\lim_{s\to\infty} \ddot{v}_s\right)^2 = \lim_{s\to\infty} \E(\ddot{v}^2_s)
  \]
  for all $n,\ell$. 
  Since $\lim_{s\rightarrow \infty} \ddot{v}_s = \ddot{v}$, we obtain
  \begin{align*}
      \Var(\ddot{v}) &= \Var\left(\lim_{s\to\infty} \ddot{v}_s\right) 
                     = \E \left(\lim_{s\to\infty} \ddot{v}_s\right)^2 -  \left\{ \E \left(\lim_{s\to\infty} \ddot{v}_s\right) \right\}^2 \\
                     &= \lim_{s\to\infty} \E (\ddot{v}^2_s )-  \lim_{s\to\infty}  \{ \E (\ddot{v}_s) \}^2 \\
                     &= \lim_{s\to\infty} \Var(\ddot{v}_s).
  \end{align*}
  Hence, as $n,\ell\to\infty$,
  \begin{align}\label{eqt:var_of_vddot}
    \Var(\ddot{v}) = {4\int_{0}^{\infty}K^2(t)\dd t v^2} \frac{\ell}{n}+o\left\{ \frac{\ell}{n} \int_{0}^{\infty}K^2(t)\dd t \right\}.
  \end{align}
  Since $A = \int_{0}^{\infty}K^2(t)\dd t < \infty$, 
  we have $\Var(\tilde{v}) \sim 4Av^2\ell/n$ 
  in view of (\ref{eqt:vddot_vtilde_equiv}) 
  and (\ref{eqt:var_of_vddot}).
  We can now conclude that when $\ell\asymp n^{1/(2p+1)}$, 
  $n^{p/(2p+1)} \left\{ \hat{v}(\theta_{\star}) - v\right\} = O_p(1)$
  by (\ref{eqt:tpwBias}) and (\ref{eqt:var_of_vddot}). 
  We also need to show that $n^{p/(2p+1)} \left\{ \hat{v}(\bar{\theta}) - \hat{v}(\theta_{\star}) \right\} = o_p(1)$.
  By Taylor's expansion, for some $\ddot{\theta}\in\mathcal{H}$ on the line segment connecting $\bar{\theta}$ and $\theta_{\star}$, 
  \begin{align*}
    n^{p/(2p+1)} \left\{ \hat{v}(\bar{\theta}) - \hat{v}(\theta_{\star}) \right\}
     &= n^{p/(2p+1)} \left.\frac{\partial \hat{v}({\theta})}{\partial \theta} \right\mid_{\theta = \ddot{\theta}} (\bar{\theta} - \theta_{\star}), 
  \end{align*}
where 
  \begin{align*}
     \left. \frac{\partial \hat{v}({\theta})}{\partial \theta} \right\mid_{\theta = \ddot{\theta}}
     &= \left\{ M(\ddot{\theta}) \frac{\partial v(\ddot{\theta})}{\partial \theta} 
     - v(\ddot{\theta}) \frac{\partial M(\ddot{\theta})}{\partial \theta}  \right\} \tilde{v}/ 
     M^2(\ddot{\theta}). 
  \end{align*}
  By assumption, $\sup_{\theta\in\mathcal{H}}\left|{\partial v(\theta)}/{\partial \theta}\right|<\infty$ and $\sup_{\theta\in\mathcal{H}} \left| {\partial M({\theta})}/{\partial \theta} \right|<\infty$,
  so ${\partial \hat{v}({\theta})}/{\partial \theta} \mid_{\theta = \ddot{\theta}} = O_p(1)$
  and $n^{p/(2p+1)} \left\{ \hat{v}(\bar{\theta}) - \hat{v}(\theta_{\star}) \right\} = o_p(1)$
  as $n^{1/2} (\bar{\theta} - \theta_{\star}) = O_p(1)$.
  Then we define and compute for $h\geq 0$,
  \[
    \epsilon^{\star}_h = \min\left[ n^{\frac{2p}{2p+1}} \left\{ \hat{v}(\bar{\theta}) - v \right\}^2, h \right] 
    - \min\left[ n^{\frac{2p}{2p+1}} \left\{ \hat{v}(\theta_{\star}) - v \right\}^2, h \right]
      = o_p(1).
  \]
  Since $\epsilon^{\star}_h$ is upper-bounded by $h$,
  by Lemma A2 of \cite{andrews1991},
  we have $\E(\epsilon^{\star}_h) = o(1)$ and thus
  \[
    \lim_{h\to\infty} \lim_{n\to\infty} \MSE_h\{\hat{v}(\bar{\theta})\}
      = \lim_{h\to\infty} \lim_{n\to\infty} \MSE_h\{\hat{v}(\theta_{\star})\}.
  \]
  Since 
  \[
    \lim_{h\to\infty} \lim_{n\to\infty} \MSE_h\{\hat{v}(\theta_{\star})\} = \lim_{n\to\infty} n^{2p/(2p+1)} \MSE\{\hat{v}(\theta_{\star})\},
  \]
  we have 
  \[
    \lim_{h\to\infty} \lim_{n\to\infty} \MSE_h\{\hat{v}(\bar{\theta})\}
      = \lim_{n\to\infty} n^{2p/(2p+1)} \MSE\{\hat{v}(\theta_{\star})\}.
  \]
	Thus, the proof is completed.  
\end{proof}

\begin{proof}[Proof of Corollary~\ref{cor:optimalBW}]
  When $\ell \asymp n^{1/(2p+1)}$, following the proof of Theorem~\ref{thm:BiasVar}, the mean-squared error of $\hat{v}(\theta_{\star})$ is 
  \begin{align}\label{eqt:MSE_vhat_in_proof}
    \MSE\{\hat{v}(\theta_{\star})\} \sim \frac{B^2v^2}{\ell^{2p}}\left\{ \frac{v_p}{v} - \frac{v_{p'}({\theta_{\star}})}{v({\theta_{\star}})} \right\}^2 + \frac{4A\ell v^2}{n}.
  \end{align}
  Minimizing the right-hand side of (\ref{eqt:MSE_vhat_in_proof}) with respect to $\ell$, the formula is obtained.
\end{proof}

\begin{proof}[Proof of Proposition~\ref{prop:dRobust}]
  In this proof, 
  all estimators use the same kernel $K$.
  Also, the optimal bandwidths $\ell_j = \ell_{\textsc{tail},j}$ ($j=1, \ldots, J$) are used.
Thus, we write $\hat{v}^{[1:J]}({\theta}_{1:J}) = \hat{v}^{[1:J]}({\theta}_{1:J}; K; \ell_{1:J})$, 
where $\theta_{1:J}$ will be specified clearly as $\bar{\theta}_{1:J}$ or ${\theta}_{1:J}^{\star}$ in the context. 

  As in the proof of Theorem~\ref{thm:BiasVar},
  since $\bar{\theta}_j = \theta_{\star j} + O_p(n^{-1/2})$ for each $j=1,\ldots,J$,
  \[
    \lim_{h\to\infty} \lim_{n\to\infty} \MSE_h\left\{ \hat{v}^{[1:J]}(\bar{\theta}_{1:J})\right\}
      = \lim_{n\to\infty} n^{\frac{2p}{2p+1}} \MSE\left\{\hat{v}^{[1:J]}(\theta_{\star 1:J}) \right\}.
  \]

  To compute the bias of $\hat{v}^{[1:J]}(\theta_{\star 1:J})$, 
  we write
  \[
    \E\left\{ \hat{v}^{[1:J]}(\theta_{\star 1:J}) \right\} - v
      = \sum_{j=1}^{J} w_j \left[ \E \left\{ \hat{v}^{[j]}(\theta_{\star j})\right\} - v \right]
      + \sum_{j=1}^{J} \E\left\{ (\hat{w}_j - w_j) \hat{v}^{[j]}(\theta_{\star j}) \right\}.
  \]
  By Cauchy--Schwarz inequality, $\E\{ \vert (\hat{w}_j - w_j) \hat{v}^{[j]}(\theta_{\star j}) \vert\} 
  \leq \| \hat{w}_j - w_j \| \| \hat{v}^{[j]}(\theta_{\star j})\| $.
  It is easy to see that $\| \hat{v}^{[j]}(\theta_{\star j}) \|  = O(1)$. 
  Since $\| \hat{w}_j - w_j\| = o\left\{ n^{-p/(2p+1)} \right\}$, we have 
  $\E\{| (\hat{w}_j - w_j) \hat{v}^{[j]}(\theta_{\star j}) |\}  = o\{ n^{-p/(2p+1)} \}$.
  Combining with previous results,
  \begin{align*}
      \E\left\{ \hat{v}^{[1:J]}(\theta_{\star 1:J})  \right\}  - v 
        &= \sum_{j=1}^{J} w_j \left[ \E \left\{ \hat{v}^{[j]}(\theta_{\star j}) \right\}  - v \right]
          + o\left\{ n^{-p/(2p+1)} \right\} \\
        &=   	 
        \sum_{j=1}^{J} w_j B \left\{ \frac{2A}{pB^2 (\xi^{[j]}_{p})^2 n } \right\}^{p/(2p+1)} \xi^{[j]}_{p} v 
          + o\left\{n^{-p/(2p+1)}\right\} \\
        &= \left(\frac{2A}{pn}\right)^{p/(2p+1)} \sum_{j=1}^{J} w_j \left\vert \frac{B}{\xi^{[j]}_{p}}\right\vert^{1/(2p+1)} 
        \sgn(B\xi^{[j]}_{p}) v 
          + o\left\{n^{-p/(2p+1)}\right\}.
  \end{align*}
  Rearranging the terms, we obtain the desired expression. 
  To calculate the variance, by Minkowski inequality, we have 
  \begin{align}\label{apdx:var1}
    \left\| \hat{v}^{[1:J]}(\theta_{\star 1:J}) - \E\left\{ \hat{v}^{[1:J]}(\theta_{\star 1:J}) \right\} \right\| 
      &\lesseqqgtr 
      \left\| \sum_{j=1}^J w_j\hat{v}^{[j]}(\theta_{\star j}) - \E\left\{  \sum_{j=1}^J w_j\hat{v}^{[j]}(\theta_{\star j}) \right\} \right\| \\
      & \quad \pm \left\| \hat{v}^{[1:J]}(\theta_{\star 1:J}) -  \sum_{j=1}^Jw_j\hat{v}^{[j]}(\theta_{\star j}) \right\|  \pm
           \E \left|  \sum_{j=1}^J w_j \hat{v}^{[j]}(\theta_{\star j}) - \hat{v}^{[1:J]}(\theta_{\star 1:J}) \right|. \nonumber
  \end{align}
  Then write
  \[
    \hat{v}^{[1:J]}(\theta_{\star 1:J}) -  \sum_{j=1}^Jw_j\hat{v}^{[j]}(\theta_{\star j})
      = \sum_{j=1}^{J}(\hat{w}_j - w_j) \hat{v}^{[j]}(\theta_{\star j}).
  \]
  The assumed condition implies that $\hat{w}_j - w_j = o_p\left\{ n^{-p/(2p+1)} \right\}$, so we have 
  $
    \sum_{j=1}^{J}(\hat{w}_j - w_j) \hat{v}^{[j]}(\theta_{\star j})  = o_p\left\{ n^{-p/(2p+1)} \right\}.
  $
  Also for $\nu>0$, $\left\| \sum_{j=1}^{J}(\hat{w}_j - w_j) \hat{v}^{[j]}(\theta_{\star j}) \right\|_{2+\nu}
  \leq \sum_{i=1}^{J}2\left\| \hat{v}^{[j]}(\theta_{\star j}) \right\|_{2+\nu} < \infty$. 
  Hence, 
  $$
    \left\| \sum_{j=1}^{J}(\hat{w}_j - w_j) \hat{v}^{[j]}(\theta_{\star j}) \right\| = o\left\{ n^{-p/(2p+1)} \right\}.
  $$
  Also note that
  \[
    \E \left\{ \left| \sum_{j=1}^J w_j \hat{v}^{[j]}(\theta_{\star j}) - \hat{v}^{[1:J]}(\theta_{\star 1:J})  \right| \right\}\leq 
    \left\| \sum_{j=1}^J w_j \hat{v}^{[j]}(\theta_{\star j}) - \hat{v}^{[1:J]}(\theta_{\star 1:J}) \right\| = o\left\{ n^{-p/(2p+1)} \right\}.
  \]
  Then by (\ref{apdx:var1}), $\left\| \hat{v}^{[1:J]}(\theta_{\star 1:J}) - \E\left\{ \hat{v}^{[1:J]}(\theta_{\star 1:J}) \right\} \right\| 
  = \left\| \sum_{j=1}^J w_j\hat{v}^{[j]}(\theta_{\star j}) - \E\left\{  \sum_{j=1}^J w_j\hat{v}^{[j]}(\theta_{\star j}) \right\} \right\| +  o\left\{ n^{-p/(2p+1)} \right\}$.
  Using results from Theorem~\ref{thm:BiasVar},
  \begin{align*}
     \left\| \sum_{j=1}^J w_j\hat{v}^{[j]}(\theta_{\star j}) - \E\left\{  \sum_{j=1}^J w_j\hat{v}^{[j]}(\theta_{\star j}) \right\} \right\| 
     &\leq \sum_{j=1}^J w_j \left\|\hat{v}^{[j]}(\theta_{\star j}) -  \E\left\{\hat{v}^{[j]}(\theta_{\star j}) \right\}  \right\| \\
     &\sim \sqrt{2p} \left(\frac{2A}{pn}\right)^{p/(2p+1)} \sum_{j=1}^{J} w_j \left\vert\frac{B}{\xi^{[j]}_{p}}\right\vert^{1/(2p+1)} v.\\
  \end{align*}
  Therefore,
  \begin{align*}
   \limsup_{n\to\infty} n^{2p/(2p+1)}\Var\left\{\hat{v}^{[1:J]}\left(\theta_{\star 1:J}\right) \right\} 
             &\leq 
             2p 
             \left\{ \left(\frac{2A}{p}\right)^p  |B| \right\}^{2/(2p+1)} \left\{ \sum_{j=1}^{J} w_j {\left\vert {\xi^{[j]}_{p}}\right\vert}^{-1/(2p+1)}\right\}^2 v^2.
  \end{align*}
  Thus, the desired result follows. 
\end{proof}

\begin{proof}[Proof of Theorem~\ref{thm:BiasVarGeneral}]
  We first derive the bias and variance of $\hat{v}(\ell; K; H; \theta_{\star})$.
  First consider the case when $p'\in\mathbb{N}$,
  let $r'(t) = H(t) - \left( 1 + B'{|t|}^{p'} \right) \mathbb{1}(|t|\leq1)$.
  Since $H(t) = H(-t)$ for all $t\geq 0$ and
  $H(0)=1$,
  $\lim_{t \to 0 } |r'(t)| / {|t|}^{p'} = 0$.
  The bias is then computed as 
  \[
    \E\left\{ \hat{v}(\ell; K; H; \theta_{\star})\right\} - v 
      = \left\{ \frac{\sum_{k\in\mathbb{Z}}\gamma_k(\theta_{\star})}{\sum_{k=-n+1}^{n-1} H(k/\ell) \gamma_k(\theta_{\star})} - 1\right\} \E(\tilde{v})
      +\E(\tilde{v}) - v.
  \]
  Then we follow the arguments in the proof of Theorem~\ref{thm:BiasVar} but with kernel $H$ in the postcoloring coefficient.
  Together with the results above, we have 
  \[
    \frac{\sum_{k\in\mathbb{Z}}\gamma_k({\theta_{\star}})}{\sum_{k\in\mathbb{Z}} H(k/\ell) \gamma_k({\theta_{\star}})} - 1 
      = \frac{\sum_{k\in\mathbb{Z}}\gamma_k({\theta_{\star}})}{\sum_{k\in\mathbb{Z}}\gamma_k({\theta_{\star}}) + B'v_{p'}(\theta_{\star})/\ell^{p'} + o(1/\ell^{p'})} - 1 
      = \frac{-B'v_{p'}(\theta_{\star})/\ell^{p'}}{v(\theta_{\star})} + o\left(\frac{1}{\ell^{p'}}\right).
  \]
  Hence, the bias is 
  \begin{align}\label{eqt:tpwBias}
    \E\left\{ \hat{v}(\ell; K; H; \theta_{\star}) \right\} - v 
    &= \left\{ \frac{-B'v_{p'}(\theta_{\star})}{v(\theta_{\star})\ell^{p'}} + o\left(\frac{1}{\ell^{p'}}\right) \right\}  \E(\tilde{v}) +  \E(\tilde{v}) - v \nonumber\\
    &= \left\{ \frac{-B'v_{p'}(\theta_{\star})}{v(\theta_{\star})\ell^{p'}} + o\left(\frac{1}{\ell^{p'}}\right) \right\}\left\{v+o(1)\right\} + B \frac{v_p}{\ell^{p}} + O\left(\frac{\ell}{n}\right) + o\left(\frac{1}{\ell^p}\right) \nonumber\\
      &= \frac{B}{\ell^p}\left\{ v_p - \frac{B' v_{p'}({\theta_{\star}})\ell^{p-p'}}{ Bv({\theta_{\star}})} v \right\} + o\left(\frac{1}{\ell^p}\right) + O\left(\frac{\ell}{n}\right).
  \end{align}
  If $p' = \infty$,
  the coefficient in (\ref{eqt:bvcoef1}) can be written as 
  \[
    \frac{\sum_{k\in\mathbb{Z}}\gamma_k({\theta_{\star}})}{\sum_{k\in\mathbb{Z}} H(k/\ell) \gamma_k({\theta_{\star}})} - 1 
      =  \frac{\sum_{k\in\mathbb{Z}}\gamma_k({\theta_{\star}})}{\sum_{|k| \leq \ell}\gamma_k({\theta_{\star}}) + \sum_{k\in\mathbb{Z}} r'\left(k/\ell\right)\gamma_k({\theta_{\star}}) }  - 1,
  \]
  where $r'(t) = H(t) - \mathbb{1}(|t| \leq 1)$.
  So for $p < \infty$, $\lim_{t\to 0} |r'(t)|/{|t|}^{p} = 0$ and
  \begin{equation}\label{eqt:bvterms2}
     \left\lvert \sum_{k\in\mathbb{Z}} r'\left(\frac{k}{\ell}\right)  \gamma_k({\theta_{\star}})\right\rvert  
      \leq \left\lvert \sum_{|k| \leq \sqrt{\ell}} r'\left(\frac{k}{\ell}\right)  \gamma_k({\theta_{\star}})\right\rvert  + 
          \left\lvert \sum_{|k| > \sqrt{\ell}} r'\left(\frac{k}{\ell}\right)  \gamma_k({\theta_{\star}})\right\rvert.
  \end{equation}
  As $\ell\to \infty$,
  $$
    \sup_{|k| \leq \sqrt{\ell}} \left\mid \frac{r'\left(\frac{k}{\ell} \right)}{\left(\frac{k}{\ell}\right)^{p}}\right\mid  
     = o(1)
     \quad\text{and}\quad
     \left\mid  \sum_{|k| \leq \sqrt{\ell}} r'\left(\frac{k}{\ell}\right)  \gamma_k({\theta_{\star}})\right\mid  
     \leq \left\{ \sup_{|k| \leq \sqrt{\ell}} \left\mid \frac{r'\left(k/\ell \right)}{\left(k/\ell\right)^{p}}\right\mid  \right\} 
     \sum_{|k| \leq \sqrt{\ell}}\left\mid  \frac{k}{\ell}\right\mid ^{p}\left\lvert\gamma_k({\theta_{\star}})\right\rvert= o\left(\frac{1}{\ell^{p}}\right),
  $$
  where $\sum_{|k| \leq \sqrt{\ell}} {|k|}^{p} \left\mid  \gamma_k({\theta_{\star}}) \right\mid  = O(1)$ as $u_p({\theta_{\star}})<\infty$.
  For the second term in (\ref{eqt:bvterms2}),
  \[
    \left\mid \sum_{|k| > \sqrt{\ell}} r'\left(\frac{k}{\ell} \right) \gamma_k({\theta_{\star}})\right\mid  
      = \left\mid \sum_{|k| > \sqrt{\ell}} \frac{r'\left(k/\ell\right)}{\left\mid  k/\ell \right\mid^{p}}\left\mid  \frac{k}{\ell}\right\mid ^{p} \gamma_k({\theta_{\star}})\right\mid 
      \leq \left\{ \sup_{|k| > \sqrt{\ell}} \left\mid \frac{r'\left( k/\ell \right)}{\left( k/\ell \right)^{p}}\right\mid  \right\} \cdot
      \sum_{|k| > \sqrt{\ell}} \left\mid  \frac{k}{\ell}\right\mid ^{p}\left\mid \gamma_k({\theta_{\star}})\right\mid .
  \]
  Note that 
  \begin{align*}
    \sup_{|k| > \sqrt{\ell}} \left\mid \frac{r''\left( k/\ell \right)}{\left( k/\ell \right)^{p}}\right\mid 
    &\leq 
    \sup_{t>0} \frac{|r(t)|}{{|t|}^{p}} 
    = \sup_{t>0} \frac{|K(t) - \mathbb{1}(|t|\leq 1)|}{{|t|}^{p}} \\
    & \leq \sup_{0<t<1} \frac{|K(t)-1|}{{|t|}^{p}} 
    + \sup_{t\geq 1} \frac{|K(t)|}{{|t|}^{p}}
    + 1 = O(1)
  \end{align*}
  as $|K(t)| \leq 1$ and $\lim_{t\downarrow 0} |K(t)-K(0)| / {|t|}^p < \infty$.
  In addition,
  $\sum_{|k| > \sqrt{\ell}} {|k/\ell|}^{p}\left\lvert \gamma_k({\theta_{\star}})\right\rvert  = o(1/\ell^{p})$
  following previous arguments.
  Hence, 
  $$
    \left\mid \sum_{|k| > \sqrt{\ell}} r'\left(\frac{k}{\ell}\right) \gamma_k({\theta_{\star}})\right\mid  = o\left(\frac{1}{\ell^{p}}\right)
      \qquad\text{and}\qquad
    \left\mid  \sum_{|k| \leq \sqrt{\ell}} r'\left(\frac{k}{\ell}\right)  \gamma_k({\theta_{\star}})\right\mid  = o\left(\frac{1}{\ell^{p}}\right).
  $$
  Therefore, 
  \[
    \frac{\sum_{k\in\mathbb{Z}}\gamma_k({\theta_{\star}})}{\sum_{k\in\mathbb{Z}} H(k/\ell) \gamma_k({\theta_{\star}})} - 1 = o\left(\frac{1}{\ell^{p'}}\right)
    \qquad\text{and}\qquad
    \E\{\hat{v}(\ell; K; H; \theta_{\star})\} - v \sim \E(\tilde{v}) - v.
  \]
  In both cases, as $n,\ell\to\infty$, $\Var\{\hat{v}(\ell; K; H; \theta_{\star})\} \sim \Var\{ \tilde{v}(\ell; K) \}$.

We can now conclude that when $\ell\asymp n^{1/(2p+1)}$, 
  $n^{p/(2p+1)} \left\{ \hat{v}(\ell; K; H; \theta_{\star}) - v\right\} = O_p(1)$.
  We also need to show that $n^{p/(2p+1)} \left\{ \hat{v}(\ell; K; H; \bar{\theta}) - \hat{v}(\ell; K; H; \theta_{\star}) \right\} = o_p(1)$.
  By Taylor's expansion, for some $\ddot{\theta}$ between $\bar{\theta}$ and $\theta_{\star}$, 
  \begin{align*}
    n^{p/(2p+1)} \left\{ \hat{v}(\ell; K; H; \bar{\theta}) - \hat{v}(\ell; K; H; \theta_{\star}) \right\}
     &= n^{p/(2p+1)} \left.\frac{\partial \hat{v}(\ell; K; H;\theta) }{\partial \theta} \right\mid_{\theta = \ddot{\theta}} (\bar{\theta} - \theta_{\star}) , 
\end{align*}
where
\begin{align*}	
     \left. \frac{\partial \hat{v}(\ell; K; H;\theta)}{\partial \theta} \right\mid_{\theta = \ddot{\theta}}
     &= \left\{ M_{\ell, H}(\ddot{\theta}) \frac{\partial v(\ddot{\theta})}{\partial \theta} 
     - v(\ddot{\theta}) \frac{\partial M_{\ell, H}(\ddot{\theta})}{\partial \theta}  \right\} \tilde{v}/ 
     M_{\ell, H}^2(\ddot{\theta}).
  \end{align*}
  By assumption, $\sup_{\theta\in\mathcal{H}} \left|{\partial v({\theta})}/{\partial \theta}\right|<\infty$ and $\sup_{\theta\in\mathcal{H}} \left|{\partial M_{\ell, H}({\theta})}/{\partial \theta}\right|<\infty$. 
  So ${\partial \hat{v}(\ell; K; H;\theta)}/{\partial \theta} \mid_{\theta = \ddot{\theta}} = O_p(1)$
  and $n^{p/(2p+1)} \left\{ \hat{v}(\ell; K; H; \bar{\theta}) - \hat{v}(\ell; K; H; \theta_{\star}) \right\} = o_p(1)$
  as $n^{1/2} (\bar{\theta} - \theta_{\star}) = O_p(1)$.
  Then we define a general version of $\epsilon^{\star}_h$ for $h\geq 0$ such that
  \[
    \epsilon^{\star}_h
    = \min\left[ n^{\frac{2p}{2p+1}} \left\{ \hat{v}(\ell; K; H;\bar{\theta}) - v \right\}^2, h \right] 
    - \min\left[ n^{\frac{2p}{2p+1}} \left\{ \hat{v}(\ell; K; H;\theta_{\star}) - v \right\}^2, h \right] = o_p(1).
  \]
  Since $\epsilon^{\star}_h$ is upper-bounded by $h$,
  by Lemma A2 of \cite{andrews1991},
  we have $\E(\epsilon^{\star}_h) = o(1)$ and thus
  \[
    \lim_{h\to\infty} \lim_{n\to\infty} \MSE_h\{ \hat{v}(\ell; K; H;\bar{\theta})\}
      = \lim_{h\to\infty} \lim_{n\to\infty} \MSE_h\{\hat{v}(\ell; K; H;\theta_{\star})\}.
  \]
  Since $\lim_{h\to\infty} \lim_{n\to\infty} \MSE_h\{\hat{v}(\ell; K; H;\theta_{\star})\} = \lim_{n\to\infty} n^{2p/(2p+1)} \MSE\{\hat{v}(\ell; K; H;\theta_{\star})\}$,
  $$\lim_{h\to\infty} \lim_{n\to\infty} \MSE_h\{ \hat{v}(\ell; K; H;\bar{\theta})\}
        = \lim_{n\to\infty} n^{2p/(2p+1)} \MSE\{\hat{v}(\ell; K; H;\theta_{\star})\}.$$
   Thus, we obtained the desired result.  
\end{proof}

\begin{proof}[Proof of Corollary~\ref{cor:TailBVgeneral}]
  The conclusions can be drawn trivially following Theorem~\ref{thm:BiasVarGeneral}.
\end{proof}

\begin{proof}[Proof of Proposition~\ref{prop:multivariateBiasVar}]
  {
  (i) We first show that}
  \begin{align}\label{eqt:WMSE_var_star}
    \lim_{h\to\infty} \lim_{n\to\infty}  \WMSE_h\left\{ \hat{V}(\bar{\theta}); W\right\}
    = \lim_{h\to\infty} \lim_{n\to\infty} \WMSE_h \left\{ \hat{V}(\theta_{\star}); W\right\}.
  \end{align}
  By Taylor's expansion,
  for some $\ddot{\theta} \in \mathcal{H}$ on the line segment joining $\bar{\theta}$ and $\theta_{\star}$,
  \begin{align}\label{eqt:deriv}
    n^{{p}/(2p+1)} \left\{ \hat{V}^{(u,v)}(\bar{\theta}) -  \hat{V}^{(u,v)}(\theta_{\star})\right\} 
      = n^{{p}/(2p+1)} 
        \left\{ \nabla \hat{V}^{(u,v)}(\ddot{\theta}) 
    \right\}^{\T}
    (\bar{\theta}-\theta_{\star}),
  \end{align}
  where $\nabla \hat{V}^{(u,v)}(\ddot{\theta})$ is the gradient of 
  $\hat{V}^{(u,v)}(\theta)$ at $\theta = \ddot{\theta}$.
  After doing matrix multiplication, we have 
  \[
    \hat{V}^{(u,v)}(\theta)
    = \frac{1}{2} \sum_{j=1}^d \sum_{k=1}^d \left[
      V^{(u,k)}(\theta) \{ M^{-1}(\theta) \}^{(k,j)} \tilde{V}^{(j,v)}
      + \tilde{V}^{(u,j)} \{ M^{-1}(\theta) \}^{(j,k)} V^{(k,v)}(\theta)
    \right],
  \]
  where $\{ M^{-1}(\theta) \}^{(k,j)}$ is the $(k,j)$th element of $M^{-1}(\theta)$
  and other variables are defined similarly.  
  Hence, the gradient of $\hat{V}^{(u,v)}(\theta)$ can be found as follows: 
  \begin{align*}  
      \nabla\hat{V}^{(u,v)}(\theta)
    &= \frac{1}{2} \sum_{j=1}^d \sum_{k=1}^d \bigg(
       \left[\{\nabla V^{(u,k)}(\theta)\} \{ M^{-1}(\theta) \}^{(k,j)} 
      + V^{(u,k)}(\theta) \nabla\{ M^{-1}(\theta) \}^{(k,j)} \right]\tilde{V}^{(j,v)} \\
    &\qquad
      + \tilde{V}^{(u,j)} \left[\nabla\{ M^{-1}(\theta) \}^{(j,k)} V^{(k,v)}(\theta)
      +\{ M^{-1}(\theta) \}^{(j,k)} \{\nabla V^{(k,v)}(\theta)\} \right] \bigg).
  \end{align*}
   Note that 
  \[
    \frac{\partial M^{-1}(\theta)}{\partial  \theta^{(h)}} = - M^{-1}(\theta)  \left\{ \frac{\partial M(\theta)}{\partial  \theta^{(h)}}\right\} M^{-1}(\theta).
  \]
  Using the assumptions that $M(\theta)$ is invertible and the conditions that 
  \[
    \sup_{\theta\in\mathcal{H}} \left| \frac{\partial V^{(u,v)}(\theta)}{\partial  \theta^{(h)}} \right| <\infty
      \qquad\text{and}\qquad
    \sup_{\theta\in\mathcal{H}} \left| \frac{\partial M^{(u,v)}(\theta)}{\partial  \theta^{(h)}} \right| < \infty
    \qquad (u,v,h\in\{1, \ldots, d\}),
  \]
  we have, for all $h=1,\ldots,d$, that  
   \begin{align}
    &\sup_{\theta\in\mathcal{H}} \left| \frac{\hat{V}^{(u,v)}(\theta)}{\partial \theta^{(h)}} \right| \nonumber\\
      &\quad\leq \frac{1}{2} \sum_{j=1}^d \sum_{k=1}^d \sup_{\theta\in\mathcal{H}}\left| 
      \left[\frac{\partial V^{(u,k)}(\theta)}{\partial \theta^{(h)}} \{ M^{-1}(\theta) \}^{(k,j)} 
      + V^{(u,k)}(\theta) \frac{\partial \{ M^{-1}(\theta) \}^{(k,j)}}{\partial \theta^{(h)}} \right]\tilde{V}^{(j,v)} \right.\nonumber\\
    &\quad\qquad
      + \left.\tilde{V}^{(u,j)} \left[\frac{\partial \{ M^{-1}(\theta) \}^{(j,k)}}{\partial \theta^{(h)}} V^{(k,v)}(\theta)
      + \{ M^{-1}(\theta) \}^{(j,k)} \frac{\partial V^{(k,v)}(\theta)}{\partial \theta^{(h)}} \right]\right|\nonumber\\
    &\quad\leq \frac{1}{2} \sum_{j=1}^d \sum_{k=1}^d \bigg( 
      \left[\sup_{\theta\in\mathcal{H}}\left|\frac{\partial V^{(u,k)}(\theta)}{\partial \theta^{(h)}}\right| \left|\{ M^{-1}(\theta) \}^{(k,j)}\right| \right.
      + \left.\left| V^{(u,k)}(\theta) \right| \sup_{\theta\in\mathcal{H}}\left|\frac{\partial \{ M^{-1}(\theta) \}^{(k,j)}}{\partial \theta^{(h)}}\right| \right]\left|\tilde{V}^{(j,v)}\right| \nonumber\\
    &\quad\qquad
      + \left|\tilde{V}^{(u,j)}\right| \left[\sup_{\theta\in\mathcal{H}}\left|\frac{\partial \{ M^{-1}(\theta) \}^{(j,k)}}{\partial \theta^{(h)}}\right| \left|V^{(k,v)}(\theta)\right|\right.
      + \left.\left| \{ M^{-1}(\theta) \}^{(j,k)}\right| \sup_{\theta\in\mathcal{H}}\left|\frac{\partial V^{(k,v)}(\theta)}{\partial \theta^{(h)}}\right| \right] \bigg)\nonumber\\
      &\quad = O_p(1), \label{eqt:bounded_derivativeVhat}
   \end{align}
   where the last line follows from the property that $\sum_{j=1}^d \sum_{k=1}^d |\tilde{V}^{(j,k)}| = O_p(1)$.
  Therefore, in view of (\ref{eqt:bounded_derivativeVhat})
  and the assumption that $\bar{\theta}-\theta_{\star} = O_p(n^{-1/2})$, we can bound (\ref{eqt:deriv}) as  
  $
  	n^{{p}/(2p+1)} \left\{ \hat{V}^{(u,v)}(\bar{\theta}) -  \hat{V}^{(u,v)}(\theta_{\star})\right\} = o_p(1).
  $  
  Then for $h\geq 0$, we re-define  
  \begin{align*}
    \epsilon_h^* 
      &= \min\left[ n^{\frac{2p}{2p+1}}\vect\left\{ \hat{V}(\bar{\theta}) - V\right\}^\T W \vect\left\{ \hat{V}(\bar{\theta}) - V\right\}, h\right] \\
                 & \quad - \min\left[ n^{\frac{2p}{2p+1}}\vect\left\{ \hat{V}(\theta_{\star}) - V\right\}^\T W \vect\left\{ \hat{V}(\theta_{\star}) - V\right\}, h\right] 
                 = o_p(1).
  \end{align*}
  Since $\epsilon_h^*$ is upper-bounded by $h$, we have by Lemma A2 of \cite{andrews1991} that $\E(\epsilon_h^*) = o(1)$, 
  which implies that (\ref{eqt:WMSE_var_star}) is true.
  Note that 
  \begin{align*}
    \lim_{h\to\infty} \lim_{n\to\infty} \WMSE_h \left\{ \hat{V}(\theta_{\star}); W\right\} 
    = \lim_{n\to\infty} n^{2p/(2p+1)}\WMSE \left\{ \hat{V}(\theta_{\star}); W\right\}.
  \end{align*}
  Hence, the result in part (i) follows. 

	(iii) Next, we derive the bias of $\hat{V}(\theta_{\star})$. 
  By similar arguments in Theorem~\ref{thm:BiasVar} but with $\gamma_k$ replaced by the autocovariance matrix $\Gamma_k$, 
  we have
  \begin{align}\label{eqt:biasVtil}
    \E(\tilde{V}) - V = \frac{BV_p}{\ell^p} + o\left(\frac{1}{\ell^p}\right) + O\left(\frac{\ell}{n}\right)
  \end{align}
  where we have used the property that  $\E(\tilde\Gamma_k) = \Gamma_k + O(1/n)$ 
  for $k=1,\ldots, n$.
  Similarly, we also have 
  \begin{align}\label{eqt:M_formula_matrix}
    M(\theta_{\star}) - V(\theta_{\star}) = \frac{BV_p(\theta_{\star})}{\ell^p} + o\left(\frac{1}{\ell^p}\right).
  \end{align}
  Recall that by the definition of the proposed estimator, we have
  \begin{align}\label{eqt:estDecomp}
    \hat{V}(\theta_{\star}) 
    	&= \frac{V(\theta_{\star}) M^{-1}(\theta_{\star})\tilde{V}+\tilde{V} M^{-1}(\theta_{\star}) V(\theta_{\star})}{2} \nonumber \\
    &= \tilde{V} + \frac{\left\{V(\theta_{\star}) M^{-1}(\theta_{\star})-I_d\right\}\tilde{V}+
    	\tilde{V} \left\{ M^{-1}(\theta_{\star}) V(\theta_{\star})-I_d\right\}}{2} ,
  \end{align}
  where $I_d$ is a $d\times d$ identity matrix.
  By Woodbury matrix identity and (\ref{eqt:M_formula_matrix}), we have 
  \begin{align*}
    M^{-1}(\theta_{\star}) 
    &= \left\{ V(\theta_{\star}) + \frac{BV_p(\theta_{\star})}{\ell^p} + o\left(\frac{1}{\ell^p}\right)\right\}^{-1} \\
    &= V^{-1}(\theta_{\star}) - V^{-1}(\theta_{\star}) \left\{ \frac{\ell^p V(\theta_{\star})V_p^{-1}(\theta_{\star})}{B} + I_d \right\}^{-1} + o\left(\frac{1}{\ell^p}\right).
  \end{align*}
  Hence, when $\ell\to\infty$,
  \begin{align}
    V(\theta_{\star})M^{-1}(\theta_{\star}) - I_d
    &=  - \left\{ \frac{\ell^p V(\theta_{\star})V_p^{-1}(\theta_{\star})}{B} + I_d \right\}^{-1} + o\left(\frac{1}{\ell^p}\right) \nonumber \\
    &=
    -\frac{B}{\ell^p}V_p(\theta_{\star})V^{-1}(\theta_{\star}) + o\left(\frac{1}{\ell^p}\right). \label{eqt:VMinverse}
  \end{align}
  Taking expectation on both sides of (\ref{eqt:estDecomp}), and applying (\ref{eqt:biasVtil}) and (\ref{eqt:VMinverse}), 
  we obtain the expression of the bias as follows:
	{
  \begin{align*}
    \E\left\{ \hat{V}(\theta_{\star}) \right\} - V 
    &= \E(\tilde{V}) - V  + \frac{\left\{V(\theta_{\star}) M^{-1}(\theta_{\star})-I_d\right\}\E(\tilde{V})+
    	\E(\tilde{V}) \left\{ M^{-1}(\theta_{\star}) V(\theta_{\star})-I_d\right\}}{2}\\
    &=  \frac{B }{\ell^p} V_p
    	- \frac{B}{2\ell^p}V_p(\theta_{\star})V^{-1}(\theta_{\star})V 
		- \frac{B}{2\ell^p}V V^{-1}(\theta_{\star}) V_p(\theta_{\star})
			+ o\left( \frac{1}{\ell^p}\right) + O\left( \frac{\ell}{n}\right) \\
    &=  \frac{B}{\ell^p}\frac{1}{2} 
    	\left[
		\left\{ V_p - V_p(\theta_{\star})V^{-1}(\theta_{\star})V\right\}
		+\left\{ V_p - VV^{-1}(\theta_{\star})V_p(\theta_{\star})\right\}
		\right]
		+ o\left( \frac{1}{\ell^p}\right) + O\left( \frac{\ell}{n}\right) \\
    &=  \frac{B}{\ell^p}\frac{1}{2} 
    	\left[
		\left\{ V_pV^{-1} - V_p(\theta_{\star})V^{-1}(\theta_{\star})\right\}V
		+V\left\{ V^{-1}V_p - V^{-1}(\theta_{\star})V_p(\theta_{\star})\right\}
		\right]
		+ o\left( \frac{1}{\ell^p}\right) + O\left( \frac{\ell}{n}\right) , 
  \end{align*}
  }
  which is our desired result for part (iii). 

  {
  (ii) 
  It remains to derive the variance of $\hat{V}(\theta_{\star})$. 
  In view of (\ref{eqt:estDecomp}) and (\ref{eqt:VMinverse}),
  we know that 
  $\Var\left\{\vect\hat{V}(\theta_{\star})\right\} \sim \Var \left( \vect\tilde{V}\right)$ for all $u,v\in \{1,\ldots, d\}$.
  To derive $\Var \left(\vect\tilde{V} \right)$, 
  it suffices for us to derive $\Cov\left\{  \tilde{V}^{(u,v)} , \tilde{V}^{(u',v')} \right\}$
  for each $u,v,u',v'\in\{1, \ldots, d\}$.
  The derivation of $\Cov\left\{  \tilde{V}^{(u,v)} , \tilde{V}^{(u',v')} \right\}$
  is similar to the derivation of $\Var(\tilde{v})$ in Theorem~\ref{thm:BiasVar}, 
  except that we need to use 
  Lemmas~\ref{lemma:multblock} and \ref{lemma:multblock2} 
  instead of Lemmas~\ref{lemma:hoc1} and \ref{lemma:hoc2}. 
  Following the steps, we obtain 
  \[
  	\Cov\left\{  \tilde{V}^{(u,v)} , \tilde{V}^{(u',v')} \right\}
		= 2A\left(V^{(u,u')}V^{(v,v')} + V^{(u,v')}V^{(u',v)} \right) \frac{\ell}{n} + o\left(\frac{\ell}{n}\right) . 
  \]
  Upon careful checking the matrix form, we obtain 
    \begin{equation*}
      \Var\left( \vect  \tilde{V} \right)
      = 2A (I_{dd} + C_{dd}) (V\otimes V) \frac{\ell}{n} + o\left(\frac{\ell}{n}\right) . 
    \end{equation*}  
  Hence, the proof for part (ii) is completed. }
\end{proof}

\begin{proof}[Proof of Proposition~\ref{prop:hacRobust}]

We will reuse some  parts of the analysis in the proof of Proposition~\ref{prop:multivariateBiasVar}.
(i) For the bias, 
we recall (\ref{eqt:estDecomp}): 
\begin{equation*}
  \hat{V}(\theta_{\star}) 
      = \tilde{V} + \frac{\left\{V(\theta_{\star}) M^{-1}(\theta_{\star})-I_d\right\}\tilde{V}+\tilde{V} \left\{ M^{-1}(\theta_{\star}) V(\theta_{\star})-I_d\right\}}{2}.
\end{equation*}
We can still use (\ref{eqt:VMinverse}) and the bias is 
\begin{align}\label{eqt:bias_Vhat_multi}
&\E\left\{ \hat{V}(\theta_{\star}) \right\} - V_n \nonumber\\
    &\qquad=  \E(\tilde{V}) - V_n
    	- \frac{B}{2\ell^p}V_p(\theta_{\star})V^{-1}(\theta_{\star})\E(\tilde{V}) 
		- \frac{B}{2\ell^p}\E(\tilde{V}) V_p(\theta_{\star})V^{-1}(\theta_{\star}) 
			+ o\left( \frac{1}{\ell^p}\right) + O\left( \frac{\ell}{n}\right).
\end{align}
So it remains to study $\E(\tilde{V})$.
Define  $\Gamma_{k, \mathdag}$ such that 
\[
  \Gamma^{(u,v)}_{k, \mathdag} = \Gamma^{(v,u)}_{-k, \mathdag} 
  = \sup_{i\geq1}\E\left[\left\{X^{(u)}_{i+k}-\mu^{(u)}\right\}\left\{X^{(v)}_i-\mu^{(v)}\right\}\right]
  \qquad\text{for }k = 0,1,\ldots,n-1.
\]
The finite-$n$ counterpart of $V^{(u,v)}_p$ is defined as $V^{(u,v)}_{p,n} = \sum_{|k| <n} |k| ^{p} \Gamma^{(u,v)}_{k,n}$,
where
the finite-$n$ version of the autocovariance matrix is
\[
 {\Gamma}_{k, n} = {\Gamma}^{\T}_{-k, n} = \frac{1}{n} \sum_{i=k+1}^{n} \E\{(X_i-\mu)(X_{i-k}-\mu)^\T\},\qquad\text{for }k = 0,1,\ldots,n-1.
\]
It is easy to see that $\E(\tilde{\Gamma}_{k}) = {\Gamma}_{k, n} + O(1/n)$
still holds from a natural extension of the arguments in Theorem~\ref{thm:BiasVar}.
Therefore, let $r(t) = K(t) - (1+B{|t|}^p)\mathbb{1}(|t|\leq 1)$,
when $n,\ell\to\infty$,
\begin{align*}
  \E(\tilde{V}) - V_n 
  &= \sum_{|k|\leq \ell} \left( 1 + B{\left|\frac{k}{\ell} \right|}^p\right)\E(\tilde{\Gamma}_{k})
  + \sum_{k=-n+1}^{n-1} r\left(\frac{k}{\ell}\right)\E(\tilde{\Gamma}_{k}) - \sum_{k=-n+1}^{n-1}{\Gamma}_{k, n}, \\
  \sum_{|k|\leq \ell} \left\{ \E(\tilde{\Gamma}_{k}) -  \Gamma_{k, n}\right\}
  &= O\left(\frac{\ell}{n}\right), \\
  \sum_{k=-n+1}^{n-1} r\left(\frac{k}{\ell}\right)\E(\tilde{\Gamma}_{k}) &= o\left(\frac{1}{\ell^p}\right),\\
  B\sum_{|k|\leq \ell}{\left|\frac{k}{\ell} \right|}^p \E(\tilde{\Gamma}_{k}) 
  &= \frac{BV_{p,n}}{\ell^p} + o\left(\frac{1}{\ell^p}\right),
\end{align*}
where
we have used similar arguments as in the proof of Theorem~\ref{thm:BiasVar} and 
the assumption that $\sum_{u=1}^{d}\sum_{v=1}^{d}U^{(u,v)}_{p, \mathdag}<\infty$ in the last equality.
Hence,
we have $\E(\tilde{V}) - V_n = BV_{p,n}/\ell^p + o(1/\ell^p) + O(\ell/n)$,
where $V_{p,n}$ satisfies that $|V_{p,n}^{(u,v)}|\leq |V^{(u,v)}_{p, \mathdag}| \leq U^{(u,v)}_{p, \mathdag} <\infty$ for all $u,v\in\{1, \ldots, d\}$.
Using this result, (\ref{eqt:bias_Vhat_multi}) can be simplified as follows:
  \begin{align*}
\E\{\hat{V}(\theta_{\star})\} - V_n
	= \frac{B}{\ell^p} V_{p,n}
    	- \frac{B}{2\ell^p}V_p(\theta_{\star})V^{-1}(\theta_{\star})V_n
		- \frac{B}{2\ell^p}V_n V_p(\theta_{\star})V^{-1}(\theta_{\star}) 
			+ o\left( \frac{1}{\ell^p}\right) + O\left( \frac{\ell}{n}\right).
  \end{align*}
  Together with the fact that $0 \leq V_n^{(u,v)} \leq V_{0, \mathdag}^{(u,v)}$ for all $u,v\in\{1, \ldots, d\}$, 
  we have 
  \[
    \limsup_{n\to\infty} \ell^p\left\vert \E\{\hat{V}^{(u,v)}(\theta_{\star})\}- V^{(u,v)}_n\right\vert
    \leq |B| \left\{U^{(u,v)}_{p,\mathdag} + \Psi^{(u,v)}_{p,\mathdag}\right\} ,
  \]
where $\Psi_{p,\mathdag} = \left\vert V_p(\theta_{\star})V^{-1}(\theta_{\star}) V_{0,\mathdag} 
  			+ V_{0,\mathdag} V^{-1}(\theta_{\star})V_p(\theta_{\star}) \right\vert /2$. 
For the variance, we can easily modify the arguments in Proposition~\ref{prop:multivariateBiasVar} to obtain the desired result
after noting that $0\leq V_{n}^{(u,v)} \leq V_{0,\mathdag}^{(u,v)}$ for all $u,v \in\{1,\ldots,d\}$.

(ii) Following the same arguments as in Proposition~\ref{prop:multivariateBiasVar}, 
we obtain $n^{1/(2p+1)} \left\{ \hat{V}(\bar{\theta}) - \hat{V}({\theta}^{\star}) \right\} = o_p(1)$. 
Thus, it completes the proof of the proposition. 
\end{proof}

{

\begin{proof}[Proof of Theorem~\ref{thm:spectralBV}]
The proof follows similar arguments as in the proof of Theorem~\ref{thm:BiasVar}. 
To facilitate understanding, we also present the details of the proof below, with the key differences highlighted.

  First, we derive the bias 
  $\hat{f}(\omega; \ell; K; \theta_{\star} )$
  and write it in short as $\hat{f}(\omega; \theta_{\star}) $.
  Let $r(t) = K(t) - \left( 1 + B{|t|}^{p} \right) \mathbb{1}(|t|\leq1)$.
  Since $K(t) = K(-t)$ for all $t\geq 0$ and
  $K(0)=1$,
  $\lim_{t \to 0 } |r(t)| / {|t|}^{p} = 0$.
  The bias is computed following the definition as 
  \[
    \E\left\{ \hat{f}(\omega;\theta_{\star}) \right\} - f(\omega)
      = \left\{ \frac{\sum_{k\in\mathbb{Z}}\gamma_k(\theta_{\star})\exp(\iota k\omega)}{\sum_{k\in\mathbb{Z}} K(k/\ell) \gamma_k(\theta_{\star})\exp(\iota k\omega)} - 1\right\} \E\left\{
      \tilde{f}(\omega)\right\}
      +\E\left\{ \tilde{f}(\omega)\right\} - f(\omega).
  \]
  In particular, consider the coefficient
  \begin{align*}\label{eqt:bvcoef1_spec}
     &\frac{\sum_{k\in\mathbb{Z}}\gamma_k(\theta_{\star})\exp(\iota k\omega)}{\sum_{k\in\mathbb{Z}} K(k/\ell) \gamma_k(\theta_{\star})\exp(\iota k\omega)} - 1 \\
      &\qquad= \frac{\sum_{k\in\mathbb{Z}}\gamma_k(\theta_{\star})\exp(\iota k\omega)}{\left\{\sum_{|k| \leq \ell} \left(1 + B{\left\lvert k/\ell \right\vert}^{p}\right)\gamma_k(\theta_{\star}) + \sum_{k\in\mathbb{Z}} r\left( k/\ell \right)\gamma_k(\theta_{\star})\right\}\exp(\iota k\omega)} - 1.
  \end{align*}
  Then we compute the terms separately,
  \begin{align*}
    \sum_{|k| \leq \ell} \left(1 + B{\left\lvert \frac{k}{\ell}\right\vert}^{p}\right)\gamma_k({\theta_{\star}})\exp(\iota k\omega)
      &= \sum_{k\in\mathbb{Z}} \gamma_k({\theta_{\star}})\exp(\iota k\omega) -  \sum_{|k| > \ell} \gamma_k({\theta_{\star}})\exp(\iota k\omega) \\
      &\quad + B\sum_{|k| \leq \ell}{\left\lvert \frac{k}{\ell}\right\vert}^{p}\gamma_k({\theta_{\star}})\exp(\iota k\omega);\\
   B\sum_{|k| \leq \ell} {\left\lvert \frac{k}{\ell}\right\rvert}^{p}\gamma_k({\theta_{\star}})\exp(\iota k\omega)
   & = B \sum_{k \in \mathbb{Z}} {\left\lvert \frac{k}{\ell}\right\rvert}^{p}\gamma_k({\theta_{\star}})\exp(\iota k\omega) 
      - B \sum_{|k| > \ell} {\left\lvert \frac{k}{\ell}\right\rvert}^{p}\gamma_k({\theta_{\star}})\exp(\iota k\omega) \\
    & = B\frac{2\pi f_p(\omega;\theta_{\star})}{\ell^p} + o\left(\frac{1}{\ell^p}\right),
  \end{align*}
  where in the last step, we use the fact that 
  $u_{p}(\theta_{\star})<\infty$ and 
  by the triangular inequality and Kronecker's lemma, 
  \[
  \left|\sum_{|k|>\ell} {\left\lvert \frac{k}{\ell}\right\rvert}^{p}\gamma_k({\theta_{\star}})\exp(\iota k\omega) \right| 
  \leq 
  \sum_{|k|>\ell}  {\left\lvert \frac{k}{\ell}\right\rvert}^{p}
  |\gamma_k({\theta_{\star}})||\exp(\iota k\omega)| 
  \leq \sum_{|k|>\ell}  {\left\lvert \frac{k}{\ell}\right\rvert}^{p}
  |\gamma_k({\theta_{\star}})|
  = o\left(\frac{1}{\ell^p}\right).
  \]
  Similarly, 
  \[
    \left| \sum_{|k| > \ell} \gamma_k({\theta_{\star}}) \exp(\iota k\omega)\right|  
    \leq \sum_{|k| > \ell} \left| \gamma_k({\theta_{\star}}) \right|
    \leq \sum_{|k| > \ell} {\left\lvert \frac{k}{\ell} \right\rvert}^p\left|\gamma_k({\theta_{\star}})\right|  = o\left(\frac{1}{\ell^p}\right).
  \]
  Therefore,
  \[
     \sum_{|k| \leq \ell} \left(1 + B{\left\lvert \frac{k}{\ell}\right\vert}^{p}\right)\gamma_k({\theta_{\star}})
     \exp(\iota k\omega) 
     =  \sum_{k\in\mathbb{Z}} \gamma_k({\theta_{\star}})\exp(\iota k\omega)  + B \frac{2\pi f_p(\omega;\theta_{\star})}{\ell^p} + o\left(\frac{1}{\ell^p}\right).
  \]
  Now consider $\sum_{k\in\mathbb{Z}} r({k}/{\ell})\gamma_k({\theta_{\star}})\exp(\iota k\omega)$,
  \begin{equation}\label{eqt:bvterms1_spec}
     \left\lvert \sum_{k\in\mathbb{Z}} r\left(\frac{k}{\ell}\right)  \gamma_k({\theta_{\star}})\exp(\iota k\omega) \right\rvert  
      \leq \left\lvert \sum_{|k| \leq \sqrt{\ell}} r\left(\frac{k}{\ell}\right)  \gamma_k({\theta_{\star}})\exp(\iota k\omega) \right\rvert  + 
          \left\lvert \sum_{|k| > \sqrt{\ell}} r\left(\frac{k}{\ell}\right)  \gamma_k({\theta_{\star}})\exp(\iota k\omega) \right\rvert.
  \end{equation}
  Similar to the proof of Theorem~\ref{thm:BiasVar} with slight modifications,
  we have 
  \[
  \left|\sum_{k\in\mathbb{Z}}r\left(\frac{k}{\ell}\right)\gamma_k(\theta_{\star})\exp(\iota k\omega)\right|
  = o\left(\frac{1}{\ell^{p}}\right).
  \]
  Combining the results above,
  \begin{align*}
    \frac{\sum_{k\in\mathbb{Z}}\gamma_k({\theta_{\star}})\exp(\iota k\omega)}{\sum_{k\in\mathbb{Z}} K(k/\ell) \gamma_k({\theta_{\star}})\exp(\iota k\omega)} - 1 
      &= \frac{\sum_{k\in\mathbb{Z}}\gamma_k({\theta_{\star}})\exp(\iota k\omega)}{\sum_{k\in\mathbb{Z}}\gamma_k({\theta_{\star}})\exp(\iota k\omega) + 2\pi Bf_p(\omega;\theta_{\star}) /\ell^p + o(1/\ell^p)} - 1 \\
      &= \frac{-Bf_p(\omega;\theta_{\star}) }{f(\omega;\theta_{\star}) \ell^p} + o\left(\frac{1}{\ell^p}\right).
  \end{align*}
  Following the arguments in Theorem~\ref{thm:BiasVar},  
  the bias of $\tilde{f}(\omega)$ can be found as 
  \[
    \E\{\tilde{f}(\omega)\} - f(\omega)
    = \frac{B f_p(\omega)}{\ell^p} + o\left(\frac{1}{\ell^p}\right)
    + O\left(\frac{\ell}{n}\right).
  \]
  Therefore, the bias of $\hat{f}(\omega;\theta_{\star})$ is 
  \begin{align*}
    \E\left\{ \hat{f}(\omega;\theta_{\star})  \right\} - f(\omega)
    &= \left\{ \frac{-Bf_p(\omega;\theta_{\star}) }{f(\omega;\theta_{\star}) \ell^p} + o\left(\frac{1}{\ell^p}\right) \right\}  \E\{\tilde{f}(\omega)\} +  \E\{\tilde{f}(\omega)\} - f(\omega) \\
    &= \left\{ \frac{-Bf_p(\omega;\theta_{\star}) }{f(\omega;\theta_{\star}) \ell^p} + o\left(\frac{1}{\ell^p}\right) \right\}\left\{f(\omega)+o(1)\right\} + B \cdot \frac{f_p(\omega)}{\ell^{p}} + O\left(\frac{\ell}{n}\right) + o\left(\frac{1}{\ell^p}\right) \\
    &= \frac{B}{\ell^p}\left\{ \frac{f_p(\omega)}{f(\omega)} - \frac{f_p(\omega;\theta_{\star}) }{f(\omega;\theta_{\star}) }  \right\}f(\omega) + o\left(\frac{1}{\ell^p}\right) + O\left(\frac{\ell}{n}\right).
  \end{align*}

Second, we derive the variance of $\hat{f}(\omega; \theta_{\star} )$. 
  As $\|\hat{f}(\omega;\theta_{\star})  - \tilde{f}(\omega)\| = O(1/\ell^p)$,
  we know that $\Var\{ \hat{f}(\omega; \theta_{\star}) \} \sim \Var\{\tilde{f}(\omega)\}$.
  So, it suffices for us to derive the variance of $\tilde{f}(\omega)$.
  Note that by Euler's formula, we have
  $ \exp(\iota k\omega) = \cos(k\omega) + \iota \sin(\iota k\omega) $.
  So $\tilde{f}(\omega)$ can be written as 
  \begin{align*}
    \tilde{f}(\omega)
    &= \frac{1}{2\pi}\sum_{k=1-n}^{n-1} K\left(\frac{k}{\ell}\right)\exp(\iota k\omega)\tilde{\gamma}_k\\
    &= \frac{1}{2\pi}\sum_{k=1-n}^{n-1} K\left(\frac{k}{\ell}\right)\left\{\cos(k\omega) + \iota \sin(\iota k\omega)\right\}\tilde{\gamma}_k \\
    &= \frac{1}{2\pi}\sum_{k=1-n}^{n-1} K\left(\frac{k}{\ell}\right)\cos(k\omega)\tilde{\gamma}_k.
  \end{align*}
  We 
  approximate $\tilde{f}(\omega)$ by 
  $\ddot{f}(\omega)$, which is defined as $\tilde{f}(\omega)$, but with $\bar{X}_n$ replaced by $\mu$, i.e.,
  \begin{align*}
    \ddot{f}(\omega)
    & = \frac{1}{2\pi}\sum_{i=1}^{n}\sum_{j=1}^{n}K\left( \frac{|i-j|}{\ell} \right)\cos\left\{ (i-j)\omega\right\} \frac{(X_i-\mu)(X_{j}-\mu)}{n} \\
    &=  \frac{1}{2\pi}\sum_{i=1}^{n}\sum_{j=1}^{n}K\left( \frac{|i-j|}{\ell} \right)e^{\iota(i-j)\omega} \frac{(X_i-\mu)(X_{j}-\mu)}{n} \\
    &= \frac{1}{2\pi n}\sum_{i=1}^{n}\sum_{j=1}^{n}K\left( \frac{|i-j|}{\ell} \right)(X_i-\mu)e^{\iota i\omega} (X_{j}-\mu)e^{-\iota j\omega}.
  \end{align*}
  To show that $\tilde{f}(\omega)$ can be approximated by $\ddot{f}(\omega)$, 
  we use Lemma A.6 in \cite{LiuChan2022} with
  $W_n(i,j)=K(|i-j|/\ell)\cos\left\{(i-j)\omega\right\}/(2\pi n)$.
  Checking the necessary conditions as before, we have 
  \begin{align*}
     E_{4,n} &= \frac{1}{n}; \qquad E_{5,n} \leq E^2_{6,n} = O\left(\frac{\ell}{n}\right);
     \qquad
    E_{7,n} = O\left(\frac{\ell^{1/2}}{n^{1/2}}\right).
  \end{align*}
  Therefore,
  using Minkowski inequality and noting that 
  $
    \E\left(\left|\tilde{f}(\omega)-\ddot{f}(\omega)\right|\right) \leq \|\tilde{f}(\omega)-\ddot{f}(\omega)\| = O\left(\ell^{1/2}/{n}\right),
  $
  we have 
  \begin{align}\label{eqt:equiv_ddot_tilde_f}
    \|\tilde{f}(\omega)-\E\{\tilde{f}(\omega)\}\| 
    = \|\ddot{f}(\omega)-\E\{\ddot{f}(\omega)\}\| 
    + O({\ell^{1/2}}/{n}).
  \end{align}
  Without loss of generality,
  we assume $\mu=0$ from now on.
  If $\mu\neq0$, we can always subtract the mean from the data and all results apply.
  Now 
  \[
  	\ddot{f}(\omega) 
	= \frac{1}{2\pi n}\sum_{1\leq i,j\leq n}K(|i-j|/\ell)X_i e^{\iota i\omega} X_{j} e^{-\iota j\omega} 
  	= \frac{1}{2\pi n}\sum_{1\leq i,j\leq n}K(|i-j|/\ell)X_i X_{j} \cos\{(i-j)\omega\}.
	\]
  {
  When $\omega=0$, we have $\ddot{f}(\omega) \equiv \ddot{v}/(2\pi)$, where $\ddot{v}$ is defined in (\ref{eqt:definition_vddot}). 
  Thus, the results follows by rescaling the estimator by a factor of $2\pi$. 
  When $\omega=\pi$, we let $X'_i = (-1)^i X_i$ for each $i\in\mathbb{Z}$, and  
  denote $\gamma_k' = (-1)^k \gamma_k = \E(X'_0X'_{k})$ for each $k\in\mathbb{Z}$. 
  In this case, 
  $\ddot{f}(\omega) = \sum_{1\leq i,j\leq n}K(|i-j|/\ell) X'_i X'_{j}  / (2\pi n)$
  and 
  $f(\omega) = \sum_{k\in\mathbb{Z}} (-1)^k\gamma_k = \sum_{k\in\mathbb{Z}} \gamma_k'$, 
  where we have used the fact that $\cos\{(i-j)\pi\} = (-1)^{i-j} = (-1)^{i+j}$.
  Then following the proof for the case $\omega=0$, we obtain the desired result. 
  So, it remains to prove the result for $\omega\in(0, \pi)$. }

  We decompose $\ddot{f}(\omega)$ into multiple blocks using the same $m,r_{n}\in\mathbb{N}_{0}$
  as defined in the proof of Theorem~\ref{thm:BiasVar}.
  Also recall the partitions
  $\mathcal{A}_t=\{(i,j)\in([1,n]\cap \mathbb N)^2: (t-1)\ell+1 \leq i-j\leq t \ell \}$
  for $t\in\mathbb{N} \bigcap [1,s]$
  and partition $\mathcal{A}_t$ into two parts:
  $\bigcup_{k=1}^{m-t+1}\mathcal{B}_{k,t}$ and
  $\mathcal R_t = \mathcal{A}_t - \left(\bigcup_{k=1}^{m-t+1}\mathcal{B}_{k,t}\right)$, 
  where
  $
    \mathcal B_{k,t}=\left\{(i,j) \in \mathcal{A}_t : i\in\{(k+t-1)\ell+1,\ldots,(k+t)\ell\} \right\}.
  $
  Following this partition, define the quantity
  \begin{align}\label{DefOfGkStar}
    \ddot{f}_s(\omega) 
      &= \frac{1}{2\pi n}\underset{(i,j)\in\bigcup_{t=1}^s\mathcal{A}_t}{\sum\sum} K\left( \frac{|i-j|}{\ell}\right) \left\{ X_i e^{\iota i\omega} X_{j} e^{-\iota j\omega} + X_i e^{-\iota i\omega} X_{j} e^{\iota j\omega} \right\} \nonumber 
      + \frac{1}{n}\sum_{i=1}^n K(0) X_i^2\\
      &= \underbrace{\sum_{t=1}^{s} \frac{\ell}{2\pi n} \sum_{k=1}^{m-t+1} \left[\frac{1}{\ell}\underset{(i,j)\in\mathcal B_{k,t}}{\sum\sum}  K\left(\frac{|i -j|}{\ell}\right)\left\{ X_i e^{\iota i\omega} X_{j} e^{-\iota j\omega} + X_i e^{-\iota i\omega} X_{j} e^{\iota j\omega} \right\}\right]}_{M^{\textnormal{Spec}}_{n,s}} \\
      &\qquad +  \underbrace{\frac{\ell}{2\pi n}\sum_{t=1}^{s}\left[\frac{1}{\ell} \underset{(i,j)\in\mathcal R_t}{\sum\sum}  K\left(\frac{|i-j|}{\ell}\right)\left\{ X_i e^{\iota i\omega} X_{j} e^{-\iota j\omega} + X_i e^{-\iota i\omega} X_{j} e^{\iota j\omega} \right\}\right]}_{R^{\textnormal{Spec}}_{1,n,s}} \\
      &\quad +  \underbrace{\frac{1}{2\pi n}\sum_{i=1}^n K(0)X_i^2}_{R^{\textnormal{Spec}}_{2,n}}.
  \end{align}

  Next, decompose the main part $M^{\textnormal{Spec}}_{n,s}$ into multiple blocks as
   \begin{align}\label{def_MnsStar}
    M^{\textnormal{Spec}}_{n,s} 
      &= \frac{\ell}{2\pi n} \sum_{t=1}^{s} \sum_{k=1}^{m-t+1} \underbrace{\left[\frac{1}{\ell}\sum_{i=(k+t-1)\ell+1}^{(k+t)\ell} \sum_{j=i-t\ell }^{i-(t-1  )\ell-1} 
      K\left(\frac{|i-j|}{\ell}\right) \left\{ X_i e^{\iota i\omega} X_{j} e^{-\iota j\omega} + X_i e^{-\iota i\omega} X_{j} e^{\iota j\omega} \right\} \right]}_{G^{\textnormal{Spec}}_{k,t}}. 
  \end{align}
  Also note that
  $
    \Var(\ddot{f}(\omega)) = \Var(\lim_{s\to\infty} \ddot{f}_{s}(\omega)).
  $
  By construction, for fixed $s$, and $r_n<\ell$ thus
  $\|R^{\textnormal{Spec}}_{1,n,s}-\E(R^{\textnormal{Spec}}_{1,n,s})\| 
  =o(\|M^{\textnormal{Spec}}_{n,s}  - \E(M^{\textnormal{Spec}}_{n,s} )\|)$.
  We also have
  $\|R^{\textnormal{Spec}}_{2,n}-\E(R^{\textnormal{Spec}}_{2,n})\|=O(1/\sqrt{n})$.
  So we focus on $
  \|\ddot{f}_{s}(\omega)- 
  \E\{\ddot{f}_{s}(\omega)\}\| 
  \sim \|M^{\textnormal{Spec}}_{n,s}- \E(M^{\textnormal{Spec}}_{n,s})\|$
   from now on.
  In other words,
  $$
    \Var(\ddot{f}_s(\omega)) \sim
    \Var(M^{\textnormal{Spec}}_{n,s}) =
    \Var\left(\frac{\ell}{2\pi n}\sum_{t=1}^{s} \sum_{k=1}^{m-t+1} G^{\textnormal{Spec}}_{k,t}\right).
  $$
  Further decompose the variance of $\sum_{t=1}^{s} \sum_{k=1}^{m-t+1}  G^{\textnormal{Spec}}_{k,t}$ into
  \begin{equation}
    \left\|\sum_{t=1}^{s} \sum_{k=1}^{m-t+1} \left\{ G^{\textnormal{Spec}}_{k,t} - \E (G^{\textnormal{Spec}}_{k,t}) \right\} \right\| \lesseqgtr I_1^{\textnormal{Spec}} \pm I_0^{\textnormal{Spec}},
  \end{equation}
  where
  $$
    I_1^{\textnormal{Spec}} = \left\| \sum_{t=1}^{s} \sum_{k=1}^{m-t+1} \left\{ G^{\textnormal{Spec}}_{k,t}  - \E (G^{\textnormal{Spec}}_{k,t}  \mid \mathcal{F}_{\varsigma_{k,t}})\right\}\right\|
      \quad\text{and}\quad
    I_0^{\textnormal{Spec}} = \left\| \sum_{t=1}^{s} \sum_{k=1}^{m-t+1} \left\{ \E (G^{\textnormal{Spec}}_{k,t}  \mid \mathcal{F}_{\varsigma_{k,t}}) - \E (G^{\textnormal{Spec}}_{k,t}) \right\}\right\|,
  $$
  and $ \varsigma_{k,t}$ is the smallest value of the index $i$ of $X_i$ in defining $G^{\textnormal{Spec}}_{k,t}$.
  Using Lemma~\ref{lemma:spec2} and Lemma~\ref{lemma:spec1} (iii), 
  as $n,\ell\to\infty$ and $s$ is fixed,
  \[
    \Var\left(\ddot{f}_s(\omega) \right) = {2 f(\omega)^2 \int_{0}^{s}K^2(t)\dd t } \frac{\ell}{n}+o\left\{ \frac{\ell}{n} \int_{0}^{s}K^2(t)\dd t \right\}.
  \]  
  Therefore, for fixed $s$,
  \[
    \lim_{n,\ell\to\infty} \frac{n}{\ell} \left\{ \Var\left(\ddot{f}_s(\omega) \right) - {2f(\omega)\int_{0}^{s}K^2(t)\dd t} \frac{\ell}{n}\right\}
     = 0.
  \]
  For all $s$, the statistic $\ddot{f}_s(\omega)$ can be bounded as follows:
  \begin{align*}
        \left|\ddot{f}_s(\omega)\right| 
        &= \left\lvert \underset{|i-j|\leq s\ell}{\sum\sum} K\left( \frac{|i-j|}{\ell} \right)
        e^{\iota (i-j) \omega}\frac{X_iX_j}{n} \right\rvert \\
        & \leq \underset{|i-j|\leq n}{\sum\sum} \left\lvert K\left( \frac{|i-j|}{\ell} \right)
        \right\rvert  \left\lvert \frac{X_iX_j}{n} \right\rvert |e^{\iota (i-j) \omega}| \\
        &\leq \underset{|i-j|\leq n}{\sum\sum} \left\lvert K\left( \frac{|i-j|}{\ell} \right)
        \right\rvert  \left\lvert \frac{X_iX_j}{n} \right\rvert 
        \equiv U.
  \end{align*}
  As in the proof of Theorem~\ref{thm:BiasVar},
  \[
    \Var\left(\ddot{f}(\omega)\right) = \lim_{s\to\infty} \Var\left(\ddot{f}_s(\omega)\right).
  \]
  Hence, as $n,\ell\to\infty$,
  \begin{align}\label{eqt:var_of_fddot}
    \Var\left(\ddot{f}(\omega)\right) = {2f(\omega)^2\int_{0}^{\infty}K^2(t)\dd t} \frac{\ell}{n}+o\left\{ \frac{\ell}{n} \int_{0}^{\infty}K^2(t)\dd t \right\}.
  \end{align}
  Recall that $A = \int_{0}^{\infty}K^2(t)\dd t < \infty$ and $\zeta(\omega) = \{1+\mathbb{1}(\omega/\pi\in\mathbb{Z})\}/2=1/2$ 
  for $\omega\in(0,\pi)$.
  In view of (\ref{eqt:equiv_ddot_tilde_f}) and (\ref{eqt:var_of_fddot}),
  we have 
  \[
  	\Var(\tilde{f}(\omega)) \sim 2f(\omega)^2A\ell/n  = 4 A \zeta(\omega) f(\omega)^2 \ell/n.
  \] 
  Finally, we note that 
  \[
    \lim_{h\to\infty} \lim_{n\to\infty} \MSE_h\{\hat{f}(\omega; \ell; K; \theta_{\star})\} = \lim_{n\to\infty} n^{2p/(2p+1)} \MSE\{\hat{f}(\omega; \ell; K; \theta_{\star})\}.
  \]
  So, we have 
  \[
    \lim_{h\to\infty} \lim_{n\to\infty} \MSE_h\{\hat{f}(\omega; \ell; K; \bar{\theta})\}
      = \lim_{n\to\infty} n^{2p/(2p+1)} \MSE\{\hat{f}(\omega; \ell; K; \theta_{\star})\}.
  \]
  It completes the proof.
\end{proof}
}

\subsection{Proof of lemmas}

The following lemmas are used
for kernels without the restriction that 
that $K(t)=0$ for $|t|>1$, which is the main difference from those results in \citet{chanyau2015_hoc}. 
Lemmas~\ref{lemma:hoc1} and \ref{lemma:hoc2} are useful in deriving the variance of the kernel estimator in Theorem~\ref{thm:BiasVar}.
Lemma~\ref{lemma:multblock} and \ref{lemma:multblock2} are useful in deriving the covariance of different entires of 
the covariance matrix estimator in Theorem~\ref{prop:multivariateBiasVar}.
{Lemma~\ref{lemma:spec1} and \ref{lemma:spec2} are useful in deriving the variance of 
the spectral density estimator in Theorem~\ref{thm:spectralBV}.}

\begin{lemma}[Slight modifications of Lemma 3 in \citet{chanyau2015_hoc}]\label{lemma:hoc1}
  Suppose that $X_i\in\mathcal{L}^\nu$ and $\Delta_4<\infty$ for some $\nu>4$.
  Also let $\E(X_1)=0$. Suppose $K(\cdot)\in\mathcal{K}_p$.
  Let $\mathbb{B}_u$ be a standard Brownian motion
  and $G_{k,t}$ as defined in (\ref{def_Mns}).
  Define additionally, for $k,t=1,2,\ldots$, that
  \begin{align}\label{eqt:definition_Gkt}
    \mathbb{G}_{k,t} = v\int_{k+t-1}^{k+t}\int_{s-t}^{s-t+1}K(s-h) \dd\mathbb{B}_h \dd\mathbb{B}_s.
  \end{align}
  Then as $\ell\to\infty$, the following results hold. 

  (i) $\left(G_{1,1},G_{1,2},G_{2,1},G_{2,2}\right)\Rightarrow\left(\mathbb{G}_{1,1},\mathbb{G}_{1,2},\mathbb{G}_{2,1},\mathbb{G}_{2,2}\right)$. 

  (ii) For $k,t=1,2$,
    \begin{eqnarray*}
      \E\left(G_{k,t}\right) &\to& v K(0),\qquad \Var\left(G_{k,t}\right) \to v^2\int_{t-1}^{t}K^2(u) \,\dd u , \\
     \Cov\left(G_{1,1}, G_{1,2}\right)&\to& 0,\qquad  \Cov\left(G_{1,1}, G_{2,2}\right)\to 0, \qquad \Cov\left(G_{1,1}, G_{2,1}\right)\to 0, \\
     \Cov\left(G_{1,2}, G_{2,2}\right)&\to& 0,\qquad \Cov\left(G_{2,1}, G_{2,2}\right)\to 0 .
    \end{eqnarray*} 

  (iii) For $m\to\infty$,
    \[
          \left\| \sum_{t=1}^{s} \sum_{k=1}^{m-t+1}\left[\E\left\{G_{k,t}|\mathcal{F}_{(k-1)\ell+1}\right\}-\E G_{k,t}\right]\right\| = o\left[\left\{  m \int_0^{s} K^2(t) \dd t\right\}^{1/2}\right].
    \]
\end{lemma}

\begin{proof}[Proof of Lemma~\ref{lemma:hoc1}]
  The proof is identical with the original lemma with the following modifications:
  (i) $a_\ell(t) = \ell^{p+1}K(t/\ell)$;
  (ii) $\alpha(u) = a_\ell(\ell u)/\ell^{p+1} = K(t)$ is independent of $\ell$,
  where $a_\ell(\cdot), \alpha(\cdot)$ are the functions in their paper.

  For (iii), let $\mathcal{P}_h \cdot = \E(\cdot \mid \mathcal{F}_h) - \E(\cdot \mid \mathcal{F}_{h-1})$ be the projection operator. 
  By the orthogonality of martingale differences,
    \begin{eqnarray*}
      \left\| \sum_{t=1}^{s} \sum_{k=1}^{m-t+1}\left[\E\left\{G_{k,t}|\mathcal{F}_{(k-1)\ell+1}\right\}-\E G_{k,t}\right]\right\|^2
        &=& \left\| \sum_{t=1}^{s} \sum_{k=1}^{m-t+1} \sum_{h=-\infty}^{(k-1)\ell+1} \mathcal{P}_h G_{k,t}\right\|^2\\
        &=& \left\| \sum_{h=-\infty}^{(m-1)\ell+1} \mathcal{P}_h \left(\sum_{t=1}^{s} \sum_{k=1}^{m-t+1} G_{k,t} \mathbb{1} \{h\leq (k-1)\ell + 1\}\right) \right\|^2\\
        &=& \sum_{h=-\infty}^{(m-1)\ell+1} \left\|\mathcal{P}_h \left(\sum_{t=1}^{s} \sum_{k=1}^{m-t+1} G_{k,t} \mathcal{J}_{h,k}\right)\right\|^2,
    \end{eqnarray*}
  where $\mathcal{J}_{h,k}= \mathbb{1} \{h\leq (k-1)\ell + 1\}$.
  By the definition of $\mathcal{P}_h$,
  $$
     \mathcal{P}_{0}(X_{i-h}X_{j-h}) = \E\left( X_{i-h}X_{j-h} - {X}_{i-h,\{0\}}{X}_{j-h,\{0\}} \mid \mathcal{F}_0\right).
  $$
  Using similar arguments as (45) and (51) in \cite{chanyau2015_hoc}, we have 
    \begin{eqnarray*}
    \left\|\mathcal{P}_h \left(\sum_{t=1}^{s} \sum_{k=1}^{m-t+1} G_{k,t} \mathcal{J}_{h,k}\right) \right\|
      &=& \left\| \sum_{t=1}^{s} \sum_{k=1}^{m-t+1} \frac{1}{\ell}\sum_{i=(k+t-1)\ell+1}^{(k+t)\ell} \sum_{j=i-t\ell+1}^{i-(t-1)\ell-1} K\left(\frac{i-j}{\ell}\right) \mathcal{P}_{h}(X_iX_j) \mathcal{J}_{h,k}\right\|\\
      &\leq & \sum_{t=1}^{s} \sum_{k=1}^{m-t+1} \frac{1}{\ell} \sum_{i=(k+t-1)\ell+1}^{(k+t)\ell} \delta_{4,i-h} \mathcal{J}_{h,k}\left\|\sum_{j=i-t\ell+1}^{i-(t-1)\ell-1} K\left(\frac{i-j}{\ell}\right) X_{j-h}\right\|_4  
      \\
      && \qquad + \sum_{t=1}^{s} \sum_{k=1}^{m-t+1} \frac{1}{\ell} \sum_{i=(k-1)\ell+1}^{(k+1)\ell-1} \delta_{4,j-h}\mathcal{J}_{h,k}\left\|\sum_{i=j+(t-1)\ell+1}^{j+t\ell-1} K\left(\frac{i-j}{\ell}\right) X_{i-h}\right\|_4,
  \end{eqnarray*}
  by Minkowski and H\"{o}lder's inequalities.
  Using Lemma 1 of \cite{wu2010}, we have
  \begin{eqnarray}\label{eqt:lemmaderive}
   &&\sum_{t=1}^{s} \sum_{k=1}^{m-t+1} \frac{1}{\ell} \sum_{i=(k+t-1)\ell+1}^{(k+t)\ell} \delta_{4,i-h}\mathcal{J}_{h,k}\left\|\sum_{j=i-t\ell+1}^{i-(t-1)\ell-1} K\left(\frac{i-j}{\ell}\right) X_{j-h}\right\|_4 \nonumber\\
   &\leq& \frac{C\Delta_4}{\ell} \sum_{t=1}^{s} \sum_{k=1}^{m-t+1} \sum_{i=(k+t-1)\ell+1}^{(k+t)\ell} \delta_{4,i-h}\mathcal{J}_{h,k} \left\{ \sum_{(t-1)\ell+1}^{t\ell-1} K^2(u/\ell)\right\}^{1/2} 
  \end{eqnarray}
  for some constant $C$ and similarly for the other term.
  Using (\ref{eqt:lemmaderive}) and the fact that $\Delta_4<\infty$, for large enough $\ell$,   
  \begin{eqnarray*}
     && \left\| \sum_{t=1}^{s} \sum_{k=1}^{m-t+1}\left[\E\left\{G_{k,t}|\mathcal{F}_{(k-1)\ell+1}\right\}-\E(G_{k,t}) \right]\right\|^2\\
      &=& \sum_{h=-\infty}^{(m-1)\ell+1} \left\|\mathcal{P}_h \left(\sum_{t=1}^{s} \sum_{k=1}^{m-t+1} G_{k,t} \mathcal{J}_{h,k}\right)\right\|^2 \\
      &\leq& \sum_{h=-\infty}^{(m-1)\ell+1} \left(\frac{C\Delta_4}{\ell}\right)^2 \sum_{t=1}^{s} \left\{ \sum_{u=(t-1)\ell+1}^{t\ell-1} K^2(u/\ell)\right\}^{1/2}\sum_{k=1}^{m-t+1} \sum_{i=(k+t-1)\ell+1}^{(k+t)\ell} \Delta_4 \delta_{4,i-h}\mathcal{J}_{h,k} \\
      &&\qquad \times \sum_{t'=1}^{s} \left\{ \sum_{u'=(t'-1)\ell+1}^{t'\ell-1} K^2(u'/\ell)\right\}^{1/2}\\
      &\leq&  \frac{C^2\Delta^3_4}{\ell^2} \sum_{t=1}^{s} \sum_{k=1}^{m-t+1}\left\{\sum_{u=(t-1)\ell+1}^{t\ell-1} K^2(u/\ell)\right\}^{1/2}  \sum_{h=-\infty}^{(k-1)\ell+1} \sum_{i=(k+t-1)\ell+1}^{(k+t)\ell} \delta_{4,i-h} \left\{ \sum_{u'=1}^{s\ell-1} K^2(u'/\ell)\right\}^{1/2}\\
      &\leq&  \frac{C^2\Delta^3_4}{\ell^2} \sum_{t=1}^{s} \sum_{k=1}^{m-t+1}\left\{ \sum_{u=(t-1)\ell+1}^{t\ell-1} K^2(u/\ell)\right\}^{1/2}  \sum_{b=t\ell}^{(t+1)\ell+1} \sum_{h=b}^{\infty} \delta_{4,h} \left\{ \sum_{u'=1}^{s\ell-1} K^2(u'/\ell)\right\}^{1/2}\\
      &\leq&  \frac{C^2\Delta^3_4}{\ell^2} \sum_{t=1}^{s} \sum_{k=1}^{m-t+1}\left\{ \sum_{u=(t-1)\ell+1}^{t\ell-1} K^2(u/\ell)\right\}^{1/2}  \left\{ \sum_{u'=1}^{s\ell-1} K^2(u'/\ell) \right\}^{1/2}  o(\ell) \\
      &\leq&  \frac{C^2\Delta^3_4}{\ell^2} \sum_{k=1}^{m} \sum_{u=1}^{s\ell-1} K^2(u/\ell) \times o(\ell) \\
      &\leq& o\left\{  m \int_0^{s} K^2(t) \dd t\right\}.
  \end{eqnarray*}
  And the proof is completed.
\end{proof}

\begin{lemma}[Slight modifications of Lemma 4 in \citet{chanyau2015_hoc}]\label{lemma:hoc2}
Under conditions of Theorem~\ref{thm:BiasVar}. As $\ell\to\infty$,
  \[
    I_1 = \left\| \sum_{t=1}^{s} \sum_{k=1}^{m-t+1} \left\{ G_{k,t} - \E (G_{k,t} \mid \mathcal{F}_{\varsigma_{k,t}})\right\}\right\| \sim \left( \frac{\int_{0}^s K^2(t) \dd t v^2n}{\ell} \right)^{1/2}.
  \]
\end{lemma}

\begin{proof}[Proof of Lemma~\ref{lemma:hoc2}]
The proof can be done by replacing Lemma 3 in the original proof by Lemma~\ref{lemma:hoc1}.
\end{proof}

\begin{lemma}[Multivariate version of Lemma~\ref{lemma:hoc1}]\label{lemma:multblock}
Suppose that $X_{i}^{(u)}\in\mathcal{L}^{\nu}$ for some $\nu>4$, $\Delta_{4}^{(u)}<\infty$ 
and $\E\{X_{1}^{(u)} \}=0$ for all $u\in\{1,\ldots,d\}$.
The kernel function $K(\cdot)\in\mathcal{K}_p$.
Let $\mathbb{B}_{r}= (\mathbb{B}^{(1)}_{r},\ldots,\mathbb{B}^{(d)}_{r} )^{\T}$ be a $d$-dimensional Brownian motion
for $r\geq 0$ with $\Var(\mathbb{B}_{r})= {\rho} r$, where 
$\rho\in\mathbb{R}^{d\times d}$ satisfies that $\rho^{(u,v)}=V^{(u,v)}/\sqrt{V^{(u,u)}V^{(v,v)}}$
for $u,v\in\{1,\ldots,d\}$.
For $k,t=1,2,\dots$ and $u,v\in\{1,\ldots,d\}$, define that
\begin{align}
  G_{k,t}^{(u,v)}&=\frac{1}{\ell}\sum_{i=(k+t-1)\ell+1}^{(k+t)\ell} \sum_{j=i-t\ell}^{i-(t-1)\ell-1} K\left(\frac{|i-j|}{\ell}\right) \left\{X_i^{(u)} X_{j}^{(v)} + X_i^{(v)} X_{j}^{(u)}\right\},\\
  \mathbb{G}_{k,t}^{(u,v)}&=\sqrt{V^{(u,u)}V^{(v,v)}}\int_{k+t-1}^{k+t}\int_{s-t}^{s-t+1}K(s-h)
  \left\{ \dd\mathbb{B}^{(u)}_h \dd\mathbb{B}^{(v)}_{s} +\dd\mathbb{B}^{(v)}_h \dd\mathbb{B}^{(u)}_{s}\right\}.
\end{align}
Then, as $\ell\to\infty$, the following results hold: 
\begin{enumerate}
  \item $G_{k,t}^{(u,v)}\Rightarrow \mathbb{G}_{k,t}^{(u,v)}$ jointly for all $u,v\in\{1, \ldots, d\}$ and $k,t\in\{1,2\}$. 

  \item For $u,v, u', v'\in\{1,\ldots,d\}$ and $k,t\in\{1,2\}$,
  {
    \begin{align*}
    \E\left\{ G_{k,t}^{(u,v)} \right\} &\to0,\\
    \Cov\left\{G_{k,t}^{(u,v)}, G_{k,t}^{(u',v')}\right\}&\to 2\left\{ V^{(u,u')}V^{(v,v')} + V^{(u,v')}V^{(u,v')} \right\}\int_{t-1}^{t} K^2(u) \dd u,\\
    \Cov\left\{G_{1,1}^{(u,v)},G_{1,2}^{(u',v')}\right\} &\to 0,\qquad 
    \Cov\left\{G_{1,1}^{(u,v)},G_{2,2}^{(u',v')}\right\} \to 0,\qquad 
    \Cov\left\{G_{1,1}^{(u,v)},G_{2,1}^{(u',v')}\right\} \to 0,\\
    \Cov\left\{G_{1,2}^{(u,v)},G_{2,2}^{(u',v')}\right\} &\to 0,\qquad 
    \Cov\left\{G_{2,1}^{(u,v)},G_{2,2}^{(u',v')}\right\} \to 0.
    \end{align*}
    }
  \item For $u,v\in\{1, \ldots, d\}$, we have, as $m\to\infty$, that 
    \[
      \left\| \sum_{t=1}^{s} \sum_{k=1}^{m-t+1}\left[\E\left\{G^{(u,v)}_{k,t}|\mathcal{F}_{(k-1)\ell+1}\right\}-\E G^{(u,v)}_{k,t}\right]\right\| 
      = o\left[\left\{  m \int_0^{s} K^2(t) \dd t\right\}^{1/2}\right].
    \]
\end{enumerate}
\end{lemma}

\begin{proof}[Proof of Lemma~\ref{lemma:multblock}]
Similar proofs can be found in \cite{chanyau2014_rTACM} Lemma D.4 and \cite{LiuChan2022} Lemma B.8.
\begin{enumerate}
  \item {
  Using the functional central limit theorem on the class $D[0,3]$ (Theorem 3 of \citet{wu2005})
  and Cramer--Wold Theorem, we have 
  \[
      \left\{\frac{1}{\sqrt{\ell}} \sum_{i=1}^{\lfloor \ell t\rfloor} X_i: t\in[0,3]\right\} 
      \Rightarrow \{ \rho^{1/2}\mathbb{B}_t :t\in[0,3]\}.
  \]
  By the continuous mapping theorem,
  $G_{k,t}^{(u,v)}\Rightarrow \mathbb{G}_{k,t}^{(u,v)}$ jointly for all $u,v\in\{1, \ldots, d\}$ and $k,t\in\{1,2\}$. 

  \item  Because $X_i\in\mathcal{L}^\nu$ for $\nu>4$, 
  $\{ G_{k,t}^{(u,v)}\}^2 $ is uniformly integrable.
  Therefore, weak convergence implies convergence of moments, i.e.,
  \[
    \E\left\{G_{k,t}^{(u,v)}\right\}           \to \E\left\{\mathbb{G}_{k,t}^{(u,v)}\right\},\quad
    \E\left\{G_{k,t}^{(u,v)}\right\}^2         \to \E\left\{\mathbb{G}_{k,t}^{(u,v)}\right\}^2, \quad
    \E\left\{G_{1,1}^{(u,v)}G_{1,2}^{(u,v)}\right\}\to \E\left\{\mathbb{G}_{1,1}^{(u,v)}\mathbb{G}_{1,2}^{(u,v)}\right\}.
  \]
  Then we can write $\mathbb{G}_{k,t}^{(u,v)}$ as 
  \[
      \mathbb{G}_{k,t}^{(u,v)} 
      = \sqrt{V^{(u,u)}V^{(v,v)}}\left\{\overbrace{\int_{k+t-1}^{k+t}\int_{s-t}^{s-t+1}K(s-h)d\mathbb{B}^{(u)}_hd\mathbb{B}^{(v)}_{s}}^{\mathbb{H}_{k,t}^{(u,v)}}
      + \overbrace{\int_{k+t-1}^{k+t}\int_{s-t}^{s-t+1}K(s-h)d\mathbb{B}^{(v)}_hd\mathbb{B}^{(u)}_{s}}^{\mathbb{H}_{k,t}^{(v,u)}}\right\}.
  \]
  Since the integrand $\psi^{(u)}(s)=\int_{s-t}^{s-t+1}K(s-h)d\mathbb{B}^{(u)}_h$ is adapted to $\mathcal{F}_t$, 
  similarly for $\psi^{(v)}(s)$, thus by the isometry,
  \[
    \E\left\{\mathbb{H}_{k,t}^{(u,v)}\right\} = \E\left\{\mathbb{H}_{k,t}^{(v,u)}\right\} = 0
    \qquad\text{and}\qquad
      \E\left\{\mathbb{G}_{k,t}^{(u,v)}\right\} = \E\left\{\mathbb{G}_{k,t}^{(v,u)}\right\} = 0.
  \]
  For the second moment, again using the isometry,
  \begin{align*}
  	\E\left\{\mathbb{H}_{k,t}^{(u,v)}\mathbb{H}_{k,t}^{(u',v')}\right\}
	&= \E \left\{ \int_{k+t-1}^{k+t}\int_{s-t}^{s-t+1} K(s-h)  \mathbb{B}^{(u)}_{h} \mathbb{B}^{(v)}_{s} \right\}
		\left\{ \int_{k+t-1}^{k+t}\int_{s'-t}^{s'-t+1} K(s'-h')  \mathbb{B}^{(u')}_{h'} \mathbb{B}^{(v')}_{s'} \right\}\\
	&= \int_{k+t-1}^{k+t}\int_{s-t}^{s-t+1} K^2(s-h)  \rho^{(u,u')} \rho^{(v,v')} \dd h  \dd s  \\
	&= \rho^{(u,u')} \rho^{(v,v')} \int_{t-1}^{t}K^2(w)\dd w.
  \end{align*}
  Combining the above results, we have 
  \begin{align*}
      &\Cov\left\{\mathbb{G}_{k,t}^{(u,v)}, \mathbb{G}_{k,t}^{(u',v')}\right\} \\
      &\qquad= \sqrt{V^{(u,u)}V^{(v,v)} V^{(u',u')}V^{(v',v')}}
      		\Cov\left\{ \mathbb{H}_{k,t}^{(u,v)} + \mathbb{H}_{k,t}^{(v,u)}, \; 
					\mathbb{H}_{k,t}^{(u',v')} + \mathbb{H}_{k,t}^{(v',u')}
					\right\}\\
      &\qquad= \sqrt{V^{(u,u)}V^{(v,v)} V^{(u',u')}V^{(v',v')}} \int_{t-1}^{t}K^2(w)\dd w\\
      &\qquad\qquad
      	\times \left\{ \rho^{(u,u')} \rho^{(v,v')} 
			+ \rho^{(u,v')} \rho^{(v,u')}
			+ \rho^{(v,u')} \rho^{(u,v')}
			+\rho^{(v,v')} \rho^{(u,u')}
		\right\}\\
	&\qquad= \int_{t-1}^{t}K^2(w)\dd w \times
	 	\left\{ V^{(u,u')} V^{(v,v')} 
			+ V^{(u,v')} V^{(v,u')}
			+ V^{(v,u')} V^{(u,v')}
			+V^{(v,v')} V^{(u,u')}
		\right\} \\
	&\qquad= 2\int_{t-1}^{t}K^2(w)\dd w \times 
		\left\{ V^{(u,u')} V^{(v,v')} 
			+ V^{(u,v')} V^{(v,u')}
		\right\}
  \end{align*}
  For the covariance,
  we know that $\E\left\{\mathbb{H}_{1,1}^{(u,v)}\mathbb{H}_{1,2}^{(u',v')}\right\} 
  = \E\left\{ \mathbb{H}_{1,1}^{(u,v)}\mathbb{H}_{1,2}^{(v',u')} \right\}
  =  \E\left\{ \mathbb{H}_{1,1}^{(v,u)}\mathbb{H}_{1,2}^{(u',v')} \right\}
  = \E\left\{ \mathbb{H}_{1,1}^{(v,u)}\mathbb{H}_{1,2}^{(v',u')} \right\}=0$,
  therefore $\Cov\left\{ G_{1,1}^{(u,v)},G_{1,2}^{(u',v')} \right\}\to 0$.
  Similarly for other combinations of $k,t$.}

  \item
  By the same arguments in the proof of Lemma~\ref{lemma:hoc1} and consider the entry-wise $G^{(u,v)}_{k,t}$, we have
  \[
    \left\| \sum_{t=1}^{s} \sum_{k=1}^{m-t+1}\left[\E\left\{G^{(u,v)}_{k,t}|\mathcal{F}_{(k-1)\ell+1}\right\}-\E G^{(u,v)}_{k,t}\right]\right\|^2
        = \sum_{h=-\infty}^{(m-1)\ell+1} \left\|\mathcal{P}_h \left(\sum_{t=1}^{s} \sum_{k=1}^{m-t+1} G^{(u,v)}_{k,t} \mathcal{J}_{h,k}\right)\right\|^2
  \]
  where recall $\mathcal{J}_{h,k} = \mathbb{1} \{h\leq (k-1)\ell + 1\}$.
  By the definition of $\mathcal{P}_h$,
  $$
     \mathcal{P}_{0}\left(X^{(u)}_{i-h}X^{(v)}_{j-h} + X^{(v)}_{i-h}X^{(u)}_{j-h}\right) = \E\left( X^{(u)}_{i-h}X^{(v)}_{j-h} - {X}^{(u)}_{i-h,\{0\}}{X}^{(v)}_{j-h,\{0\}}
     +  X^{(v)}_{i-h}{X}^{(u)}_{j-h} - {X}^{(v)}_{i-h,\{0\}}{X}^{(u)}_{j-h,\{0\}} \mid \mathcal{F}_0\right).
  $$
  Using similar arguments as Equations (45) and (51) in \cite{chanyau2015_hoc}, we have 
    \begin{eqnarray*}
      & & \left\|\mathcal{P}_h \left(\sum_{t=1}^{s} \sum_{k=1}^{m-t+1} G^{(u,v)}_{k,t} \mathcal{J}_{h,k}\right) \right\| \\
      &=& \left\| \sum_{t=1}^{s} \sum_{k=1}^{m-t+1} \frac{1}{\ell}\sum_{i=(k+t-1)\ell+1}^{(k+t)\ell} \sum_{j=i-t\ell+1}^{i-(t-1)\ell-1} K\left(\frac{i-j}{\ell}\right) \mathcal{P}_{h}(X^{(u)}_iX^{(v)}_j+X^{(v)}_iX^{(u)}_j) \mathcal{J}_{h,k}\right\|\\
      &\leq & \sum_{t=1}^{s} \sum_{k=1}^{m-t+1} \frac{1}{\ell} \sum_{i=(k+t-1)\ell+1}^{(k+t)\ell} \delta^{(u)}_{4,i-h} \mathcal{J}_{h,k}\left\|\sum_{j=i-t\ell+1}^{i-(t-1)\ell-1} K\left(\frac{i-j}{\ell}\right) X^{(v)}_{j-h}\right\|_4  \\
      && \qquad + \sum_{t=1}^{s} \sum_{k=1}^{m-t+1} \frac{1}{\ell} \sum_{i=(k+t-1)\ell+1}^{(k+t)\ell} {\delta^{(v)}_{4,i-h} \mathcal{J}_{h,k}\left\|\sum_{j=i-t\ell+1}^{i-(t-1)\ell-1} K\left(\frac{i-j}{\ell}\right) X^{(u)}_{j-h}\right\|_4} \\
      && \qquad + \sum_{t=1}^{s} \sum_{k=1}^{m-t+1} \frac{1}{\ell} \sum_{i=(k-1)\ell+1}^{(k+1)\ell-1} \delta^{(u)}_{4,j-h}\mathcal{J}_{h,k}\left\|\sum_{i=j+(t-1)\ell+1}^{j+t\ell-1} K\left(\frac{i-j}{\ell}\right) X^{(v)}_{i-h}\right\|_4\\
      && \qquad + \sum_{t=1}^{s} \sum_{k=1}^{m-t+1} \frac{1}{\ell} \sum_{i=(k-1)\ell+1}^{(k+1)\ell-1} \delta^{(v)}_{4,j-h}\mathcal{J}_{h,k}\left\|\sum_{i=j+(t-1)\ell+1}^{j+t\ell-1} K\left(\frac{i-j}{\ell}\right) X^{(u)}_{i-h}\right\|_4,
  \end{eqnarray*}
  by Minkowski and H\"{o}lder's inequalities.
  Using Lemma 1 of \cite{wu2010}, we have
  \begin{eqnarray}\label{eqt:lemmaderiveMV}
   &&\sum_{t=1}^{s} \sum_{k=1}^{m-t+1} \frac{1}{\ell} \sum_{i=(k+t-1)\ell+1}^{(k+t)\ell} \delta^{(u)}_{4,i-h}\mathcal{J}_{h,k}\left\|\sum_{j=i-t\ell+1}^{i-(t-1)\ell-1} K\left(\frac{i-j}{\ell}\right) X^{(v)}_{j-h}\right\|_4 \nonumber\\
   &\leq& \frac{C\Delta^{(v)}_4}{\ell} \sum_{t=1}^{s} \sum_{k=1}^{m-t+1} \sum_{i=(k+t-1)\ell+1}^{(k+t)\ell} \delta^{(u)}_{4,i-h}\mathcal{J}_{h,k} \left\{ \sum_{(t-1)\ell+1}^{t\ell-1} K^2(u/\ell)\right\}^{1/2} 
  \end{eqnarray}
  for some constant $C$ and similarly for the other three terms.
  Using (\ref{eqt:lemmaderiveMV}) and the fact that $\Delta^{(u)}_4<\infty$ for all $u\in\{1,\ldots,d\}$, for large enough $\ell$,
  \begin{eqnarray*}
    \left\| \sum_{t=1}^{s} \sum_{k=1}^{m-t+1}\left[\E\left\{G^{(u,v)}_{k,t}|\mathcal{F}_{(k-1)\ell+1}\right\}-\E\left\{ G^{(u,v)}_{k,t}\right\} \right]\right\|^2
      &=& \sum_{h=-\infty}^{(m-1)\ell+1} \left\|\mathcal{P}_h \left(\sum_{t=1}^{s} \sum_{k=1}^{m-t+1} G^{(u,v)}_{k,t} \mathcal{J}_{h,k}\right)\right\|^2\\
      &\leq& o\left\{  m \int_0^{s} K^2(t) \dd t\right\}.
  \end{eqnarray*}
  And the proof is completed.
\end{enumerate}
\end{proof}

\begin{lemma}[Multivariate version of Lemma~\ref{lemma:hoc2}]\label{lemma:multblock2}
Under conditions of Theorem~\ref{thm:BiasVar}. As $\ell\to\infty$,
{
  \begin{align*}
    I^{(u,v, u',v')}_1 
    	&= \E \left\{\sum_{t=1}^{s} \sum_{k=1}^{m-t+1} \left[ G^{(u,v)}_{k,t} - \E\left\{G^{(u,v)}_{k,t} \mid \mathcal{F}_{\varsigma_{k,t}}\right\}\right]\right\} 
			\left\{\sum_{t=1}^{s} \sum_{k=1}^{m-t+1} \left[ G^{(u',v')}_{k,t} - \E\left\{G^{(u',v')}_{k,t} \mid \mathcal{F}_{\varsigma_{k,t}}\right\}\right]\right\}	\\
	&\sim \frac{ 2\left\{ V^{(u,u')} V^{(v,v')} + V^{(u,v')} V^{(v,u')} \right\} 
			\int_{0}^s K^2(t) \dd t n}{\ell}.
  \end{align*}
  }
\end{lemma}
\begin{proof}[Proof of Lemma~\ref{lemma:multblock2}]
The proof can be done by replacing Lemma 3 in the original proof by Lemma~\ref{lemma:multblock}.
\end{proof}

{

\begin{lemma}[Spectral density version of Lemma 3 in \citet{chanyau2015_hoc}]\label{lemma:spec1}
  Suppose that $X_i\in\mathcal{L}^\nu$ and $\Delta_4<\infty$ for some $\nu>4$.
  Also let $\E(X_1)=0$. Suppose $K(\cdot)\in\mathcal{K}_p$.
  Let $\omega\in(0, \pi)$. 
  Let $\mathbb{B}_1(t)$ and $\mathbb{B}_2(t)$ be two independent standard Brownian motions.
  Define $G^{\textnormal{Spec}}_{k,t}$ as in (\ref{def_MnsStar}).
  Define additionally, for $k,t=1,2,\ldots$, that
  $$
    \mathbb{G}^{\textnormal{Spec}}_{k,t} = 2\pi f(\omega)\int_{k+t-1}^{k+t}\int_{s-t}^{s-t+1}K(s-h) \left\{\dd\mathbb{B}_1(h) \dd\mathbb{B}_1(s) - \dd\mathbb{B}_2(h) \dd\mathbb{B}_2(s)\right\}.
  $$
  Then as $\ell\to\infty$, the following results hold. 

  (i) $\left(G^{\textnormal{Spec}}_{1,1},G^{\textnormal{Spec}}_{1,2},G^{\textnormal{Spec}}_{2,1},G^{\textnormal{Spec}}_{2,2}\right)\Rightarrow\left(\mathbb{G}^{\textnormal{Spec}}_{1,1},\mathbb{G}^{\textnormal{Spec}}_{1,2},\mathbb{G}^{\textnormal{Spec}}_{2,1},\mathbb{G}^{\textnormal{Spec}}_{2,2}\right)$. 

  (ii) For $k,t=1,2$,
    \begin{align*}
      \E\left(G^{\textnormal{Spec}}_{k,t}\right) &\to 0, \qquad 
      \Var\left(G^{\textnormal{Spec}}_{k,t}\right) \to 8\pi^2 f(\omega)^2\int_{t-1}^{t}K^2(u) \,\dd u , \\
     \Cov\left(G^{\textnormal{Spec}}_{1,1}, G^{\textnormal{Spec}}_{1,2}\right)&\to 0,\qquad  \Cov\left(G^{\textnormal{Spec}}_{1,1}, G^{\textnormal{Spec}}_{2,2}\right)\to 0, \qquad \Cov\left(G^{\textnormal{Spec}}_{1,1}, G^{\textnormal{Spec}}_{2,1}\right)\to 0, \\
     \Cov\left(G^{\textnormal{Spec}}_{1,2}, G^{\textnormal{Spec}}_{2,2}\right)&\to 0,\qquad \Cov\left(G^{\textnormal{Spec}}_{2,1}, G^{\textnormal{Spec}}_{2,2}\right)\to 0 .
    \end{align*} 

  (iii) For $m\to\infty$,
    \[
          \left\| \sum_{t=1}^{s} \sum_{k=1}^{m-t+1}\left[\E\left\{G^{\textnormal{Spec}}_{k,t}|\mathcal{F}_{(k-1)\ell+1}\right\}-\E G^{\textnormal{Spec}}_{k,t}\right]\right\| = o\left[\left\{  m \int_0^{s} K^2(t) \dd t\right\}^{1/2}\right].
    \]
\end{lemma}

\begin{proof}[Proof of Lemma~\ref{lemma:spec1}]\leavevmode
(i)  Let $S_{\ell}(t) = \sum_{i=1}^{\lfloor{\ell t}\rfloor} X_i e^{\iota \omega i}/\sqrt{\ell}$ for $t\in[0,1]$.
    Also let 
    \[
      g(\omega) = \lim_{\ell\to\infty} \frac{1}{n}\E\left|\sum_{i=1}^{n} X_i e^{\iota \omega i}\right|^2 = 2\pi f(\omega).
    \]
    {By Theorem 2 of \cite{wu2005Fourier} and Cram\'{e}r--Wold device, we have the convergence of finite-dimensional distributions.
    The tightness of $\{S_{\ell}(t) :t\in[0,1]\}$ follows from Lemma 4.3 of \cite{peligrad2010central}.
    Thus, it follows that, as $\ell\to\infty$,}
    \[
     \left\{ 
     \begin{bmatrix}
     \text{Re}(S_\ell(t)) \\
     \text{Im}(S_\ell(t))
     \end{bmatrix}:
     t \in [0,1]
     \right\}
     \Rightarrow 
     \left\{ 
     \sqrt{\frac{g(\omega)}{2}}
     \begin{bmatrix}
     \text{Re}(\mathbb{B}_1(t)) \\
     \text{Im}(\mathbb{B}_2(t))
     \end{bmatrix}:
     t \in [0,1]
     \right\}.
    \]
    By the Cram\'{e}r--Wold device and continuous mapping theorem,
    similar to the proof of Lemma~\ref{lemma:multblock} (i),
    we have the first result.

(ii) For the second result, 
    using Lemma~\ref{lemma:hoc1}, the linearity of expectations, and independence of $\mathbb{B}_1$ and $\mathbb{B}_2$,
    we have, for $k,t=1,2$, that 
    \begin{align*}
      \E\left(\mathbb{G}^{\textnormal{Spec}}_{k,t}\right) &= 
      2\pi f(\omega) \left\{ \E\left(\mathbb{G}_{k,t}/v \right) - \E\left(\mathbb{G}_{k,t}/v \right)\right\} 
	= 0 ,\\
      \Var \left(\mathbb{G}^{\textnormal{Spec}}_{k,t}\right) &= 2\{2\pi f(\omega)\}^2 \Var \left(\mathbb{G}_{k,t}/ v\right) = 8\pi^2 f(\omega)^2 \int_{t-1}^t K^2(u) \dd u,
    \end{align*}    
    where $\mathbb{G}_{k,t}$ is defined in (\ref{eqt:definition_Gkt}).
    Similarly, the covariances 
    $\Cov\left(\mathbb{G}^{\textnormal{Spec}}_{1,1}, \mathbb{G}^{\textnormal{Spec}}_{1,2}\right)$,
    $\Cov\left(\mathbb{G}^{\textnormal{Spec}}_{1,1}, \mathbb{G}^{\textnormal{Spec}}_{1,2}\right)$,
    $\Cov\left(\mathbb{G}^{\textnormal{Spec}}_{1,1}, \mathbb{G}^{\textnormal{Spec}}_{1,2}\right)$,
    $\Cov\left(\mathbb{G}^{\textnormal{Spec}}_{1,1}, \mathbb{G}^{\textnormal{Spec}}_{1,2}\right)$
    are zero.
    By the same arguments of the convergence of moments,
    we have the second result.

    For (iii), recall that $\mathcal{P}_h \cdot = \E(\cdot \mid \mathcal{F}_h) - \E(\cdot \mid \mathcal{F}_{h-1})$ is the projection operator. 
    Again, by the orthogonality of martingale differences,
    \begin{eqnarray*}
      \left\| \sum_{t=1}^{s} \sum_{k=1}^{m-t+1}\left[\E\left\{G^{\textnormal{Spec}}_{k,t}|\mathcal{F}_{(k-1)\ell+1}\right\}-\E G^{\textnormal{Spec}}_{k,t}\right]\right\|^2
        &=& \left\| \sum_{t=1}^{s} \sum_{k=1}^{m-t+1} \sum_{h=-\infty}^{(k-1)\ell+1} \mathcal{P}_h G^{\textnormal{Spec}}_{k,t}\right\|^2\\
        &=& \left\| \sum_{h=-\infty}^{(m-1)\ell+1} \mathcal{P}_h \left(\sum_{t=1}^{s} \sum_{k=1}^{m-t+1} G^{\textnormal{Spec}}_{k,t} \mathbb{1} \{h\leq (k-1)\ell + 1\}\right) \right\|^2\\
        &=& \sum_{h=-\infty}^{(m-1)\ell+1} \left\|\mathcal{P}_h \left(\sum_{t=1}^{s} \sum_{k=1}^{m-t+1} G^{\textnormal{Spec}}_{k,t} \mathcal{J}_{h,k}\right)\right\|^2,
    \end{eqnarray*}
  where $\mathcal{J}_{h,k}= \mathbb{1} \{h\leq (k-1)\ell + 1\}$.
  By the definition of $\mathcal{P}_h$,
  $$
     \mathcal{P}_{0}(X_{i-h}X_{j-h}) = \E\left( X_{i-h}X_{j-h} - {X}_{i-h,\{0\}}{X}_{j-h,\{0\}} \mid \mathcal{F}_0\right).
  $$
  Using similar arguments as (45) and (51) in \cite{chanyau2015_hoc}, we have 
    \begin{eqnarray*}
    \left\|\mathcal{P}_h \left(\sum_{t=1}^{s} \sum_{k=1}^{m-t+1} G^{\textnormal{Spec}}_{k,t} \mathcal{J}_{h,k}\right) \right\|
      &=& \left\| \sum_{t=1}^{s} \sum_{k=1}^{m-t+1} \frac{1}{\ell}\sum_{i=(k+t-1)\ell+1}^{(k+t)\ell} \sum_{j=i-t\ell+1}^{i-(t-1)\ell-1} K\left(\frac{i-j}{\ell}\right) \cdot \right.\\
      &&\ \cdot \left.\mathcal{P}_{h}(e^{\iota \omega i}X_ie^{-\iota \omega j}X_j +  e^{\iota \omega j}X_je^{-\iota \omega i}X_i) \mathcal{J}_{h,k}\right\|\\
      &\leq & \sum_{t=1}^{s} \sum_{k=1}^{m-t+1} \frac{1}{\ell} \sum_{i=(k+t-1)\ell+1}^{(k+t)\ell} \delta_{4,i-h} \mathcal{J}_{h,k}\left\|\sum_{j=i-t\ell+1}^{i-(t-1)\ell-1} K\left(\frac{i-j}{\ell}\right) e^{\iota(i-j)\omega} X_{j-h}\right\|_4  
      \\
      && \ + \sum_{t=1}^{s} \sum_{k=1}^{m-t+1} \frac{1}{\ell} \sum_{i=(k+t-1)\ell+1}^{(k+t)\ell} \delta_{4,i-h} \mathcal{J}_{h,k}\left\|\sum_{j=i-t\ell+1}^{i-(t-1)\ell-1} K\left(\frac{i-j}{\ell}\right) e^{\iota(j-i)\omega} X_{j-h}\right\|_4  \\
      && \ + \sum_{t=1}^{s} \sum_{k=1}^{m-t+1} \frac{1}{\ell} \sum_{i=(k-1)\ell+1}^{(k+1)\ell-1} \delta_{4,j-h}\mathcal{J}_{h,k}\left\|\sum_{i=j+(t-1)\ell+1}^{j+t\ell-1} K\left(\frac{i-j}{\ell}\right) e^{\iota(i-j)\omega}X_{i-h}\right\|_4 \\
      && \ + \sum_{t=1}^{s} \sum_{k=1}^{m-t+1} \frac{1}{\ell} \sum_{i=(k-1)\ell+1}^{(k+1)\ell-1} \delta_{4,j-h}\mathcal{J}_{h,k}\left\|\sum_{i=j+(t-1)\ell+1}^{j+t\ell-1} K\left(\frac{i-j}{\ell}\right) e^{\iota(j-i)\omega}X_{i-h}\right\|_4
  \end{eqnarray*}
  by Minkowski and H\"{o}lder's inequalities.
  Using Lemma 1 of \cite{wu2010}, we have
  \begin{eqnarray}\label{eqt:lemmaderive_spec}
   &&\sum_{t=1}^{s} \sum_{k=1}^{m-t+1} \frac{1}{\ell} \sum_{i=(k+t-1)\ell+1}^{(k+t)\ell} \delta_{4,i-h}\mathcal{J}_{h,k}\left\|\sum_{j=i-t\ell+1}^{i-(t-1)\ell-1} K\left(\frac{i-j}{\ell}\right) e^{\iota(j-i)\omega} X_{j-h}\right\|_4 \nonumber\\
   &\leq& \frac{C\Delta_4}{\ell} \sum_{t=1}^{s} \sum_{k=1}^{m-t+1} \sum_{i=(k+t-1)\ell+1}^{(k+t)\ell} \delta_{4,i-h}\mathcal{J}_{h,k} \left\{ \sum_{(t-1)\ell+1}^{t\ell-1} K^2(u/\ell)\right\}^{1/2} 
  \end{eqnarray}
  as $|e^{-\iota u \omega}| < 1$
  for some constant $C$ and similarly for the other term.
  The rest of the proofs is identical to the proof of Theorem~\ref{lemma:hoc1}.
  \end{proof}

  \begin{lemma}[Spectral density version of Lemma~\ref{lemma:hoc2}]\label{lemma:spec2}
  Under conditions of Theorem~\ref{thm:spectralBV}. 
  As $\ell\to\infty$,
  \[
    I_1^{\textnormal{Spec}} = \left\| \sum_{t=1}^{s} \sum_{k=1}^{m-t+1} \left[ G^{\textnormal{Spec}}_{k,t} - \E\left\{G^{\textnormal{Spec}}_{k,t} \mid \mathcal{F}_{\varsigma_{k,t}}\right\}\right]\right\| \sim \left\{ \frac{8\pi^2 f(\omega)^2 n \int_{0}^s K^2(t) \dd t  }{\ell} \right\}^{1/2}.
  \]
\end{lemma}

\begin{proof}[Proof of Lemma~\ref{lemma:spec2}]
The proof can be done by replacing Lemma 3 in the original proof by Lemma~\ref{lemma:spec1}.
\end{proof}
}

  \bibliographystyle{rss}
  {\small
  \setstretch{1.0}
  \bibliography{myRef}
  }

\vspace{-2cm}
\end{document}